\documentclass{llncs}

\usepackage{times}
\usepackage{amsmath}
\usepackage{amsfonts}
\usepackage{amssymb}
\usepackage[dvipsnames,svgnames,x11names]{xcolor}
\usepackage{array}
\usepackage{empheq}
\usepackage{latexsym}
\usepackage{multicol}
\usepackage{paralist}
\usepackage{relsize}
\usepackage{stmaryrd}
\usepackage{tabularx}
\usepackage[]{units}
\usepackage{url}
\usepackage{verbatim}
\usepackage{wasysym}
\allowdisplaybreaks

\newtheorem{fatto}{Fact}
\newtheorem{asz}{Simplifying Assumption}

\renewenvironment{proof}{\begin{trivlist} \item[\hspace{\labelsep}\bf Proof:]}{\hfill $\Box$\end{trivlist}}
\newenvironment{assumption}{\begin{asz} \rm  }{\end{asz}}

\newenvironment{apx-proof}[1] 
        {\noindent \textbf{Proof of #1.}} 
        {\qed}

\usepackage{thmtools}
\usepackage{thm-restate}

\newcommand{\trans}[1][]{\xrightarrow{\, {#1} \, }}
\newcommand{\ntrans}[1][]{\mathrel{{\trans[#1]}\makebox[0em][r]{$\not$\hspace{2ex}}}{\!}}

\newcommand{\SOSrule}[2]{\frac{\displaystyle #1}{\displaystyle #2}}

\newcommand{\Act}{\mathcal A}

\newcommand{\Var}{\mathcal{V}}

\newcommand{\DVar}{\mathcal{V}_{\mathrm{d}}}

\newcommand{\rel}{{\mathcal R}}

\newcommand{\init}[1]{\mathrm{init}(#1)}

\newcommand{\trt}[2]{\triangleleft_{#1}^{#2}}

\newcommand{\leftmerge}{\mathbin{
                  \setlength{\unitlength}{1ex}
                  \begin{picture}(1,1.75)
                  \put(0,0){\line(1,0){1}}
                  \put(0,0){\line(0,1){1.75}}
                  \put(0.45,0){\line(0,1){1.75}}
                  \end{picture}
                 }}
\newcommand{\rightmerge}{\mathbin{
                  \setlength{\unitlength}{1ex}
                  \begin{picture}(1,1.75)
                  \put(0,0){\line(1,0){1}}
                  \put(1,0){\line(0,1){1.75}}
                  \put(0.55,0){\line(0,1){1.75}}
                  \end{picture}
                 }}
\newcommand{\cmerge}{~|~}
\newcommand{\rx}{{x}}
\newcommand{\ry}{{y}}

\newcommand{\chiudi}{\operatorname{cl}}
\newcommand{\bfmid}{~\mathbf{|}~}
\renewcommand{\em}{\sl}

\newcommand{\preox}{\approx}

\newcommand{\depth}{{\mathit{depth}}}
\newcommand{\norm}{{\mathit{norm}}}
\newcommand{\var}{\mathit{var}}

\newcommand{\mv}[1]{\xrightarrow{\, {#1} \, }}

\newcommand{\nv}[1]{\mathrel{\stackrel{#1}{\nrightarrow}}}

\newcommand{\nil}{\mathbf{0}}
\newcommand{\hmerge}{\mathrel{|^{\negmedspace\scriptstyle /}}}

\newcommand{\Rule}[2]{
\begin{displaymath}
\begin{array}{c}
#1 \\\hline 
#2
\end{array}
\end{displaymath}
}

\newcommand{\RuleLab}[3]{
\begin{equation}
\label{#3}
\begin{array}{c}
#1 \\\hline 
#2
\end{array}
\end{equation}
}

\newcommand{\ThreeRules}[6]{
\begin{displaymath}
\begin{array}{c}
#1 \\\hline 
#2
\end{array}
\qquad
\begin{array}{c}
#3 \\\hline 
#4
\end{array}
\qquad
\begin{array}{c}
#5 \\\hline 
#6
\end{array}
\end{displaymath}
}

\newcommand{\bistext}{{\underline{\leftrightarrow}}}
\newcommand{\bis}{~{\underline{\leftrightarrow}}~}
\newcommand{\nbis}{~{\underline{\leftrightarrow}}\!\!\!\!/~~}
\newcommand{\FCCS}{\text{CCS}_{\scalebox{0.6}{$f$}}}

\newcommand{\ff}{\scalebox{0.6}{$f$}}

\newcommand{\bb}{\mathrm{b}}

\newcommand{\dd}{\mathrm{d}}

\newcommand{\E}{\mathcal E}

\newcommand{\sx}{\mathrm{l}}

\newcommand{\rr}{\mathrm{r}}

\newcommand{\ww}{\mathrm{w}}

\begin{document}

\title{Are Two Binary Operators Necessary to Obtain a Finite Axiomatisation of Parallel Composition?
\thanks{A preliminary version of this work appeared as~\cite{ACFIL21}.}}

\author{Luca Aceto\inst{1,2}
\and
Valentina Castiglioni\inst{1}
\and
Wan Fokkink\inst{3}
\and
Anna Ing{\'o}lfsd{\'o}ttir\inst{1}
\and
Bas Luttik\inst{4}
}

\institute{
Reykjavik University, Reykjavik, Iceland
\and
Gran Sasso Science Institute, L'Aquila, Italy
\and
Vrije Universiteit Amsterdam, Amsterdam, The Netherlands
\and
Eindhoven University of Technology, Eindhoven, The Netherlands
}

\maketitle

\begin{abstract}
Bergstra and Klop have shown that \emph{bisimilarity} has a \emph{finite} equational axiomatisation over ACP/CCS extended with the binary \emph{left} and \emph{communication merge} operators. 
Moller proved that auxiliary operators are \emph{necessary} to obtain a finite axiomatisation of bisimilarity over CCS, and Aceto et al. showed that this remains true when \emph{Hennessy's merge} is added to that language. 
These results raise the question of whether there is \emph{one} auxiliary \emph{binary} operator whose addition to CCS leads to a finite axiomatisation of bisimilarity. 
We contribute to answering this question in the simplified setting of the recursion-, relabelling-, and restriction-free fragment of CCS. 
We formulate three natural assumptions pertaining to the operational semantics of auxiliary operators and their relationship to parallel composition, and prove that an auxiliary binary operator facilitating a finite axiomatisation of bisimilarity in the simplified setting cannot satisfy all three assumptions.

\medskip 


\end{abstract}


\section{Introduction}

The purpose of this paper is to provide an answer to the following problem (see~\cite[Problem 8]{Aceto2003}):
\emph{Are the left merge and the communication merge operators} 
necessary 
\emph{to obtain a finite equational axiomatisation of bisimilarity over the language CCS?}
The interest in this problem is threefold, as an answer to it would:
\begin{enumerate}
\item provide the first study on the finite axiomatisability of operators whose operational semantics is not determined a priori,
\item clarify the status of the auxiliary operators \emph{left merge} and \emph{communication merge}, proposed in~\cite{BK84b}, in the finite axiomatisation of parallel composition, and
\item give further insight into properties that auxiliary operators used in the finite equational characterisation of parallel composition ought to afford.
\end{enumerate}
We prove that, under some simplifying assumptions, whose role in our technical developments we discuss below, there is no auxiliary binary operator that can be added to CCS to yield a finite equational axiomatisation of bisimilarity. 
Despite falling short of solving the above-mentioned problem in full generality, our negative result is a substantial generalisation of previous non-finite axiomatisability theorems by Moller~\cite{Mo89,Mo90} and Aceto et al.~\cite{AFIL05}.

In order to put our contribution in context, we first describe the history of the problem we tackle and then give a bird's eye view of our results.


\subsubsection*{The story so far}

In the late 1970s, Milner developed the \emph{Calculus of Communicating Systems} (CCS)~\cite{Mi80}, a formal language based on a message-passing paradigm and aimed at describing communicating processes from an operational point of view.
In detail, a \emph{labelled transition system} (LTS)~\cite{Ke76} was used to equip language expressions with an \emph{operational semantics}~\cite{Pl81} and was defined using a collection of syntax-driven rules.
The analysis of process behaviour was carried out via an observational \emph{bisimulation}-based theory~\cite{Pa81} that defines when two states in an LTS describe the same behaviour.
In particular, CCS included a \emph{parallel composition operator} $\|$ to model the interactions among processes.
Such an operator, also known as \emph{merge}~\cite{BK84b,BK85}, allows one both to \emph{interleave} the behaviours of its argument processes (modelling concurrent computations) and to enable some form of \emph{synchronisation} between them (modelling interactions).
Later on, in collaboration with Hennessy, Milner studied the \emph{equational  theory} of (recursion-free) CCS and proposed a \emph{ground-complete axiomatisation} for it modulo bisimilarity~\cite{HM85}. 
More precisely, Hennessy and Milner presented a set $\E$ of \emph{equational axioms} from which all equations over closed CCS terms (namely those with no occurrences of variables) that are \emph{valid modulo bisimilarity} can be derived using the rules of \emph{equational logic}~\cite{T77}.
Notably, the set $\E$ included infinitely many axioms, which were instances of the \emph{expansion law} that was used to `simulate equationally' the operational semantics of the parallel composition operator.

The ground-completeness result by Hennessy and Milner started the quest for a finite axiomatisation of CCS's parallel composition operator modulo bisimilarity.

Bergstra and Klop showed in~\cite{BK84b} that a finite ground-complete axiomatisation modulo bisimilarity can be obtained by enriching CCS with two auxiliary operators, namely the \emph{left merge} $\leftmerge$ and the \emph{communication merge} $\!\!\!\cmerge\!\!$, expressing respectively one step in the asymmetric pure interleaving and the synchronous behaviour of $\|$.
Their result was then strengthened by Aceto et al.\ in~\cite{AFIL09}, where it is proved that, over the fragment of CCS without recursion, restriction and relabelling, the auxiliary operators $\leftmerge$ and $\!\!\cmerge\!\!$ allow for finitely axiomatising $\|$ modulo bisimilarity also when CCS terms with variables are considered.
Moreover, in~\cite{AILT08} that result is extended to the fragment of CCS with relabelling and restriction, but without communication.
From those studies, we can infer that the left merge and communication merge operators are \emph{sufficient} to finitely axiomatise parallel composition modulo bisimilarity. 
But is the addition of auxiliary operators \emph{necessary} to obtain a finite equational axiomatisation, or can the use of the expansion law in the original axiomatisation of bisimilarity by Hennessy and Milner be replaced by a finite set of sound CCS equations?

To address that question, in~\cite{Mo89,Mo90} Moller considered a minimal fragment of CCS, including only action prefixing, nondeterministic choice and interleaving, and proved that, even in the presence of a single action, bisimilarity does not afford a finite ground-complete axiomatisation over the closed terms in that language. This showed that auxiliary operators are indeed necessary to obtain a finite equational axiomatisation of bisimilarity.
Adapting Moller's proof technique, Aceto et al.\ proved, in~\cite{AFIL05}, that if we replace $\leftmerge$ and $\!\!\cmerge\!$ with the so called \emph{Hennessy's merge} $\hmerge$~\cite{He88}, which denotes an asymmetric interleaving with communication, then the collection of equations that hold modulo bisimilarity over the recursion-, restriction-, and relabelling-free fragment of CCS enriched with $\hmerge$ is not finitely based (in the presence of at least two distinct complementary actions).

A natural question that arises from those \emph{negative} results is the following:
\begin{equation}
\tag{P}\label{eq:problem}
\parbox{\dimexpr\linewidth-4em}{
\strut
\emph{Can one obtain a finite axiomatisation of the parallel composition operator in bisimulation semantics by adding} 
only one binary operator
\emph{to the signature of (recursion-, restriction-, and relabelling-free) CCS?}
\strut
}
\end{equation}

In this paper, we provide a partial \emph{negative answer} to that question.
\emph{Henceforth, we consider the recursion-, restriction-, and relabelling-free fragment of CCS which, for simplicity, we still denote as CCS.}
(Note that, in (\ref{eq:problem}), we focus on binary operators, like all the variations on parallel composition mentioned above, since using a ternary operator one can express the left and communication merge operators and, in fact, an arbitrary number of binary operators.)


\subsubsection*{Our contribution}

We analyse the axiomatisability of parallel composition over the language $\FCCS$, namely CCS enriched with a binary operator $f$ that we use to express $\|$ as a derived operator.
We prove that, under three simplifying assumptions, an auxiliary operator $f$ alone does not allow us to obtain a finite ground-complete axiomatisation of $\FCCS$ modulo bisimilarity.
We remark that the non-existence of a finite \emph{ground-complete} axiomatisation implies the non-existence of a finite complete one.
Hence, we actually provide a partial answer to a stronger version of \eqref{eq:problem} concerning the existence of a finite ground-complete axiomatisation over CCS.

To this end, the only knowledge we assume on the operational semantics of $f$ is that it is formally defined by rules in the de Simone format~\cite{dS85} (Assumption~\ref{Ass:deSimone}) and that the behaviour of the parallel composition operator is expressed equationally by a law that is akin to the one used by Bergstra and Klop to define $\|$ in terms of $\,\leftmerge$ and $\!\!\cmerge\!$ (Assumption~\ref{Ass:equation}). 
We then argue that the latter assumption yields that the equation
\begin{equation}
\label{eq:intro}
x \| y \approx f(x,y) + f(y,x) 
\tag{A}
\end{equation}
is valid modulo bisimilarity.
Next we proceed by a case analysis over the possible sets of de Simone rules defining the behaviour of $f$, in such a way that the validity of Equation~\eqref{eq:intro} modulo bisimilarity is guaranteed.
To fully characterise the sets of rules that may define $f$, we introduce a third simplifying assumption: the target of each rule for $f$ is either a variable or a term obtained by applying a single $\FCCS$ operator to the variables of the rule, according to the constraints of the de Simone format (Assumption~\ref{assumption:targets}).
Then, for each of the resulting cases, we show the desired negative result using proof-theoretic techniques that have their roots in Moller's classic results in~\cite{Mo89,Mo90}.
This means that we identify a (case-specific) property of terms denoted by $W_n$ for $n \ge 0$.
The idea is that, when $n$ is \emph{large enough}, $W_n$ is preserved by provability from finite, sound axiom systems.
Hence, whenever $\E$ is a finite, sound axiom system and an equation $p \approx q$ is derivable from $\E$, then either both terms $p$ and $q$ satisfy $W_n$, or none of them does.
The negative result is then obtained by exhibiting a (case-specific) infinite family of valid equations $\{e_n \mid n \ge 0\}$ in which $W_n$ is not preserved, that is, for each $n \ge 0$, $W_n$ is satisfied only by one side of $e_n$.
Due to the choice of $W_n$, this means that the equations in the family cannot all be derived from a finite set of valid axioms and therefore no finite, sound axiom system can be complete.

To the best of our knowledge, in this paper we propose the first non-finite axiomatisability result for a process algebra in which one of the operators, namely the auxiliary operator $f$, does not have a fixed semantics. 
However, for our technical developments, it has been necessary to restrict the search space for $f$ by means of the aforementioned simplifying assumptions.
We proceed to give some justifications for these assumptions.
There are three main reasons behind Assumption~\ref{Ass:deSimone}:
\begin{itemize}
\item The de Simone format is the simplest congruence format for bisimilarity.
Hence we must be able to deal with this case before proceeding to any generalisation.
\item The specification of parallel composition, left merge and communication merge operators (and of the vast majority of process algebraic operators) is in de Simone format.
Hence, that format was a natural choice also for operator $f$.
\item The simplicity of the de Simone rules allows us to reduce considerably the complexity of our case analysis over the sets of available rules for the operator $f$.
However, as witnessed by the developments in this article, even with this simplification, the proof of the desired negative result requires a large amount of delicate, technical work.
\end{itemize}
Assumptions~\ref{Ass:equation} and~\ref{assumption:targets} still allow us to obtain a significant generalisation of related works, such as~\cite{AFIL05}, as we can see them as an attempt to identify the requirements needed to apply Moller's proof technique to Hennessy's merge like operators.
We stress that the reason for adding Assumption~\ref{assumption:targets} is purely technical: it plays a role in the proof of \emph{one} of the claims in our combinatorial analysis of the rules that $f$ may have (see Lemma~\ref{Lem:asyncrules}).
Although we conjecture that the assumption is not actually necessary to obtain that claim, we were unable to prove it without the assumption.

Even though the vast literature on process algebras offers a plethora of non-finite axiomatisability results for a variety of languages and semantics (see, for instance, the survey~\cite{AFIL05a} from 2005), we are not aware of any previous attempt at proving a result akin to the one we present here. 
We have already addressed at length how our contribution fits within the study of the equational logic of processes and how it generalises previous results in that field. 
The proof-theoretic tools and the approach we adopt in proving our main theorem, which links equational logic with structural operational
semantics and builds on a number of previous achievements (such as those in~\cite{ABV94}), may have independent interest for researchers in logic in computer science. 
To our mind, achieving an answer to question (\ref{eq:problem}) in full generality would be very pleasing for the concurrency-theory community, as it would finally clarify the canonical role of Bergstra and Klop's auxiliary operators in the finite axiomatisation of parallel composition modulo bisimilarity. 


\subsubsection*{Organisation of contents}

In Section~\ref{sec:background} we review basic notions on process semantics, CCS and equational logic.
In Section~\ref{Sect:Problem} we present the simplifying assumptions under which we tackle the problem~\eqref{eq:problem}.
In Section~\ref{sec:rules} we study the operational semantics of auxiliary operators $f$ meeting our assumptions.
In Section~\ref{sec:proof_method} we give a detailed presentation of the proof strategy we will follow to address~\eqref{eq:problem}.
Sections~\ref{sec:axioms} and~\ref{sec:prelim} are then devoted to the technical development of our negative results, which are presented in Sections~\ref{sec:Labat}--\ref{sec:Lt}.
We conclude by discussing future work in Section~\ref{sec:conclusion}.


\section{Background}
\label{sec:background}


\subsubsection*{LTSs and bisimilarity}
\label{Sect:lts+bis}

As semantic model we consider classic \emph{labelled transition systems}~\cite{Ke76}.

\begin{definition}
\label{Def:lts}
A {\sl labelled transition system} (LTS) is a triple $(S,\Act,\trans[])$, where $S$ is a set of {\sl states} (or \emph{processes}), $\Act$ is a set of {\sl actions}, and ${\trans[]} \subseteq S \times \Act \times S$ is a ({\sl labelled}) {\sl transition relation}. 
\end{definition}

As usual, we use $p \trans[\mu] p'$ in lieu of $(p,\mu,p') \in {\trans[]}$. 
For each $p \in S$ and $\mu \in \Act$, we write $p \trans[\mu]$ if $p \trans[\mu] p'$ holds for some $p'$, and $p \ntrans[\mu]$ otherwise. 
The \emph{initials} of $p$ are the actions that label the outgoing transitions of $p$, that is, $\init{p} = \{\mu \mid p \trans[\mu] \}$. 
For a sequence of actions $\varphi = \mu_1 \cdots \mu_k$ ($k\geq 0$), and states $p,p'$, we write $p \trans[\varphi] p'$ if and only if there exists a sequence of transitions $p = p_0 \trans[\mu_1] p_1 \trans[\mu_2] \cdots \trans[\mu_k] p_k = p'$. 
If $p \trans[\varphi] p'$ holds for some state $p'$, then $\varphi$ is a {\em trace} of $p$. 
Moreover, we say that $\varphi$ is a maximal trace of $p$ if $\init{p'} = \emptyset$.
By means of traces, we associate two classic notions with a process $p$: its \emph{depth}, denoted by $\depth(p)$, and its \emph{norm}, denoted by $\norm(p)$.
For a process $p$ whose set of traces is finite, they express, respectively, the length of a \emph{longest} trace and that of a \emph{shortest} maximal trace of $p$.
Formally, $\depth(p) = \sup \{ k \mid p \text{ has a trace of length } k\}$ and $\norm(p) = \inf \{ k \mid p \text{ has a maximal trace of length } k\}$.

In this paper, we shall consider the states in a labelled transition system modulo bisimilarity~\cite{Mi89,Pa81}, allowing us to establish whether two processes have the same behaviour.

\begin{definition}
\label{Def:bisimulation}
Let $(S,\Act,\mv{})$ be a labelled transition system. 
{\sl Bisimilarity}, denoted by $\bistext$, is the largest binary symmetric relation over $S$ such that whenever $p \bis q$ and $p \mv{\mu} p'$, then there is a transition $q \mv{\mu} q'$ with $p' \bis q'$.
If $p \bis q$, then we say that $p$ and $q$ are {\sl bisimilar}. 
\end{definition}

It is well-known that bisimilarity is an equivalence relation (see, e.g.,~\cite{Mi89,Pa81}). 
Moreover, two bisimilar processes have the same depth and norm.


\subsubsection*{The Language $\FCCS$}
\label{Sect:HCCS}

The language we consider in this paper is obtained by adding a single binary operator $f$ to the recursion-, restriction-, and relabelling-free subset of Milner's CCS~\cite{Mi89}, henceforth referred to as $\FCCS$, and is given by the following grammar:
\[
t ::=\; \nil \,\bfmid\, x \,\bfmid\, a.t \,\bfmid\, \bar{a}.t \,\bfmid\, \tau. t \,\bfmid\, t+t \,\bfmid\, t \mathbin{\|} t
\,\bfmid\, f(t,t) \enspace ,
\]
where $x$ is a variable drawn from a countably infinite set $\Var$, $a$ is an action, and $\bar{a}$ is its complement. 
We assume that the actions $a$ and $\bar{a}$ are distinct. 
Following~\cite{Mi89}, the action symbol $\tau$ will result from the synchronised occurrence of the complementary actions $a$ and $\bar{a}$. 

To obtain the desired negative results, it will actually be sufficient to consider the proposed three unary prefixing operators; so there is only one action $a$ with its corresponding complementary action $\bar{a}$, so that $\Act = \{ a, \bar{a}, \tau \}$.
Our results carry over unchanged to a setting with an arbitrary number of actions, and corresponding prefixing operators.
Henceforth, we let $\mu\in \Act$ and $\alpha\in\{a,\bar{a}\}$. 
As usual, we postulate that $\bar{\bar{a}}=a$.  
We shall use the meta-variables $t,u,v,w$ to range over process terms, and write ${\it var}(t)$ for the collection of variables occurring in the term $t$.  
The {\sl size} of a term is the number of operator symbols in it. 
A term is {\em closed} if it does not contain any variables.  
Closed terms, or {\sl processes}, will be denoted by $p,q,r$. 
Moreover, trailing {\bf 0}'s will be omitted from terms.

A {\sl (closed) substitution} is a mapping from process variables to (closed) $\FCCS$ terms. 
For every term $t$ and substitution $\sigma$, the term obtained by replacing every occurrence of a variable $x$ in $t$ with the term $\sigma(x)$ will be written $\sigma(t)$. 
Note that $\sigma(t)$ is closed, if so is $\sigma$. 
We let $\sigma[x\mapsto p]$ denote the substitution that maps the variable $x$ into process $p$ and behaves like $\sigma$ on all other variables.

In the remainder of this paper, we exploit the associativity and commutativity of $+$ modulo bisimilarity and we consider process terms modulo them, namely we do not distinguish $t+u$ and $u+t$, nor $(t+u)+v$ and $t+(u+v)$.  
In what follows, the symbol $=$ will denote equality modulo the above identifications. 
We use a {\em summation} $\sum_{i\in\{1,\ldots,k\}}t_i$ to denote the term $t= t_1+\cdots+t_k$, where the empty sum represents {\bf 0}.
We can also assume that the terms $t_i$, for $i \in \{1,\dots,k\}$, do not have $+$ as head operator, and refer to them as the \emph{summands} of $t$.

Henceforth, for each action $\mu$ and $m \ge 0$, we let $\mu^0$ denote {\bf 0} and $\mu^{m+1}$ denote $\mu(\mu^m)$.  
For each action $\mu$ and positive integer $i \ge 0$, we also define
\[
\mu^{\scriptstyle \leq i} = \mu + \mu^2 + \cdots + \mu^i \enspace .
\]


\subsubsection*{Equational Logic}
\label{Sect:logic}

An \emph{axiom system} $\E$ is a collection of (\emph{process}) \emph{equations} $t \approx u$ over $\FCCS$.
An equation $t \approx u$ is \emph{derivable} from an axiom system $\E$, notation $\E \vdash t \approx u$, if there is an \emph{equational proof} for it from $\E$, namely if $t \approx u$ can be inferred from the axioms in $\E$ using the \emph{rules} of \emph{equational logic}.
The ones over $\FCCS$ are reported in Table~\ref{tab:equational_logic}. 
Rules ($e_1$)-($e_4$) are common for all process languages and they ensure that $\E$ is closed with respect to reflexivity, symmetry, transitivity and substitution, respectively.
Rules ($e_5$)-($e_8$) are tailored for $\FCCS$ and they ensure the closure of $\E$ under $\FCCS$ contexts.
They are therefore referred to as the \emph{congruence rules}.

\begin{table}[t]
\begin{gather*}
\qquad\qquad
\scalebox{0.9}{($e_1$)}\; t \preox t 
\qquad
\scalebox{0.9}{($e_2$)}\; \frac{t \preox u}{u \preox t} 
\qquad
\scalebox{0.9}{($e_3$)}\; \frac{{t \preox u ~~ u \preox v}}{{t \preox v}} 
\qquad
\scalebox{0.9}{($e_4$)}\; \frac{{t \preox u}}{{\sigma(t) \preox \sigma(u)}} 
\\[.2cm]
\scalebox{0.9}{($e_5$)}\; \frac{t \preox u}{\mu. t \preox \mu. u}
\qquad
\scalebox{0.9}{($e_6$)}\; \frac{t \preox  u~~ t' \approx u'}{t+t' \preox u+u'}
\qquad
\scalebox{0.9}{($e_7$)}\; \frac{t \preox  u~~ t' \approx u'}{f(t,t') \preox f(u,u')}
\qquad
\scalebox{0.9}{($e_8$)}\; \frac{t \preox  u~~ t' \approx u'}{t\mathbin{\|} t' \preox u\mathbin{\|} u'}
\enspace .
\end{gather*}
\caption{\label{tab:equational_logic} The rules of equational logic} 
\end{table}

Without loss of generality one may assume that substitutions happen first in equational proofs, i.e., that rule ($e_4$) may only be used when $(t \preox u) \in \E$.  
In this case $\sigma(t) \preox \sigma(u)$ is called a {\em substitution instance} of an axiom in $\E$.
Moreover, by postulating that for each axiom in $\E$ also its symmetric counterpart is present in $\E$, one may assume that applications of symmetry happen first in equational proofs, i.e., that rule ($e_2$) is never used in equational proofs.  
In the remainder of the paper, we shall always tacitly assume that axiom systems are closed with respect to symmetry.

We are interested in equations that are valid modulo some congruence relation $\rel$ over closed terms.
The equation $t \approx u$ is said to be \emph{sound} modulo $\rel$ if $\sigma(t) \,\rel\, \sigma(u)$ for all closed substitutions $\sigma$.
For simplicity, if $t \approx u$ is sound, then we write $t \,\rel\, u$.
An axiom system is \emph{sound} modulo $\rel$ if, and only if, all of its equations are sound modulo $\rel$. 
Conversely, we say that $\E$ is \emph{ground-complete} modulo $\rel$ if $p \,\rel\, q$ implies $\E \vdash p \approx q$ for all closed terms $p,q$.
We say that $\rel$ has a \emph{finite}, ground-complete, axiomatisation, if there is a \emph{finite} axiom system $\E$ that is sound and ground-complete for $\rel$.


\section{The simplifying assumptions}
\label{Sect:Problem}

The aim of this paper is to investigate whether bisimilarity admits a finite equational axiomatisation over $\FCCS$, for some binary operator $f$. 
Of course, this question only makes sense if $f$ is an operator that preserves bisimilarity. 
In this section we discuss two assumptions we shall make on the auxiliary operator $f$ in order to meet such requirement and to tackle problem (\ref{eq:problem}) in a simplified technical setting.


\subsection{The de Simone format}

One way to guarantee that $f$ preserves bisimilarity is to postulate that the behaviour of $f$ is described using Plotkin-style rules that fit a rule format that is known to preserve bisimilarity, see, e.g.,~\cite{AcetoFV2001} for a survey of such rule formats. 
The simplest format satisfying this criterion is the format proposed by de Simone in~\cite{dS85}. 
We believe that if we can't deal with operations specified in that format, then there is little hope to generalise our results. 
Therefore, we make the following

\begin{assumption}
\label{Ass:deSimone}
The behaviour of $f$ is described by rules in de Simone format. 
\end{assumption}

\begin{definition}
\label{def:de_Simone}
An SOS rule $\rho$ for $f$ is in \emph{de Simone format} if it has the form
\begin{equation}
\label{deSimone}
\rho = \SOSrule{\{ \rx_i \mv{\mu_i} \ry_i \mid i \in I\}}{f(\rx_1,\rx_2) \mv{\mu} t}
\end{equation}
where $I \subseteq \{1,2\}$, $\mu,\mu_i \in \Act$ ($i\in I$), and moreover
\begin{itemize}
\item the variables $x_1$, $x_2$ and $y_i$ ($i\in I$) are all dif{f}erent and are called the \emph{variables of the rule},
\item $t$ is a $\FCCS$ term over variables $\{\rx_1,\rx_2,\ry_i \mid i\in I\}$, called the {\sl target of the rule}, such that
\begin{itemize}
\item each variable occurs at most once in $t$, and
\item if $i\in I$, then $\rx_i$ does not occur in $t$. 
\end{itemize}
\end{itemize}
\end{definition}

Henceforth, we shall assume, without loss of generality, that the variables $x_1$, $x_2$, $y_1$ and $y_2$ are the only ones used in operational rules for $f$.
Moreover, if $\mu$ is the label of the transition in the conclusion of a de Simone rule $\rho$, we shall say that $\rho$ has $\mu$ as {\sl label}. 

The SOS rules for all of the classic CCS operators, reported below, are in de Simone format, and so are those for Hennessy's $\hmerge$ operator from~\cite{He88} and for Bergstra and Klop's left and communication merge operators~\cite{BK84}, at least if we disregard issues related to the treatment of successful termination. 
Thus restricting ourselves to operators whose operational behaviour is described by de Simone rules leaves us with a good degree of generality. 
\begin{align*}
& \SOSrule{}{\mu. x \mv{\mu} x}
\qquad
\SOSrule{x\mv{\mu} x'}{x+y\mv{\mu} x'}
\qquad
\SOSrule{y\mv{\mu} y'}{x+y\mv{\mu} y'}
\\[.2cm]
& \SOSrule{x\mv{\mu} x'}{x \mathbin{\|}y\mv{\mu} x'\mathbin{\|}y}
\qquad
\SOSrule{y\mv{\mu} y'}{x \mathbin{\|}y\mv{\mu} x\mathbin{\|}y'}
\qquad
\SOSrule{x\mv{\alpha} x',~y\mv{\bar{\alpha}} y'}{x \mathbin{\|}y\mv{\tau} x'\mathbin{\|}y'}
\end{align*}

The transition rules for the classic CCS operators above and those for the operator $f$ give rise to transitions between $\FCCS$ terms.  
The operational semantics for $\FCCS$ is thus given by the LTS whose states are $\FCCS$ terms, and whose transitions are those that are provable using the rules.

In what follows, we shall consider the collection of {\sl closed $\FCCS$ terms} modulo bisimilarity.
Since the SOS rules defining the operational semantics of $\FCCS$ are in de Simone's format, we have that bisimilarity is a congruence with respect to $\FCCS$ operators, that is, $\mu p \bis \mu q$, $p + p' \bis q + q'$, $p \| p' \bis q \| q'$ and $f(p,p') \bis f(q,q')$ hold whenever $p \bis q$, $p' \bis q'$, for processes $p,p',q,q'$.

Bisimilarity is extended to arbitrary $\FCCS$ terms thus:

\begin{definition}
\label{Def:sound}
Let $t,u$ be $\FCCS$ terms. 
We write $t\bis u$ if and only if $\sigma(t) \bis \sigma(u)$ for every closed substitution $\sigma$.
\end{definition}


\subsection{Axiomatising $\|$ with $f$}

Our second simplifying assumption concerns how the operator $f$ can be used to axiomatise parallel composition.
To this end, a fairly natural assumption on an axiom system over $\FCCS$ is that it includes an equation of the form 
\begin{equation}
\label{eq:general}
x \| y \approx t(x,y)
\end{equation}
where $t$ is a $\FCCS$ term that does not contain occurrences of $\|$ with $\var(t) \subseteq \{x,y\}$.
More precisely, the term will be in the general form $t(x,y) = \sum_{i \in I} t_i(x,y)$, where $I$ is a finite index set and, for each $i \in I$, $t_i(x,y)$ does not have $+$ as head operator.
Equation (\ref{eq:general}) essentially states that $\|$ is a derived operator in $\FCCS$ modulo bisimilarity. 
To our mind, this is a natural, initial assumption to make in studying the problem we tackle in the paper.

We now proceed to refine the form of the term $t(x,y)$, in order to guarantee the soundness, modulo bisimilarity, of Equation~\eqref{eq:general}.
Intuitively, no term $t_i(x,y)$ can have prefixing as head operator.
In fact, if $t(x,y)$ had a summand $\mu.t'(x,y)$, for some $\mu \in \Act$, then one could easily show that $\nil \| \nil \nbis  t(\nil,\nil)$, since $t(\nil,\nil)$ could perform a $\mu$-transition, unlike $\nil\|\nil$.
Similarly, $t(x,y)$ cannot have a variable as a summand, for otherwise we would have $a \| \tau \nbis  t(a,\tau)$.
Indeed, assume, without loss of generality, that $t(x,y)$ has a summand $x$.
Then, $t(a,\tau) \trans[a] \nil$, 
whereas $a \| \tau$ cannot terminate in one step.
We can therefore assume that, for each $i \in I$, $t_i(x,y) = f(t_i^1(x,y), t_i^2(x,y))$ for some $\FCCS$ terms $t_i^k(x,y)$, with $k \in \{1,2\}$.
To further narrow down the options on the form that the subterms $t_i^k(x,y)$ might have, we would need to make some assumptions on the behaviour of the operator $f$.
For the sake of generality, we assume that the terms $t_i^k(x,y)$ are in the simplest form, namely they are variables in $\{x,y\}$.
Notice that to allow prefixing and/or nested occurrences of $f$-terms in the scope of the terms $t_i(x,y)$ we would need to define (at least partially) the operational semantics of $f$, thus making our results less general as, roughly speaking, we would need to study one possible auxiliary operator at a time (the one identified by the considered set of de Simone rules).
Moreover, if we look at how parallel composition is expressed equationally as a derived operator in terms of Hennessy's merge (H), or Bergstra and Klop's left and communication merge (LC), or as in~\cite{ABV94} (LRC), viz.~via the equations
\begin{center}
(H) $x \mathbin{\|}y \approx (x \hmerge y) + (y \hmerge x)$ 
\\[.1cm]
(LC) $x \mathbin{\|}y \approx (x \leftmerge y) + (y \leftmerge x) + (x \cmerge y)$
\qquad
(LRC) $x \mathbin{\|}y \approx (x \leftmerge y) + (x\rightmerge y) + (x \cmerge y)\enspace ,$  
\end{center}
we see the emergence of a pattern: the parallel composition operator is always expressed in terms of sums of terms built from the auxiliary operators and variables. 

Therefore, from now on we will make the following:

\begin{assumption}
\label{Ass:equation}
For some $J \subseteq \{x,y\}^2$, the equation 
\begin{equation}
\label{Eqn:parf}
x \mathbin{\|}y \approx \sum \{ f(z_1,z_2) \mid (z_1,z_2)\in J \}
\end{equation}
holds modulo bisimilarity. 
We shall use $t_J$ to denote the right-hand side of the above equation and use $t_J(p,q)$ to stand for the process $\sigma[x\mapsto p, y \mapsto q](t_J)$, for any closed substitution $\sigma$.
\end{assumption}

Using our assumptions, we further investigate the relation between operator $f$ and parallel composition, obtaining a refined form for Equation~\eqref{Eqn:parf} (Proposition~\ref{assumption:f_vs_par} below).

\begin{lemma}
\label{lem:basics}
Assume that Assumptions~\ref{Ass:deSimone} and~\ref{Ass:equation} hold.
Then:
\begin{enumerate}
\item \label{Jnonempty}
The index set $J$ on the right-hand side of Equation (\ref{Eqn:parf}) is non-empty.
\item \label{frulesnonempty} 
The set of transition rules for $f$ is non-empty.
\item \label{positive-premises} 
Each transition rule for $f$ has some premise.
\item \label{lem:no-eq-vars}
The terms $f(x,x)$ and $f(y,y)$ are not summands of $t_J$.
\end{enumerate}
\end{lemma}

\begin{proof}
Statements~\ref{Jnonempty} and~\ref{frulesnonempty} are trivial because the equation
$
x  \| y \approx \nil 
$
is not sound modulo bisimilarity.  

Let us focus now on the proof for statement~\ref{positive-premises}. 
To this end, assume, towards a contradiction, that $f$ has a rule of the form
\[
f(x_1,x_2) \mv{\mu} t(x_1,x_2) \enspace , 
\]
for some action $\mu$ and term $t$. 
This rule can be used to derive that $f(\nil,\nil) \mv{\mu} t(\nil,\nil)$.
Since the set $J$ on the right-hand side of Equation (\ref{Eqn:parf}) is non-empty by statement~\ref{Jnonempty}, the term $f(\nil,\nil)$ occurs as a summand of $t_J(\nil,\nil)$. 
It follows that $t_J(\nil,\nil) \mv{\mu} t(\nil,\nil)$.
Therefore, 
\[
\nil \mathbin{\|} \nil \bis \nil \nbis t_J(\nil,\nil) \enspace ,
\]
contradicting our Assumption~\ref{Ass:equation}. 

Finally, we deal with statement~\ref{lem:no-eq-vars}.
Assume, towards a contradiction, that $f(x,x)$ is a summand of $t_J$. 
As $a\mathbin{\|} \nil \mv{a} \nil \mathbin{\|} \nil \bis \nil $ and Equation (\ref{Eqn:parf}) holds modulo bisimilarity, there is a process $p$ such that
\[
t_J(a,\nil) \mv{a} p ~\text{and}~p\bis \nil 
\enspace .
\]
This means that there is a summand $f(z_1,z_2)$ of  $t_J$ such that $f(p_1,p_2) \mv{a} p$, where, for $i\in\{1,2\}$,  
\[
p_i = 
\begin{cases}
a & \text{if $z_i = x$}\enspace ,\\
\nil & \text{if $z_i = y$}\enspace . 
\end{cases} 
\]
The transition $f(p_1,p_2) \mv{a} p$ must be provable using some de Simone rule $\rho$ for $f$ (see Equation~\eqref{deSimone} in Definition~\ref{def:de_Simone}). 
Such a rule has some premise by Lemma~\ref{lem:basics}(\ref{positive-premises}), and each such premise must have the form $x_1 \mv{\mu} y_1$ or $x_2  \mv{\mu} y_2$, for some action $\mu$. 
If both $z_1$ and $z_2$ are $y$ then $p_1=p_2=\nil$, and none of those premises can be met.
Therefore at least one of $z_1$ and $z_2$ in the summand $f(z_1,z_2)$ is $x$.  
Moreover, if $x_i \mv{\mu} y_i$ ($i\in\{1,2\}$) is a premise of $\rho$, then $z_i=x$ and $\mu=a$ (or else the premise could not be met).
So the rule $\rho$ can have one of the following three forms:
\ThreeRules{x_1 \mv{a} y_1}{f(x_1,x_2) \mv{a} t_1(y_1,x_2)} 
{x_2\mv{a} y_2}{f(x_1,x_2) \mv{a} t_2(x_1,y_2)} 
{x_1 \mv{a} y_1 \quad x_2 \mv{a} y_2}{f(x_1,x_2) \mv{a} t_3(y_1,y_2)} 
for some terms $t_1$, $t_2$ and $t_3$.  
We now proceed to argue that the existence of each of these rules contradicts the soundness of Equation~\eqref{Eqn:parf} modulo bisimilarity.
  
If $\rho$ has the form 
\Rule{x_1 \mv{a} y_1 \quad x_2 \mv{a} y_2}{f(x_1,x_2) \mv{a} t_3(y_1,y_2)}
then $z_1=z_2=x$ and $f(a,a) \mv{a} p$.
Since the term $f(a,a)$ is a summand of $t_J(a,a)$, it follows that  $t_J(a,a) \mv{a} p$ also holds. 
However, this contradicts the soundness of Equation (\ref{Eqn:parf}) because, for each transition $a\mathbin{\|} a \mv{a} q$, we have that $q \bis a \nbis \nil \bis p$.

Assume now, without loss of generality, that $\rho$ has the form 
\Rule{x_1 \mv{a} y_1}{f(x_1,x_2) \mv{a} t_1(y_1,x_2)} 
Using this rule, we can infer that $f(a,a) \mv{a} t_1(\nil,a)$.
Since $f(x,x)$ is a summand of $t_J$ by our assumption, the term $f(a,a)$ is a summand of $t_J(a,\nil)$. 
Hence, $t_J(a,\nil) \mv{a}  t_1(\nil,a)$ also holds. 
As Equation (\ref{Eqn:parf}) holds modulo bisimilarity, we have that 
\[
a \mathbin{\|} \nil \bis t_J(a,\nil) \enspace .
\]
Therefore $t_1(\nil,a) \bis \nil$, because $a \mathbin{\|} \nil \mv{a} \nil \mathbin{\|} \nil$ is the only transition afforded by the term $a \mathbin{\|} \nil$.
Observe now that
\[
t_J(a,a) \mv{a}  t_1(\nil,a)  \bis \nil \enspace . 
\]
also holds. 
However, this contradicts the soundness of Equation (\ref{Eqn:parf}) as above because, for each transition $a\mathbin{\|} a \mv{a} q$, we have that $q \bis a \nbis \nil \bis t_1(\nil,a)$.
  
This proves that $f(x,x)$ is not a summand of $t_J$, which was to be shown.
\end{proof}

As a consequence, we may infer that the index set $J$ in the term $t_J$ is either one of the singletons $\{(x,y)\}$ or $\{(y,x)\}$, or it is the set $\{(x,y),(y,x)\}$.
Due to Moller's results to the effect that bisimilarity has no finite ground-complete axiomatisation over CCS~\cite{Mo89,Mo90a}, the former option can be discarded, as shown in the following:

\begin{proposition}
\label{Propn:singleton-nonfin}
If $J$ is a singleton, then $\FCCS$ admits no finite equational axiomatisation modulo bisimilarity.
\end{proposition}

\begin{proof}
If $J$ is a singleton, then, since $\mathbin{\|}$ is commutative modulo bisimilarity, the equation 
\[
x\mathbin{\|} y \approx f(x,y) 
\]
holds modulo bisimilarity. 
Therefore the result follows from the nonexistence of a finite equational axiomatisation for CCS proven by Moller in~\cite{Mo89,Mo90a}.
\end{proof}

As a consequence, we can restate our Assumption~\ref{Ass:equation} in the following simplified form:

\begin{proposition}
\label{assumption:f_vs_par}
Equation~\eqref{Eqn:parf} can be refined to the form:
\begin{equation}
\label{Eqn:parf2}
x\mathbin{\|} y \approx f(x,y) + f(y,x) \enspace .
\end{equation}
\end{proposition}

Moreover, in the light of Moller's results in~\cite{Mo89,Mo90a}, we can restrict ourselves to considering only operators $f$ such that $x\mathbin{\|} y \approx f(x,y)$ does not hold modulo bisimilarity.

For later use, we note a useful consequence of the soundness of Equation~\eqref{Eqn:parf2} modulo bisimilarity.

\begin{lemma}
\label{Lem:finite-depth}
Assume that Equation (\ref{Eqn:parf2}) holds modulo $\bistext$. 
Then $\depth(p)$ is finite for each closed $\FCCS$ term $p$.
\end{lemma}

\begin{proof}
By structural induction on closed terms. 
For all of the standard CCS operators, it is well known that the depth of closed terms can be characterized inductively thus:
\[
\begin{array}{rcl}
\depth(\nil) & = & 0 \\
\depth(\mu p) & = & 1+\depth(p) \\
\depth(p + q) & = & \max \{\depth(p), \depth(q)\} \\
\depth(p \| q) & = & \depth(p)+ \depth(q) \enspace .
\end{array}
\]
So the depth of a closed term of the form $\mu p$, $p + q$ or $p \| q$ is finite, if so are the depths of $p$ and $q$. 

Consider now a closed term of the form $f(p,q)$. 
Since bisimilar terms have the same depth and, by the proviso of the lemma, Equation~\eqref{Eqn:parf2} holds modulo bisimilarity, we have that
\[
\depth(f(p,q)) \leq \depth(f(p,q)+f(q,p)) = \depth(p \| q) \enspace . 
\]
It follows that $\depth(f(p,q))$ is finite, if so are the depths of $p$ and $q$.
\end{proof}


\section{The operational semantics of $f$}
\label{sec:rules}

In order to obtain the desired results, we shall, first of all, understand what rules $f$ may and must have in order for Equation (\ref{Eqn:parf2}) to hold modulo bisimilarity (Proposition~\ref{prop:new} below).
We begin this analysis by restricting the possible forms the SOS rules for $f$ \emph{may} take.

\begin{lemma}
\label{Lem:frules}
Suppose that $f$ meets Assumption~\ref{Ass:deSimone}, and that Equation~\eqref{Eqn:parf2} is sound modulo bisimilarity. 
Let $\rho$ be a de Simone rule for $f$ with $\mu$ as label.
Then:
\begin{enumerate}
\item \label{tauact} 
If $\mu= \tau$ then the set of premises $\{ \rx_i \mv{\mu_i} \ry_i \mid i \in I\}$ of $\rho$ can only have one of the following possible forms:
\begin{itemize}
\item $\{\rx_i \mv{\tau} \ry_i\}$ for some $i\in \{1,2\}$, or
\item $\{\rx_1 \mv{\alpha} \ry_1, \rx_2 \mv{\bar{\alpha}} \ry_2\}$ for some $\alpha\in \{a,\bar{a}\}$.
\end{itemize}
\item \label{alphaact} 
If $\mu= \alpha$ for some $\alpha\in \{a,\bar{a}\}$, then the set of premises $\{ \rx_i \mv{\mu_i} \ry_i \mid i \in I\}$ can only have the form $\{\rx_i \mv{\alpha} \ry_i\}$ for some $i\in \{1,2\}$.
\end{enumerate}
\end{lemma}

\begin{proof}
We only detail the proof for statement~\ref{tauact}. 
The proof for statement~\ref{alphaact} follows similar lines.  

Assume, towards a contradiction, that $\mu= \tau$ and the set of premises $\{ \rx_i \mv{\mu_i} \ry_i \mid i \in I\}$ of $\rho$ has some form that differs from those in the statement. 
Then the set of premises of $\rho$ has one of the following two forms:
\begin{itemize}
\item $\{\rx_i \mv{\alpha} \ry_i\}$ for some $i\in \{1,2\}$ and $\alpha\in \{a,\bar{a}\}$, or
\item $\{\rx_1 \mv{\mu_1} \ry_1, \rx_2 \mv{\mu_2} \ry_2\}$ for some $\mu_1,\mu_2\in \Act$ such that 
\begin{itemize}
\item either $\mu_1=\tau$ or $\mu_2=\tau$, or
\item $\mu_1 = \mu_2 =\alpha$ for some $\alpha\in\{a,\bar{a}\}$.
\end{itemize}
\end{itemize}
We now proceed to argue that the existence of either of these rules for $f$ contradicts the soundness of Equation~\eqref{Eqn:parf2}.  
\begin{itemize}
\item Assume that the set of premises of $\rho$ has the form $\{\rx_i \mv{\alpha} \ry_i\}$ for some $i\in \{1,2\}$ and $\alpha\in \{a,\bar{a}\}$.  
In this case, we can use that rule to prove the existence of the transition
$
f(\alpha,\nil) \mv{\tau} t(\nil,\nil)$, or of $f(\nil,\alpha) \mv{\tau} t(\nil,\nil), 
$
depending on whether $i=1$ or $i=2$. 
Therefore 
$
f(\alpha,\nil)+ f(\nil,\alpha) \mv{\tau} t(\nil,\nil) 
$   
also holds. 
However, the existence of this transition immediately contradicts the soundness of Equation~\eqref{Eqn:parf2} modulo bisimilarity because $\alpha \mathbin{\|} \nil$ affords no $\tau$-transition.
\item Assume that the set of premises of $\rho$ has the form $\{\rx_1 \mv{\mu_1} \ry_1, \rx_2 \mv{\mu_2} \ry_2\}$ for some $\mu_1,\mu_2\in \Act$ such that
\begin{itemize}
\item either $\mu_1=\tau$ or $\mu_2=\tau$, or
\item $\mu_1 = \mu_2 =\alpha$ for some $\alpha\in\{a,\bar{a}\}$.
\end{itemize}
In this case, we can use that rule to prove the existence of the transition $f(\mu_1,\mu_2) \mv{\tau} t(\nil,\nil)$.
Therefore $f(\mu_1,\mu_2)+ f(\mu_2,\mu_1) \mv{\tau} t(\nil,\nil)$ also holds.  
By the soundness of Equation~\eqref{Eqn:parf2}, we have that 
\[
\mu_1 \| \mu_2 \bis f(\mu_1,\mu_2)+ f(\mu_2,\mu_1) \enspace .
\]
Hence $\mu_1 \| \mu_2 \mv{\tau} p$ for some $p$ such that $p\bis t(\nil,\nil)$.  
If $\mu_1 = \mu_2 =\alpha$ for some $\alpha\in\{a,\bar{a}\}$, then the above transition cannot exist, because $\alpha \mathbin{\|} \alpha$ affords no $\tau$-transition.
This immediately contradicts the soundness of Equation~\eqref{Eqn:parf2} modulo bisimilarity. 
We therefore proceed with the proof by assuming that at least one of $\mu_1$ and $\mu_2$ is $\tau$. 
In this case, we have that $\mu_1 \mathbin{\|} \mu_2 \mv{\tau} p$ implies that $p \bis \mu_1$ and $\mu_2=\tau$, or $p\bis \mu_2$ and $\mu_1=\tau$.
Assume, without loss of generality, that $\mu_1=\tau$ and 
\begin{equation}
\label{Eqn:mu2-bis-t}
t(\nil,\nil) \bis \mu_2 \enspace .
\end{equation}
Pick now an action $\alpha\neq \mu_2$. 
(Such an action exists as we have three actions in our language.) 
The soundness of Equation~\eqref{Eqn:parf2} yields that 
\[
\tau \mathbin{\|} (\mu_2 + \alpha) \bis f(\tau,\mu_2 + \alpha)+ f(\mu_2 + \alpha,\tau) \enspace .
\]
Using the rule for $f$ we assumed we had and the rules for $+$, we can prove the existence of the transition $f(\tau,\mu_2 + \alpha)+ f(\mu_2 + \alpha,\tau) \mv{\tau} t(\nil,\nil)$.
Since the source of the above transition is bisimilar to $\tau \mathbin{\|} (\mu_2 + \alpha)$, there must be a term $p$ such that $\tau \mathbin{\|} (\mu_2 + \alpha) \mv{\tau} p$ and $p \bis t(\nil,\nil)$. 
We can distinguish two cases, according to the semantics of $\mathbin{\|}$:
\begin{itemize}
\item $\tau \trans[\tau] \nil$, so that $\tau \mathbin{\|} (\mu_2+\alpha) \trans[\tau] \nil \mathbin{\|} (\mu_2+\alpha)$.
Since $\alpha \neq \mu_2$, we have that $\mu_2+\alpha \nbis \mu_2$ and thus $p = \nil \mathbin{\|} (\mu_2+\alpha) \bis \mu_2+\alpha \nbis t(\nil,\nil)$.
Hence, we need to discard this case.
\item $\mu_2 = \tau$, so that $\mu_2+\alpha \trans[\tau] \nil$ and $t \mathbin{\|} (\mu_2+\alpha) \trans[\tau] \tau \mathbin{\|} \nil$.
Hence, $p = \tau \mathbin{\|} = \mu_2 \mathbin{\|} \nil \bis \mu_2 \bis t(\nil,\nil)$, and $p$ is the term we are looking for.
\end{itemize}
We have therefore obtained that $\mu_1=\mu_2=\tau$. 
    
We are now ready to reach the promised contradiction to the soundness of Equation~\eqref{Eqn:parf2}. 
In fact, consider the term $ f(\tau + a,\tau + a)$. 
Using the rule for $f$ we assumed we had, we can again prove the existence of the transition $f(\tau + a,\tau + a) \mv{\tau} t(\nil,\nil)$.
By Equation~\eqref{Eqn:mu2-bis-t} and our observation that $\mu_2=\tau$, the term $t(\nil,\nil)$ is bisimilar to $\tau$.  
On the other hand, $(\tau + a) \mathbin{\|} (\tau + a) \mv{\tau} p$ implies that $p \bis (\tau + a) \nbis \tau$, contradicting the soundness of Equation~\eqref{Eqn:parf2} modulo bisimilarity.
\end{itemize}
\end{proof}

The previous lemma limits the form of the premises that rules for $f$ \emph{may} have in order for Equation (\ref{Eqn:parf2}) to hold modulo bisimilarity. 
We now characterise the rules that $f$ \emph{must} have in order for it to satisfy that equation.
Firstly, we deal with {\sl synchronisation}.

\begin{lemma}
\label{Lem:syncrules}
Assume that Equation (\ref{Eqn:parf2}) holds modulo bisimilarity.  
Then the operator $f$ must have a rule of the form
\RuleLab{\rx_1 \mv{\alpha} \ry_1 \quad \rx_2 \mv{\bar{\alpha}}y_2}{f(x_1,x_2) \mv{\tau} t(y_1,y_2)}{syncrule}
for some $\alpha\in \{a,\bar{a}\}$ and term $t$.  
Moreover, for each rule for $f$ of the above form the term $t(x,y)$ is bisimilar to $x\mathbin{\|}y$.
\end{lemma}

\begin{proof}
We first argue that $f$ must have a rule of the form (\ref{syncrule}) for some $\alpha\in \{a,\bar{a}\}$ and term $t$. 
To this end, assume, towards a contradiction, that $f$ has no such rule. 
Observe that the term $a \mathbin{\|} \bar{a}$ affords the transition $a \mathbin{\|} \bar{a} \mv{\tau} \nil \mathbin{\|} \nil$.
However, neither the term $f(a,\bar{a})$ nor the term $f(\bar{a}, a)$ affords a $\tau$-transition. 
In fact, using our assumption that $f$ has no rule of the form (\ref{syncrule}) and Lemma~\ref{Lem:frules}(\ref{tauact}), each rule for $f$ with a $\tau$-transition as a consequent must have the form \Rule{\rx_i \mv{\tau} \ry_i}{f(x_1,x_2) \mv{\tau} t} for some $i\in\{1,2\}$ and term $t$. 
Such a rule cannot be used to infer a transition from $f(a,\bar{a})$ or $f(\bar{a}, a)$. 
It follows that
\[
a \mathbin{\|} \bar{a} \nbis f(a,\bar{a}) + f(\bar{a}, a) \enspace ,
\]
contradicting the soundness of Equation~\eqref{Eqn:parf2}.
Therefore $f$ must have a rule of the form (\ref{syncrule}).
  
We now proceed to argue that $t(x,y)$ is bisimilar to $x\mathbin{\|}y$, for each rule of the form (\ref{syncrule}) for $f$. 
Pick a rule for $f$ of the form (\ref{syncrule}). 
We shall argue that
\[
p \mathbin{\|} q  \bis t(p,q) \enspace ,
\]
for all closed $\FCCS$ terms $p$ and $q$.
To this end, consider  the terms $\alpha.p \mathbin{\|} \bar{\alpha}.q$ and $f(\alpha.p,\bar{\alpha}.q)+f(\bar{\alpha}.q,\alpha.p)$. 
Using rule (\ref{syncrule}) and the rules for $+$, we have that $f(\alpha.p,\bar{\alpha}.q)+f(\bar{\alpha}.q,\alpha.p) \mv{\tau} t(p,q)$.
By the soundness of Equation~\eqref{Eqn:parf2}, we have that 
\[
\alpha.p \mathbin{\|} \bar{\alpha}.q \bis f(\alpha.p,\bar{\alpha}.q)+f(\bar{\alpha}.q,\alpha.p) \enspace . 
\]
Therefore there is a closed term $r$ such that $\alpha.p \mathbin{\|} \bar{\alpha}.q \mv{\tau} r$ and $r\bis t(p,q)$. 
Note that the only $\tau$-transition by $\alpha.p \mathbin{\|} \bar{\alpha}.q$ is $\alpha.p \mathbin{\|} \bar{\alpha}.q \mv{\tau} p \mathbin{\|} q$.
Hence $r = p \mathbin{\|} q  \bis t(p,q)$, which was to be shown. 
\end{proof}

Henceforth we assume, without loss of generality that the target of a rule of the form (\ref{syncrule}) is $y_1 \| y_2$.
We introduce the unary predicates $S_{a,\bar{a}}^f$ and $S_{\bar{a},a}^f$ to identify which rules of type (\ref{syncrule}) are available for $f$.
In detail, $S_{a,\bar{a}}^f$ holds if $f$ has a rule of type (\ref{syncrule}) with premises $x_1 \trans[a] y_1$ and $x_2 \trans[\bar{a}] y_2$.
$S_{\bar{a},a}^f$ holds in the symmetric case.

We consider now the {\sl interleaving} behaviour in the rules for $f$.
In order to properly characterise the rules for $f$ as done in the previous Lemma~\ref{Lem:syncrules}, we consider an additional simplifying assumption on the form that the targets of the rules for $f$ might have.

\begin{assumption}
\label{assumption:targets}
If $t$ is the target of a rule for $f$, then $t$ is either a variable or a term obtained by applying a single $\FCCS$ operator to the variables of the rule, according to the constraints of the de Simone format.
\end{assumption}

\begin{remark}
We remark that Assumption~\ref{assumption:targets} is used \emph{only} in the proof of the following lemma.
Although we conjecture that the assumption is not necessary to obtain this result, we were unable to prove it without the assumption.
\end{remark}

\begin{lemma}
\label{Lem:asyncrules}
Let $\mu\in\Act$. 
Then the operator $f$ must have a rule of the form 
\RuleLab{\rx_1 \mv{\mu} \ry_1}{f(x_1,x_2) \mv{\mu} t(y_1,x_2)}{asyncruleleft} 
or a rule of the form 
\RuleLab{\rx_2 \mv{\mu}\ry_2} {f(x_1,x_2) \mv{\mu} t(x_1,y_2)}{asyncruleright} 
for some term $t$.
Moreover, under Assumption~\ref{assumption:targets}, for each rule for $f$ of the above forms the term $t(x,y)$ is bisimilar to $x\mathbin{\|}y$.
\end{lemma}

\begin{proof}
Let $\mu\in\Act$. 
We first argue that $f$ must have a rule of the form (\ref{asyncruleleft}) or (\ref{asyncruleright}) for some term $t$. 
To this end, assume, towards a contradiction, that $f$ has no such rules. 
Observe that the term $\mu \mathbin{\|} \nil$ affords the transition $\mu \mathbin{\|} \nil \mv{\mu} \nil \mathbin{\|} \nil$.
However, neither the term $f(\mu,\nil)$ nor the term $f(\nil, \mu)$ affords a $\mu$-transition. 
In fact, if $f$ has no rule of the form (\ref{asyncruleleft}) or (\ref{asyncruleright}), Lemma~\ref{Lem:frules} yields that 
\begin{itemize}
\item either $f$ has no rule with a $\mu$-transition as a consequent, 
\item or $\mu=\tau$, and each rule for $f$ with a $\tau$-transition as a consequent has the form 
\[
\SOSrule{x_1 \trans[\alpha] y_1 \quad x_2 \trans[\bar{\alpha}] y_2}{f(x_1,x_2) \trans[\tau] t(y_1,y_2)}
\quad
\text{ for some } \alpha\in\{a,\bar{a}\}.
\]
\end{itemize}
In the latter case, such a rule cannot be used to infer a transition from $f(\mu,\nil)$ or $f(\nil, \mu)$. 
Hence $\mu \mathbin{\|} \nil \nbis f(\mu,\nil) + f(\nil, \mu)$, contradicting the soundness of Equation (\ref{Eqn:parf2}).
Therefore $f$ must have a rule of the form (\ref{asyncruleleft}) or (\ref{asyncruleright}) for each action $\mu$.

To conclude the proof we need to show that for each rule of the form \eqref{asyncruleleft} or \eqref{asyncruleright} the target term $t (x,y)$ is bisimilar to $x \mathbin{\|} y$.
For simplicity, we expand the proof only for the case of rules of the form \eqref{asyncruleleft}.
The proof for rules of the form \eqref{asyncruleright} follows by the same reasoning.

We proceed by a case analysis over the structure of $t(y_1,x_2)$, which, we recall, under Assumption~\ref{assumption:targets} can be either a variable in $\{y_1,x_2\}$ or a term of the form $g(y_1,x_2)$ for some $\FCCS$ operator $g$.
Our aim is to show that the only possibility is to have $t(y_1,x_2) = y_1 \mathbin{\|} x_2$, as any other process term would invalidate one of our simplifying assumptions.
\begin{itemize}
\item {\sc Case $t$ is a variable in $\{y_1,x_2\}$.}
We can distinguish two cases, according to which variable is considered:
\begin{itemize}
\item $t = y_1$.
Consider process $p = \mu.\nil$.
Since $p \trans[\mu] \nil$, from an application of rule \eqref{asyncruleleft} we can infer that $f(p,p) \trans[\mu] \nil$, and thus $f(p,p) + f(p,p) \trans[\mu] \nil$.
However, there is no $\mu$-transition from $p \mathbin{\|} p$ to a process bisimilar to $\nil$, as whenever $p \mathbin{\|} p \trans[\mu] q$, then $q$ is a process that will always be able to perform a second $\mu$-transition.
Hence, we would have $p \mathbin{\|} p \nbis f(p,p) + f(p,p)$, thus contradicting the soundness of Equation~\eqref{Eqn:parf2}.
\item $t = x_2$.
Consider process $p = \mu.\mu.\nil$.
Since $p \trans[\mu] \mu.\nil$, from an application of rule \eqref{asyncruleleft} we can infer that $f(p,\nil) \trans[\mu] \nil$ and thus $f(p,\nil) + f(\nil,p) \trans[\mu] \nil$.
However, there is no $\mu$-transition from $p \mathbin{\|} \nil$ to a process bisimilar to $\nil$, as whenever $p \mathbin{\|} \nil \trans[\mu] q$, then $q$ is a process that will always be able to perform a second $\mu$-transition.
Hence we would have $p \mathbin{\|} \nil \nbis f(p,\nil) + f(\nil,p)$, thus contradicting the soundness of Equation~\eqref{Eqn:parf2}.
\end{itemize}

\item {\sc Case $t$ is a term of the form $g(y_1,x_2)$} for some $\FCCS$ operator $g$.
We can distinguish three cases, according to which operator is used:
\begin{itemize}
\item {\sc $g$ is the prefix operator.}
We can distinguish two cases, according to which variable of the rule occurs in $t$:
\begin{itemize}
\item $t = \nu.y_1$.
Consider process $p = \mu.\nil$.
Since $p \trans[\mu] \nil$, from an application of rule \eqref{asyncruleleft} we can infer that $f(p,\nil) \trans[\mu] \nu.\nil \trans[\nu] \nil$, and thus $f(p,\nil) + f(\nil,p) \trans[\mu]\trans[\nu] \nil$.
However, $p \mathbin{\|} \nil \trans[\mu] \nil \mathbin{\|} \nil \ntrans[\nu]$.
Hence, we would have that $p \mathbin{\|} \nil \nbis f(p, \nil) + f(\nil,p)$, thus contradicting the soundness of Equation~\eqref{Eqn:parf2}.
\item $t = \nu.x_2$.
This case is analogous to the previous one.
\end{itemize}
\item {\sc $g$ is the nondeterministic choice operator} and thus $t = y_1 + x_2$.
Consider processes $p = \mu.\mu.\nil$ and $q = \mu.\nil$.
Since $p \trans[\mu] q$, from an application of rule \eqref{asyncruleleft} we can infer that $f(p,q) \trans[\mu] q + q \trans[\mu] \nil$, and thus $f(p,q) + f(q,p) \trans[\mu] \trans[\mu] \nil$.
However, there is no process $p'$ such that $p \mathbin{\|} q \trans[\mu]\trans[\mu] p'$ and $p' \bis\nil$, since $p'$ can always perform an additional $\mu$-transition. 
Hence, we would have $p \mathbin{\|} q \nbis f(p,q) + f(q,p)$, which contradicts the soundness of Equation~\eqref{Eqn:parf2}.

\item $g = f$.
First of all, we notice that in this case we can infer that $f$ cannot have both types of rules of the form \eqref{syncrule}, and both types of rules, \eqref{asyncruleleft} and \eqref{asyncruleright}, for all actions.
In fact, if this was the case, due to Lemmas~\ref{Lem:syncrules} and~\ref{Lem:asyncrules}, the set of rules defining the behaviour of $f(x_1,x_2)$ would be 
\begin{align*}
&\SOSrule{x_1 \trans[\mu] y_1}{f(x_1,x_2) \trans[\mu] f(y_1,x_2)}
\qquad
\SOSrule{x_2 \trans[\mu] y_2}{f(x_1,x_2) \trans[\mu] f(x_1,y_2)}
\\
&\SOSrule{x_1 \trans[\alpha] y_1 \quad x_2 \trans[\bar{\alpha}]y_2}{f(x_1,x_2) \trans[\tau] y_1 \mathbin{\|} y_2}
\qquad
\SOSrule{x_1 \trans[\bar{\alpha}] y_1 \quad x_2 \trans[\alpha]y_2}{f(x_1,x_2) \trans[\tau] y_1 \mathbin{\|} y_2}
\end{align*}
with $\mu \in \Act$ and $\alpha \in \{a,\bar{a}\}$.
Operator $f$ would then be a mere renaming of the parallel composition operator.
In particular, as a one-to-one correspondence between the rules for $f$ and those for $\mathbin{\|}$ could be established, we have that $f(x,y)$ would be \emph{bisimilar under formal hypothesis} to $x \mathbin{\|} y$ (see~\cite[Definition 1.10]{dS85}) and therefore, by~\cite[Theorem 1.12]{dS85}, we could directly conclude that $f(x,y) \approx x \mathbin{\|} y$ for all $x,y$.
However, this would contradict the fact that $x \parallel y \not\approx f(x,y)$.
Let us now consider the case of an operator $f$ having both types of rules, \eqref{asyncruleleft} and \eqref{asyncruleright}, and only one type of rules of the form \eqref{syncrule}, say the rule $S^f_{\alpha,\bar{\alpha}}$.
We proceed towards contradiction and distinguish two subcases, according to whether the order of the arguments is preserved or not by the rules of type \eqref{asyncruleleft} with label $\alpha$.
Similar arguments would allow us to deal with rules of type \eqref{asyncruleright}.
\begin{itemize}
\item The target of the rule is $f(y_1,x_2)$.
Then $f(\alpha.\bar{\alpha},\alpha) \trans[\alpha] f(\bar{\alpha},\alpha) \bis \bar{\alpha}.\alpha + \alpha.\bar{\alpha}$.
However, there is no $\alpha$-transition from $\alpha.\bar{\alpha} \| \alpha$ to a process bisimilar to $\bar{\alpha}.\alpha + \alpha.\bar{\alpha}$, thus contradicting the soundness of Equation~\eqref{Eqn:parf2}.
\item The target of the rule is $f(x_2,y_1)$.
Then $f(\alpha.\alpha,\bar{\alpha}.\alpha) \trans[\alpha] f(\bar{\alpha}.\alpha,\alpha) \ntrans[\tau]$.
However, whenever $\alpha.\alpha \| \bar{\alpha}.\alpha$ performs an $\alpha$-transition, it always reaches a process that can perform a $\tau$-move.
This contradicts the soundness of Equation~\eqref{Eqn:parf2}.
\end{itemize}

Finally, let us deal with the case in which there is at least one action $\mu \in \Act$ for which only one rule among \eqref{asyncruleleft} and \eqref{asyncruleright} is available.
According to our current simplifying assumptions, let \eqref{asyncruleleft} be the available rule for $f$ with label $\mu$. 
We can distinguish two cases, according to the occurrences of the variables of the rule in $t$:
\begin{itemize}
\item $t = f(y_1,x_2)$.
Consider process $p = \mu.\nil$.
Since $p \trans[\mu]\nil$, from an application of rule \eqref{asyncruleleft} we can infer that $f(p,p) \trans[\mu] f(\nil,p)$, and thus $f(p,p) + f(p,p) \trans[\mu] f(\nil,p)$, with $f(\nil,p) \bis \nil$, since only rules of the form \eqref{asyncruleleft} are available with respect to action $\mu$.
However, there is no $\mu$-transition from $p \mathbin{\|} p$ to a process bisimilar to $\nil$, as whenever $p \mathbin{\|} p \trans[\mu] q$ then $q$ is a process that will always be able to perform a second $\mu$-transition.
Hence, we would have $p \mathbin{\|} p \nbis f(p,p) + f(p,p)$, thus contradicting the soundness of Equation~\eqref{Eqn:parf2}.
\item $t = f(x_2,y_1)$.
Consider process $p = \mu.\mu.\nil$.
Since $p \trans[\mu] \mu.\nil$, and only rules of the from \eqref{asyncruleleft} are available with respect to action $\mu$, we can infer that $f(p,\nil) \trans[\mu] f(\nil,\mu.\nil) \ntrans[\mu]$ and $f(\nil,p) \ntrans[]$, which means that $f(p,\nil) + f(\nil,p)$ cannot perform two $\mu$-transitions in a row.
However, we have that $p \mathbin{\|} \nil \trans[\mu] \mu.\nil \mathbin{\|} \nil \trans[\mu] \nil \mathbin{\|} \nil$.
Hence, we would have $p \mathbin{\|} \nil \nbis f(p,\nil) + f(\nil,p)$, thus contradicting the soundness of Equation~\eqref{Eqn:parf2}.
\end{itemize}
\end{itemize}
\end{itemize}
\end{proof}

Henceforth we assume, without loss of generality, that the target of a rule of the form (\ref{asyncruleleft}) is $y_1 \| x_2$ and the target of a rule of the form (\ref{asyncruleright}) is $x_1 \| y_2$.

For each $\mu \in \Act$, we introduce two unary predicates, $L^{\ff}_{\mu}$ and $R^{\ff}_{\mu}$, that allow us to identify which rules with label $\mu$ are available for $f$.
In detail, 
\begin{itemize}
\item $L^{\ff}_{\mu}$ holds if $f$ has a rule of the form~\eqref{asyncruleleft} with label $\mu$;
\item $R^{\ff}_{\mu}$ holds if $f$ has a rule of the form~\eqref{asyncruleright} with label $\mu$.
\end{itemize}
We let $L^{\ff}_{\mu} \wedge R^{\ff}_{\mu}$ denote that $f$ has both a rule of the form~\eqref{asyncruleleft} and one of the form~\eqref{asyncruleright} with label $\mu$.
We stress that, for each $\mu$, the validity of predicate $L^{\ff}_{\mu}$ does not prevent $R^{\ff}_{\mu}$ from holding, and vice versa.
Throughout the paper, in case \emph{only one} of the two predicates holds, we will clearly state it.

Summing up, we have obtained that:

\begin{proposition}
\label{prop:new}
If $f$ meets Assumptions~\ref{Ass:deSimone} and~\ref{assumption:targets} and Equation (\ref{Eqn:parf2}) is sound modulo bisimilarity, then $f$ must satisfy $S_{\alpha,\bar{\alpha}}^f$ for at least one $\alpha \in \{a,\bar{a}\}$, and, for each $\mu\in\Act$, at least one of $L^{\ff}_{\mu}$ and $R^{\ff}_{\mu}$.
\end{proposition}

The next proposition states that this is enough to obtain the soundness of Equation (\ref{Eqn:parf2}).

\begin{proposition}
\label{prop:enough-rules}
Assume that all of the rules for $f$ have the form (\ref{syncrule}), (\ref{asyncruleleft}), or (\ref{asyncruleright}).
If $S_{\alpha,\bar{\alpha}}^f$ holds for at least one $\alpha \in \{a,\bar{a}\}$, and, for each $\mu\in\Act$, at least one of $L^{\ff}_{\mu}$ and $R^{\ff}_{\mu}$ holds, then Equation (\ref{Eqn:parf2}) is sound modulo bisimilarity.
\end{proposition}

\begin{proof}
We argue that the relation 
\[
\mathcal{B} = \{ (p \mathbin{\|} q, f(p,q)+f(q,p)) \mid \text{$p,q$ closed terms} \} \cup \bis 
\] 
is a bisimulation. 
To this end, pick closed terms $p,q$. 
Now show, using the information on the rules for $f$ given in the proviso of the proposition, that, for each action $\mu$ and closed term $r$,
\begin{itemize}
\item whenever $p \mathbin{\|} q \mv{\mu} r$, there is a term $r'$ that is equal to $r$ up to commutativity of $\|$ such that $f(p,q)+f(q,p) \mv{\mu} r'$, and
\item whenever $f(p,q)+f(q,p) \mv{\mu} r$, there is a term $r'$ that is equal to $r$ up to commutativity of $\|$ such that $p \mathbin{\|} q \mv{\mu} r'$.
\end{itemize}
The claim follows because $\|$ is commutative modulo $\bistext$. 
\end{proof}

When the set of actions is $\{a, \bar{a}, \tau\}$, there are $81$ operators that satisfy the constraints in Propositions~\ref{prop:new} and~\ref{prop:enough-rules}, including parallel composition and Hennessy's merge. 
In general, when the set of actions has $2n+1$ elements, there are $3^{3n+1}$ possible operators meeting those constraints. 

As an immediate consequence of the form of the rules for $f$ given in Proposition~\ref{prop:enough-rules}, we have the following lemma:

\begin{lemma}
\label{Lem:finite-branching}
Assume that all of the rules for $f$ have the form (\ref{syncrule}), (\ref{asyncruleleft}), or (\ref{asyncruleright}).
Then each closed term $p$ in $\FCCS$ is finitely branching, that is, the set $ \{ (\mu,q) \mid p \mv{\mu} q \}$ is finite.
\end{lemma}

\begin{remark}
\label{Rem:tree-equivalence}
A standard consequence of the finiteness of the depth (Lemma~\ref{Lem:finite-depth}), and of $\FCCS$ processes being finite branching, is that each closed $\FCCS$ term is bisimilar to a synchronisation tree~\cite{Mi80}, that is, a closed term built using only the constant $\nil$, the unary prefixing operations and the binary $+$ operation.  
Since bisimilarity is a congruence over $\FCCS$, this means, in particular, that an equation $t \approx u$ over $\FCCS$ is sound modulo bisimilarity if, and only if, the closed terms $\sigma(t)$ and $\sigma(u)$ are bisimilar for each substitution mapping variables to synchronisation trees. 
Moreover, we can use the sub-language of synchronisation trees, which is common to all of the languages $\FCCS$, to compare terms from these languages for dif{f}erent choices of binary operation $f$ with respect to bisimilarity.
\end{remark}


\section{The main theorem and its proof strategy}
\label{sec:proof_method}

So far we have discussed the three simplifying assumptions and the structural operational semantics of the operators satisfying them:
\begin{description}
\item[Assumption 1]
The behaviour of $f$ is described by rules in de Simone format. 
\item[Assumption 2]
Equation
\begin{equation}
\tag{4}
x\mathbin{\|} y \approx f(x,y) + f(y,x) 
\end{equation}
is sound modulo bisimilarity.
\item[Assumption 3]
If $t$ is the target of a rule for $f$, then $t$ is either a variable or a term obtained by applying a single $\FCCS$ operator to the variables of the rule, according to the constraints of the de Simone format.
\end{description}
We have obtained that if a binary operator $f$ meets these assumptions, then it must satisfy $S_{\alpha,\bar{\alpha}}^f$ for at least one $\alpha \in \{a,\bar{a}\}$ (i.e., it must guarantee a form of synchronisation between its arguments), and, for each $\mu\in\Act$, at least one of $L^{\ff}_{\mu}$ and $R^{\ff}_{\mu}$ (i.e., for each action $\mu$, at least one of its arguments is allowed to interleave a $\mu$-transition).

Our order of business will now be to use this information to prove our main result, namely the following theorem:

\begin{theorem}
\label{Thm:f(x,y)+f(y,x)_general}
Assume that $f$ satisfies Assumptions~\ref{Ass:deSimone}--\ref{assumption:targets}.
Then bisimilarity admits no finite equational axiomatisation over $\FCCS$.
\end{theorem}

This is achieved by using the \emph{proof-theoretic} technique from~\cite{Mo89,Mo90,Mo90a} to prove a stronger result, of which Theorem~\ref{Thm:f(x,y)+f(y,x)_general} is an immediate consequence, namely:
\begin{theorem}
\label{Thm:f(x,y)+f(y,x)}
Assume that $f$ satisfies Assumptions~\ref{Ass:deSimone}--\ref{assumption:targets}.
Then bisimilarity admits no finite, ground-complete axiomatisation over $\FCCS$.
\end{theorem}

In this section, we discuss the general reasoning behind the proof of Theorem~\ref{Thm:f(x,y)+f(y,x)}.
In light of Propositions~\ref{prop:new} and~\ref{prop:enough-rules}, to prove Theorem~\ref{Thm:f(x,y)+f(y,x)} we will proceed by a case analysis over the possible sets of allowed SOS rules for operator $f$.
In each case, our proof method will follow the same general schema, which has its roots in Moller's arguments to the effect that bisimilarity is not finitely based over CCS (see, e.g.,~\cite{AFIL05,Mo89,Mo90,Mo90a}), and that we present here at an informal level.

The main idea is to identify a {\sl witness property of the negative result}.
This is a specific property of $\FCCS$ terms, say $W_n$ for $n \ge 0$, that, when $n$ is \emph{large enough}, is preserved by provability from finite axiom systems.
Roughly, this means that if $\E$ is a finite set of axioms that are sound modulo bisimilarity, the equation $p \approx q$ is provable from $\E$, and $n$ is greater than the size of all the terms in the equations in $\E$, then either both $p$ and $q$ satisfy $W_n$, or none of them does. 
Then, we exhibit an infinite family of valid equations, say $e_n$, called accordingly {\sl witness family of equations for the negative result}, in which $W_n$ is not preserved, namely it is satisfied only by one side of each equation.
Thus, Theorem~\ref{Thm:f(x,y)+f(y,x)} specialises to:

\begin{theorem}
\label{Thm:nonfin-en}
Suppose that Assumptions~\ref{Ass:deSimone}--\ref{assumption:targets} are met.
Let $\E$ be a finite axiom system over $\FCCS$ that is sound modulo bisimilarity. 
Then there is an infinite family $\{e_n\}_{n\geq 0}$ of sound equations such that $\E$ does not prove equation $e_n$ for each $n$ larger than the size of each term in the equations in $\E$. 
\end{theorem}

In this paper, the property $W_n$ corresponds to having a summand that is bisimilar to a specific process.
In detail: 
\begin{enumerate}
\item We identify, for each case, a family of processes $f(\mu,p_n)$, for $n \ge 0$, with the choices of $\mu$ and $p_n$ tailored to the particular set of SOS rules allowed for $f$.
Moreover, process $p_n$ will have size at least $n$, for each $n \ge 0$.
We shall refer to the processes $f(\mu,p_n)$ as the \emph{witness processes}.
\item \label{strategy_item_2} 
We prove that by choosing $n$ {\sl large enough}, given a finite set of valid equations $\E$ and processes $p,q \bis f(\mu,p_n)$, if $\E \vdash p \approx q$ and $p$ has a summand bisimilar to $f(\mu,p_n)$, then also $q$ has a summand bisimilar to $f(\mu,p_n)$.
Informally, we will choose $n$ greater than the size of all the terms in the equations in $\E$, so that we are guaranteed that the behaviour of the summand bisimilar to $f(\mu,p_n)$ is due to a closed substitution instance of a variable.
\item We provide an infinite family of valid equations $e_n$ in which one side has a summand bisimilar to $f(\mu,p_n)$, but the other side does not.
In light of item~\ref{strategy_item_2}, this implies that such a family of equations cannot be derived from any finite collection of valid equations over $\FCCS$, modulo bisimilarity, thus proving Theorem~\ref{Thm:nonfin-en}.
\end{enumerate}

To narrow down the combinatorial analysis over the allowed sets of SOS rules for $f$, we examine first the \emph{distributivity properties}, modulo $\bistext$, of the operator $f$ over summation. 

Firstly, we notice that $f$ cannot distribute over summation in both arguments. 
This is a consequence of our previous analysis of the operational rules that such an operator $f$ may and must have in order for Equation (\ref{Eqn:parf2}) to hold.  
However, it can also be shown in a purely algebraic manner.

\begin{lemma}
\label{lem:no-two-sided-dist}
A binary operator satisfying Equation (\ref{Eqn:parf2}) cannot distribute over $+$ in both arguments.
\end{lemma}

\begin{proof}
Assume, towards a contradiction, that $f$ distributes over summation in both arguments. 
Then, using Equation~\eqref{Eqn:parf2}, we have that: 
\begin{eqnarray*}
(x+y) \mathbin{\|} z & \approx & f(x+y,z) + f(z,x+y) \\
& \approx & f(x,z) + f(y,z) + f(z,x) + f(z,y) \\
& \approx & (x \mathbin{\|} z) + (y \mathbin{\|} z) \enspace .
\end{eqnarray*}
However, this is a contradiction because, as is well known, the equation 
\[
(x+y) \mathbin{\|} z \approx (x \mathbin{\|} z) + (y \mathbin{\|} z)
\]
is not sound in bisimulation semantics. 
For example, $(a+\tau) \mathbin{\|} a \nbis (a \mathbin{\|} a) + (\tau \mathbin{\|} a)$.
\end{proof}

Hence, we can limit ourselves to considering binary operators satisfying our constraints that, modulo bisimilarity, distribute over $+$ in one argument or in none.

We consider these two possibilities in turn. 


\subsubsection*{Distributivity in one argument}
\label{sec:1_dist}

Due to our Assumptions~\ref{Ass:deSimone}--\ref{assumption:targets}, we can exploit a result from~\cite{ABV94} to characterise the rules for an operator $f$ that distributes over summation in one of its arguments.
More specifically,~\cite[Lemma 4.3]{ABV94} gives a condition on the rules for a \emph{smooth operator} $g$ in a GSOS system that includes the $+$ operator in its signature, which guarantees that $g$ distributes over summation in one of its arguments. 
(The rules defining the semantics of smooth operators are a generalisation of those in de Simone format.) Here we show that, for operator $f$, the condition in~\cite[Lemma 4.3]{ABV94} is both necessary and sufficient for distributivity of $f$ in one of its two arguments. 

\begin{lemma}
\label{lem:dist}
Let $i \in \{1,2\}$.
Modulo bisimilarity, operator $f$ distributes over summation in its $i$-th argument if and only if each rule for $f$ has a premise $x_i \trans[\mu_i] y_i$, for some $\mu_i$.
\end{lemma}

\begin{proof}
The ($\Leftarrow$) implication follows by similar arguments to those used in the proof of~\cite[Lemma 4.3]{ABV94}.
Hence, we omit it and we proceed to prove the ($\Rightarrow$) implication.

Assume that $f$ distributes over $+$ in some argument. 
We recall that by Lemma~\ref{Lem:asyncrules} for each action $\mu$ at least one between $L^{\ff}_{\mu}$ and $R^{\ff}_{\mu}$ must hold.
We aim to prove that either $L^{\ff}_{\mu}$ holds for all actions $\mu$ and none of the $R^{\ff}_{\mu}$ does, or vice versa. 
Indeed, suppose towards a contradiction that there are rules satisfying $L^{\ff}_{\mu}$ and $R^{\ff}_{\nu}$ for some actions $\mu$ and $\nu$. 
Then
\begin{itemize}
\item $f(\tau+\tau^2, \nu) \nbis f(\tau, \nu) + f(\tau^2, \nu)$, because the validity of $R^{\ff}_{\nu}$ allows us to prove that $f(\tau+\tau^2, \nu) \trans[\nu] (\tau+\tau^2) \| \nil$ and $f(\tau, \nu) + f(\tau^2, \nu)$ cannot match that transition up to bisimilarity.
\item $f(\mu,\tau+\tau^2) \nbis f(\mu,\tau) + f(\mu,\tau^2)$, because the validity of $L^{\ff}_{\mu}$ allows us to prove that $f(\mu,\tau+\tau^2) \trans[\mu] \nil \| (\tau+\tau^2)$ and $f(\mu,\tau) + f(\mu,\tau^2)$ cannot match that transition up to bisimilarity. 
\end{itemize}
\end{proof}

By Proposition~\ref{prop:new}, Lemma~\ref{lem:dist} implies that, when $f$ is distributive in one argument, either $L^{\ff}_{\mu}$ holds for all $\mu\in\Act$ or $R^{\ff}_{\mu}$ holds for all $\mu \in \Act$, and $S_{\alpha,\bar{\alpha}}$ holds for at least one $\alpha \in \{a,\bar{a}\}$.
Notice that if $L^{\ff}_{\mu}$ holds for each action $\mu$ and both $S_{a,\bar{a}}^f$ and $S_{\bar{a},a}^f$ hold, then $f$ behaves as Hennessy's merge $\hmerge$~\cite{He88}, and our Theorem~\ref{Thm:nonfin-en} specialises to~\cite[Theorem~22]{AFIL05}. 
Hence we assume, without loss of generality, that $S_{\alpha,\bar{\alpha}}^f$ holds for only one $\alpha\in\{a,\bar{a}\}$.
A similar reasoning applies if $R^{\ff}_{\mu}$ holds for each action $\mu$.

In Section~\ref{sec:Labat} we will present the proof of Theorem~\ref{Thm:nonfin-en} in the case of an operator $f$ that distributes over summation in its first argument (see Theorem~\ref{thm:Labat}).


\subsubsection*{Distributivity in neither argument}
\label{sec:no_dist}

We now consider the case in which $f$ does not distribute over summation in either argument. 

Also in this case, we can exploit Lemma~\ref{lem:dist} to obtain a characterisation of the set of rules allowed for an operator $f$ satisfying the desired constraints.
In detail, we infer that there must be $\mu, \nu \in\Act$, not necessarily distinct, such that $L^{\ff}_{\mu}$ and $R^{\ff}_{\nu}$ hold.
Otherwise, as $f$ must have at least one rule for each action (see Proposition~\ref{prop:new}), at least one argument would be involved in the premises of each rule, and this would entail distributivity over summation in that argument.

We will split the proof of Theorem~\ref{Thm:nonfin-en} for an operator $f$ that, modulo bisimilarity, does not distribute over summation in either argument into three main cases:
\begin{enumerate}
\item In Section~\ref{sec:LaRa}, we consider the case of $L^{\ff}_{\alpha} \wedge R^{\ff}_{\alpha}$ holding, for some $\alpha \in \{a,\bar{a}\}$ (Theorem~\ref{thm:LaRa}). 
\item In Section~\ref{sec:LaRba}, we deal with the case of $f$ having only one rule for $\alpha$, only one rule for $\bar{\alpha}$, and such rules are of dif{f}erent forms.
As we will see, we will need to distinguish two subcases, according to which predicate $S_{\alpha,\bar{\alpha}}^f$ holds (Theorem~\ref{thm:LaRbaC} and Theorem~\ref{thm:LaRba}). 
\item Finally, in Section~\ref{sec:Lt}, we study the case of $f$ having only one rule with label $\alpha$, only one rule with label $\bar{\alpha}$, and such rules are of the same type (Theorem~\ref{thm:Lt_1}).
\end{enumerate}

Before proceeding to the proofs of the various cases, we devote Sections~\ref{sec:axioms} and~\ref{sec:prelim} to the introduction of some general preliminary results and observations, concerning, respectively, the equational theory over $\FCCS$ and the decomposition of the semantics of terms, that will play a key role in the technical development of our proofs.


\section{The equational theory of $\FCCS$}
\label{sec:axioms}

In this section we study some aspects of the equational theory of $\FCCS$ modulo bisimilarity that are useful in the proofs of our negative results.
In particular, we show that, due to Equation~\eqref{Eqn:parf2}, proving the negative result over $\FCCS$ is equivalent to proving it over its reduct $\FCCS^-$, whose signature does not contain occurrences of $\|$ (Proposition~\ref{Prop:no-par} below).

Furthermore, we discuss the relation between the available rules for $f$ and the bisimilarity of terms of the form $f(p,q)$ with $\nil$.
As we will see, in the case of an operator $f$ that distributes over summation in one argument, it is possible to {\sl saturate} the axiom systems~\cite{Mo89} yielding a simplification in the proofs (Proposition~\ref{Propn:proofswithout0} below).
On the other hand, we cannot rely on saturation for an operator $f$ that distributes with respect to $+$ in neither of its arguments. 


\subsection{Simplifying equational proofs}

We show that it is sufficient to prove that bisimilarity admits no finite equational axiomatisation over $\FCCS^-$, consisting of the $\FCCS$ terms that do not contain occurrences of $\|$.

\begin{definition}
\label{Defn:t-hat}
For each $\FCCS$ term $t$, we define $\hat{t}$ as follows: 
\[
\begin{array}{lclcl}
\hat{\nil} = \nil & \qquad & 
\hat{x} = x & \qquad & 
\widehat{\mu t} = \mu \hat{t} \\
\widehat{t+u} = \hat{t} + \hat{u} & \qquad & 
\widehat{f(t,u)} = f(\hat{t},\hat{u}) & \qquad & 
\widehat{t \| u} = f(\hat{t},\hat{u}) + f(\hat{u},\hat{t})
\enspace .
\end{array}
\]
Then, for any axiom system $\E$ over $\FCCS$, we let $\widehat{\E} = \{ \hat{t} \approx \hat{u} \mid (t\approx u)\in \E \}$.
\end{definition}

We notice that, for each $\FCCS$ term $t$, the term $\hat{t}$ is in $\FCCS^-$. Moreover, if $t$ contains no occurrences of the parallel composition operator, then $\hat{t}=t$.  
Since Equation~\eqref{Eqn:parf2} is sound modulo bisimilarity, which is a congruence relation, it is not hard to show that each term $t$ in $\FCCS$ is bisimilar to $\hat{t}$.
Therefore if $\E$ is an axiom system over $\FCCS$ that is sound modulo bisimilarity, then $\widehat{\E}$ is an axiom system over $\FCCS^-$ that is sound modulo bisimilarity.

The following result states the reduction of the non-finite axiomatisability of $\bistext$ over $\FCCS$ to that of $\bistext$ over $\FCCS^-$.

\begin{proposition}
\label{Prop:no-par}
Let $\E$ be an axiom system over $\FCCS$. 
Then:
\begin{enumerate}
\item \label{proof-red} 
If $\E \vdash t \approx u$, then $\widehat{\E} \vdash \hat{t} \approx \hat{u}$.
\item \label{compl-red} 
If $\E$ is a complete axiomatisation of $\bistext$ over $\FCCS$, then $\widehat{\E}$ completely axiomatises $\bistext$ over $\FCCS^-$.
\item \label{nonax-red} 
If bisimilarity is not finitely axiomatisable over $\FCCS^-$, then it is not finitely axiomatisable over $\FCCS$ either.
\end{enumerate}
\end{proposition}

\begin{proof}
We prove the three statements separately. 
\begin{itemize}
\item {\sc Proof of Statement~\ref{proof-red}}.  
Assume that $\E \vdash t \approx u$. We shall argue that $\widehat{\E}$ proves the equation $\hat{t} \approx \hat{u}$ by induction on the depth of the proof of $t\approx u$ from $\E$. 
We proceed by a case analysis on the last rule used in the proof.  
Below we only consider the two most interesting cases in this analysis.
\begin{itemize}
\item {\sc Case $\E \vdash t \approx u$, because $\sigma(t')=t$ and $\sigma(u')=u$ for some equation $(t'\approx u')\in \E$}. 
Note, first of all, that, by the definition of $\widehat{\E}$, the equation $\widehat{t'} \approx \widehat{u'}$ is contained in $\widehat{\E}$. 
Observe now that $\hat{t} = \hat{\sigma}(\widehat{t'}) \text{ and } \hat{u} = \hat{\sigma}(\widehat{u'})$, where $\hat{\sigma}$ is the substitution mapping each variable $x$ to the term $\widehat{\sigma(x)}$. 
It follows that the equation $\hat{t} \approx \hat{u}$ can be proven from the axiom system $\widehat{\E}$ by instantiating the equation $\widehat{t'} \approx \widehat{u'}$ with the substitution $\hat{\sigma}$, and we are done.
\item {\sc Case $\E \vdash t \approx u$, because $t=t_1\| t_2$ and $u=u_1\| u_2$ for some $t_i, u_i$ ($i=1,2$) such that $\E \vdash t_i \approx u_i$ ($i=1,2$)}.
Using the inductive hypothesis twice, we have that $\widehat{\E}\vdash\widehat{t_i} \approx \widehat{u_i}$ ($i=1,2$). 
Therefore, using substitutivity,  $\widehat{\E}$ proves that
\[
\hat{t} = f(\widehat{t_1},\widehat{t_2}) + f(\widehat{t_2},\widehat{t_1}) 
\approx f(\widehat{u_1},\widehat{u_2}) + f(\widehat{u_2},\widehat{u_1})
= \hat{u} \enspace ,
\]
which was to be shown. 
\end{itemize}
The remaining cases are simpler, and we leave the details to the
reader.
\item {\sc Proof of Statement~\ref{compl-red}}. 
Assume that $t$ and $u$ are two bisimilar terms in the language $\FCCS^-$. 
We shall argue that $\widehat{\E}$ proves the equation ${t} \approx {u}$. 
To this end, notice that the equation $t \approx u$ also holds in the algebra of $\FCCS$ terms modulo bisimulation. 
In fact, for each term $v$ in the language $\FCCS$ and closed substitution $\sigma$ mapping variables to $\FCCS$ terms, we have that
\[
\sigma(v) \bis \hat{\sigma}(v) \enspace , 
\]
where the substitution $\hat{\sigma}$ is defined as above. 
  
Since $\E$ is complete for bisimilarity over $\FCCS$ by our assumptions, it follows that $\E$ proves the equation $t \approx u$.
Therefore, by statement~\ref{proof-red} of the proposition, we have that $\widehat{\E}$ proves the equation $\hat{t} \approx \hat{u}$.
The claim now follows because $\hat{t}=t$ and $\hat{u}=u$.
\item {\sc Proof of Statement~\ref{nonax-red}}. 
This is an immediate consequence of statement~\ref{compl-red} because $\widehat{\E}$ has the same cardinality of $\E$, and is therefore finite, if so is $\E$.
\end{itemize}
\end{proof}

In light of this result, henceforth we shall focus on proving that $\bistext$ affords no finite equational axiomatisation over $\FCCS^-$.


\subsection{Bisimilarity with $\nil$}

As a further simplification, we can focus on the $\nil$ \emph{absorption properties} of $\FCCS^-$ operators.
Informally, we can restrict the axiom system to a collection of equations that do not introduce unnecessary terms that are bisimilar to $\nil$ in the equational proofs, namely $\nil$ summands and $\nil$ \emph{factors}.

\begin{definition}
\label{Lem:nil-factors}
We say that a $\FCCS^-$ term $t$ has a {\sl $\nil$ factor} if it contains a subterm of the form  $f(t',t'')$, where either $t'$ or $t''$ is bisimilar to $\nil$.
\end{definition}

The $\nil$ absorption properties of $f$ depend crucially on the allowed set of SOS rules for $f$.
Notably, we have dif{f}erent results, according to the distributivity properties of $f$.


\subsubsection{$\nil$ absorption for $f$ that distributes in one argument}

We examine first the case of an operator $f$ that, modulo bisimilarity, distributes over summation in its first argument.

\begin{table}
\[
\setlength{\arraycolsep}{.4cm}
\scalebox{0.8}{$
\begin{array}{lll}
\text{\underline{Some common axioms}} & 
\text{\underline{Some axioms for $L^{\ff}_{\Act}$ $R^{\ff}_{\emptyset}$}} & 
\text{\underline{Some axioms for $L^{\ff}_{\emptyset}$ $R^{\ff}_{\Act}$}}\\[.7ex]
\text{A1 }\; x + x \approx x & 
\text{F1 }\; f(\nil,x) \approx \nil & 
\text{F3 }\; f(x,\nil) \approx \nil\\
\text{A2 }\; x + y \approx y + x &
\text{F2 }\; f(x,\nil) \approx x &
\text{F4 }\; f(\nil,x) \approx x\\
\text{A3 }\; (x + y) + z \approx x + (y + z) \\
\text{A4 }\; x + \nil \approx x \\
\text{F0 }\; f(\nil,\nil) \approx \nil \\[1ex]
\text{\underline{Some axioms for $L^{\ff}_{\Act}$ $R^{\ff}_{\alpha}$}} &
\text{\underline{Some axioms for $L^{\ff}_{\Act}$ $R^{\ff}_{\bar{\alpha}}$}} &
\text{\underline{Some axioms for $L^{\ff}_{\Act}$ $R^{\ff}_{\tau}$}} \\[.7ex]
\text{F5 }\; f(\nil,\bar{\alpha}.x+w) \approx f(\nil,w) &
\text{F7 }\; f(\nil,\alpha.x + w) \approx f(\nil,w) &
\; \text{F2, F5, F7}\\
\text{F6 }\; f(\nil,\tau.x + w) \approx f(\nil,w) &
\text{and }\; \text{F2, F6}\\
\text{and }\; \text{F2}\\[1ex]
\text{\underline{Some axioms for $L^{\ff}_{\alpha}$ $R^{\ff}_{\Act}$}} &
\text{\underline{Some axioms for $L^{\ff}_{\bar{\alpha}}$ $R^{\ff}_{\Act}$}} &
\text{\underline{Some axioms for $L^{\ff}_{\tau}$ $R^{\ff}_{\Act}$}} \\[.7ex]
\text{F8 }\; f(\bar{\alpha}.x+w,\nil) \approx f(w,\nil) &
\text{F10 }\; f(\alpha.x + w,\nil) \approx f(w,\nil) &
\; \text{F4, F8, F10}\\
\text{F9 }\; f(\tau.x + w,\nil) \approx f(w,\nil) &
\text{and }\; \text{F4, F9}\\
\text{and }\; \text{F4}\\[1.1ex]
\text{\underline{Some axioms for $L^{\ff}_{\Act}$ $R^{\ff}_{\alpha,\bar{\alpha}}$}} &
\text{\underline{Some axioms for $L^{\ff}_{\Act}$ $R^{\ff}_{\alpha,\tau}$}} &
\text{\underline{Some axioms for $L^{\ff}_{\Act}$ $R^{\ff}_{\bar{\alpha},\tau}$}} \\[.7ex]
\; \text{F2, F6} &
\; \text{F2, F5} &
\; \text{F2, F7}\\[1ex]
\text{\underline{Some axioms for $L^{\ff}_{\alpha,\bar{\alpha}}$ $R^{\ff}_{\Act}$}} &
\text{\underline{Some axioms for $L^{\ff}_{\alpha,\tau}$ $R^{\ff}_{\Act}$}} &
\text{\underline{Some axioms for $L^{\ff}_{\bar{\alpha},\tau}$ $R^{\ff}_{\Act}$}} \\[.7ex]
\; \text{F4, F9} &
\; \text{F4, F8} &
\; \text{F4,F10}\\[1ex]
\text{\underline{Some axioms for $L^{\ff}_{\alpha,\tau}$ $R^{\ff}_{\bar{\alpha}}$ $S_{\alpha,\bar{\alpha}}$}} &
\text{\underline{Some axioms for $L^{\ff}_{\alpha,\tau}$ $R^{\ff}_{\bar{\alpha}}$ $S_{\bar{\alpha},\alpha}$}} &
\text{\underline{Some axioms for $L^{\ff}_{\alpha,\tau}$ $R^{\ff}_{\alpha,\bar{\alpha}}$}} \\[.7ex]
\text{F11 }\; f(\bar{\alpha}.x + w, \tau.y) \approx f(w,\tau.y) &
\; \text{F6, F7, F8, F11, F12, F15} &
\; \text{F6, F8, F11, F12, F15}\\
\text{F12 }\; f(\bar{\alpha}.x, \tau.y + w) \approx f(\bar{\alpha}.x,w) \\
\text{F13 }\; f(\bar{\alpha}.x+w,\alpha.y) \approx f(w,\alpha.y) &
\text{\underline{Some axioms for $L^{\ff}_{\alpha,\tau}$ $R^{\ff}_{\bar{\alpha},\tau}$ $S_{\alpha,\bar{\alpha}}$}} &
\text{\underline{Some axioms for $L^{\ff}_{\alpha,\tau}$ $R^{\ff}_{\bar{\alpha},\tau}$ $S_{\bar{\alpha},\alpha}$}}\\
\text{F14 }\; f(\bar{\alpha}.x,\alpha.y+w) \approx f(\bar{\alpha}.x, w) &
\; \text{F7, F8, F13, F14, F16} &
\; \text{F7, F8}\\
\text{F15 }\; f(\bar{\alpha}.x,\tau.y) \approx \nil \\
\text{F16 }\; f(\bar{\alpha}.x,\alpha.y) \approx \nil\\
\text{and }\; \text{F6, F7, F8} \\[1ex]
\text{\underline{Some axioms for $L^{\ff}_{\bar{\alpha}}$ $R^{\ff}_{\alpha,\tau}$ $S_{\bar{\alpha},\alpha}$}} &
\text{\underline{Some axioms for $L^{\ff}_{\bar{\alpha}}$ $R^{\ff}_{\alpha,\tau}$ $S_{\alpha,\bar{\alpha}}$}} &
\text{\underline{Some axioms for $L^{\ff}_{\bar{\alpha},\alpha}$ $R^{\ff}_{\alpha,\tau}$}} \\[.7ex]
\text{F17 }\; f(\alpha.x + w, \bar{\alpha}.y) \approx f(w,\bar{\alpha}.y) &
\; \text{F5, F9, F10, F19, F20, F22} &
\; \text{F9, F19, F20, F22}\\
\text{F18 }\; f(\alpha.x, \bar{\alpha}.y + w) \approx f(\alpha.x,w) \\
\text{F19 }\; f(\tau.x+w,\bar{\alpha}.y) \approx f(w,\bar{\alpha}.y) &
\text{\underline{Some axioms for $L^{\ff}_{\bar{\alpha},\tau}$ $R^{\ff}_{\alpha,\tau}$ $S_{\alpha,\bar{\alpha}}$}} &
\text{\underline{Some axioms for $L^{\ff}_{\bar{\alpha},\tau}$ $R^{\ff}_{\alpha,\tau}$ $S_{\bar{\alpha},\alpha}$}}\\
\text{F20 }\; f(\tau.x,\bar{\alpha}.y+w) \approx f(\tau.x, w) &
\; \text{F5, F10, F17, F18, F21} &
\; \text{F5, F10}\\
\text{F21 }\; f(\alpha.x,\bar{\alpha}.y) \approx \nil \\
\text{F22 }\; f(\tau.x,\bar{\alpha}.y) \approx \nil\\
\text{and }\; \text{F5, F9, F10} \\[1ex]
\text{\underline{Some axioms for $L^{\ff}_{\alpha}$ $R^{\ff}_{\bar{\alpha},\tau}$ $S_{\alpha,\bar{\alpha}}$}} &
\text{\underline{Some axioms for $L^{\ff}_{\alpha}$ $R^{\ff}_{\bar{\alpha},\tau}$ $S_{\bar{\alpha},\alpha}$}} &
\text{\underline{Some axioms for $L^{\ff}_{\tau}$ $R^{\ff}_{\alpha,\bar{\alpha}}$}} \\[.7ex]
\text{F23}\; f(\tau.x +w, \alpha.y) \approx f(w,\alpha.y) &
\; \text{F7, F8, F9, F23, F24, F25} &
\text{F26 }\; f(\alpha.x+w,\tau.y) \approx f(w,\tau.y)\\
\text{F24 }\; f(\tau.x,\alpha.y+w) \approx f(\tau.x,w) &
&
\text{F27 }\; f(\alpha.x, \tau.y +w) \approx f(\alpha.x,w) \\
\text{F25 }\; f(\tau.x,\alpha.y) \approx \nil &
&
\text{F28 }\; f(\alpha.x,\tau.y) \approx \nil\\
\text{and }\; \text{F7, F8, F9, F13, F14, F16} &
&
\text{and }\; \text{F6, F8, F10, F11, F12, F15} \\[1ex]
\text{\underline{Some axioms for $L^{\ff}_{\bar{\alpha},\tau}$ $R^{\ff}_{\alpha}$ $S_{\bar{\alpha},\alpha}$}} &
\text{\underline{Some axioms for $L^{\ff}_{\bar{\alpha},\tau}$ $R^{\ff}_{\alpha}$ $S_{\alpha,\bar{\alpha}}$}} &
\text{\underline{Some axioms for $L^{\ff}_{\alpha,\bar{\alpha}}$ $R^{\ff}_{\tau}$}} \\[.7ex]
\; \text{F5, F6, F10, F17, F18,} &
\; \text{F5, F6, F10, F26, F27, F28} &
\; \text{F5, F7, F9, F19, F20,} \\
\; \text{F21, F26, F27, F28} & 
& 
\;\text{F22, F23, F24, F25} \\[1ex]
\text{\underline{Some axioms for $L^{\ff}_{\alpha,\bar{\alpha}}$ $R^{\ff}_{\bar{\alpha},\tau}$}} &
\text{\underline{Some axioms for $L^{\ff}_{\bar{\alpha},\tau}$ $R^{\ff}_{\alpha,\bar{\alpha}}$}} \\[.7ex]
\; \text{F7, F9, F23, F24, F25} &
\; \text{F6, F10, F26, F27, F28}
\end{array}
$}
\]
\caption{\label{tab:axioms} Some sets of axioms, according to which rules are available for $f$.
For $A \subseteq \Act$ and $X \in \{L,R\}$, we use $X^f_A$ as a shorthand for $\bigwedge_{\mu \in A} X^f_\mu$.}
\end{table}

In this case, an example of a collection of equations over $\FCCS^-$ that are sound modulo $\bistext$ is given by axioms A1--A4, F0--F2 in Table~\ref{tab:axioms}.
Interestingly, axioms A4 and F1 in Table~\ref{tab:axioms} (used from left to right) are enough to establish that each $\FCCS^-$ term that is bisimilar to $\nil$ is also provably equal to $\nil$. 

Before proceeding, we observe the following:
\begin{remark}
\label{rmk:summands}
Whenever a process term $t$ has neither $\nil$ summands nor factors then we can assume that, for some finite non-empty index set $I$, $t = \sum_{i \in I} t_i$ for some terms $t_i$ such that none of them has $+$ as head operator and moreover, none of them has $\nil$ summands nor factors.
\end{remark}

\begin{lemma}
\label{lem:prove-nil}
Let $t$ be a $\FCCS^-$ term. Then $t \bis \nil$ if, and only if, the equation $t\approx \nil$ is provable using axioms A4 and F1 in Table~\ref{tab:axioms} from left to right.
\end{lemma}

\begin{proof}
The ``if'' implication is an immediate consequence of the soundness of the equations A4 and F1 with respect to $\bis$. 
To prove the  ``only if'' implication, define, first of all, the collection $\text{NIL}$ of $\FCCS^-$ terms as the set of terms generated by the following grammar:
\[
t::= \nil \bfmid t+t \bfmid  f(t, u) \enspace , 
\]
where $u$ is an arbitrary $\FCCS^-$ term. We claim that each $\FCCS^-$ term $t$ is bisimilar to $\nil$ if, and only if, $t\in\text{NIL}$.
Using this claim and structural induction on $t\in\text{NIL}$, it is a simple matter to show that if $t\bis\nil$, then $t\approx \nil$ is provable using axioms A0 and F0 from left to right, which was to be shown.
  
To complete the proof, it therefore suffices to show the above claim. 
To establish the ``if'' implication one proves, using structural induction on $t$ and the congruence properties of bisimilarity, that if $t\in\text{NIL}$, then $\sigma(t) \bis \nil$ for every closed substitution $\sigma$.  
To show the ``only if'' implication, we establish the contrapositive statement, viz.~that if $t\not\in\text{NIL}$, then $\sigma(t) \nbis \nil$ for some closed substitution $\sigma$. 
To this end, it suffices only to show, using structural induction on $t$, that if $t\not\in\text{NIL}$, then $\sigma_a(t)\mv{\mu}$ for some action $\mu\in\Act$, where $\sigma_a$ is the closed substitution mapping each variable to the closed term $a\nil$. 
The details of this argument are not hard, and are therefore omitted.
\end{proof}

In light of the above result, in the technical developments to follow, when dealing with an operator $f$ that distributes over $+$ in its first argument we shall assume, without loss of generality, that each axiom system we consider includes the equations A1--A4, F0--F2 in Table~\ref{tab:axioms}.
This assumption means, in particular, that our axiom systems will allow us to identify each term that is bisimilar to $\nil$ with $\nil$.

It is well-known (see, e.g., Sect.~2 in~\cite{Gr90}) that if an equation relating two closed terms can be proved from an axiom system $\E$, then there is a closed proof for it.  
Moreover, if $\E$ satisfies a further closure property, called \emph{saturation}, in addition to those mentioned earlier, and that closed equation relates two terms containing no occurrences of $\nil$ as a summand or factor, then there is a closed proof for it in which all of the terms have no occurrences of $\nil$ as a summand or factor.

\begin{definition}
\label{Defn:nf}
For each $\FCCS^-$ term $t$, we define ${t}/\nil$ thus: 
\[
\begin{array}{lclcl}
{\nil}/\nil = \nil & \qquad\qquad &
{x}/\nil = x & \qquad\qquad & 
{\mu t}/\nil = \mu ({t}/\nil)
\end{array}
\]
\[
\begin{array}{lcl}
({t+u})/\nil = 
\begin{cases}
{u}/\nil & \text{if $t\bis \nil$} \\
{t}/\nil & \text{if $u\bis \nil$} \\
({t}/\nil)+ ({u}/\nil) & \text{otherwise}
\end{cases} 
& \qquad \quad \qquad &
f(t, u)/\nil =  
\begin{cases}
\nil & \text{if $t\bis \nil$} \\
{t}/\nil & \text{if $u\bis \nil$} \\
f({t}/\nil, {u}/\nil) & \text{otherwise}
\end{cases}  
\end{array}
\]
\end{definition}

Intuitively, $t /\nil$ is the term that results by removing {\sl all} occurrences of $\nil$ as a summand or factor from $t$. 

The following lemma, whose simple proof by structural induction on terms is omitted, collects the basic properties of the above construction. 

\begin{lemma}
\label{Lem:-nil}
For each $\FCCS^-$ term $t$, the following statements hold:
\begin{enumerate}
\item \label{nil1} 
the equation $t \approx t /\nil$ can be proven using the equations A1--A4, F0--F2, and therefore $t \bis t  /\nil$;
\item \label{nil2} 
the term $t /\nil$ has no occurrence of $\nil$ as a summand or factor;
\item \label{nil3} 
$t /\nil = t$, if $t$ has no occurrence of $\nil$ as a summand or factor;
\item \label{nil4} 
$\sigma(t /\nil)/\nil = \sigma(t)/\nil$, for each substitution $\sigma$.
\end{enumerate}
\end{lemma}

\begin{definition}
\label{Defn:0subs}
We say that a substitution $\sigma$ is a {\sl $\nil$-substitution} if{f} $\sigma(x)\neq x$ implies that $\sigma(x)=\nil$, for each variable $x$.
\end{definition}

\begin{definition}
\label{Defn:closure}
Let $\E$ be an axiom system. We define the axiom system $\chiudi(\E)$ thus:
\[
\chiudi(\E) = \E \cup \{ \sigma(t)/\nil \approx \sigma(u)/\nil \mid (t\approx u)\in \E,~\sigma~\text{a $\nil$-substitution} \} \enspace .
\]
An axiom system $\E$ is {\sl saturated} if $\E =\chiudi(\E)$.
\end{definition}

The following lemma collects some basic sanity properties of the closure operator $\chiudi(\cdot)$. 
(Note that the application of $\chiudi(\cdot)$ to an axiom system preserves closure with respect to symmetry.)

\begin{lemma}
\label{Lem:properties-closure}
Let $\E$ be an axiom system. Then the following statements hold. 
\begin{enumerate}
\item \label{cl-cl} $\chiudi(\E) = \chiudi(\chiudi(\E))$. 
\item \label{cl-finite} $\chiudi(\E)$ is finite, if so is $\E$.
\item $\chiudi(\E)$ is sound, if so is $\E$.
\item $\chiudi(\E)$ is closed with respect to symmetry, if so is $\E$.
\item \label{cl-equi} $\chiudi(\E)$ and $\E$ prove the same equations, if $\E$ contains the equations A1--A4, F0--F2.
\end{enumerate}
\end{lemma} 

\begin{proof}
We limit ourselves to sketching the proofs of statements~\ref{cl-cl} and~\ref{cl-equi} in the lemma. 
  
In the proof of statement~\ref{cl-cl}, the only non-trivial thing to check is that the equation
\[
\sigma(\sigma'(t)/\nil))/\nil \approx \sigma(\sigma'(u)/\nil))/\nil 
\]
is contained in $\chiudi(\E)$, whenever $(t \approx u)\in \E$ and $\sigma,\sigma'$ are $\nil$-substitutions. 
This follows from Lemma~\ref{Lem:-nil}(\ref{nil4}) because the collection of $\nil$-substitutions is closed under composition.
  
To show statement~\ref{cl-equi}, it suffices only to argue that each equation $t\approx u$ that is provable from $\chiudi(\E)$ is also provable from $\E$, if $\E$ contains the equations A1--A4, F0--F2.
This can be done by induction on the depth of the proof of the equation $t\approx u$ from $\chiudi(\E)$, using Lemma~\ref{Lem:-nil}(\ref{nil1}) for the case in which $t\approx u$ is a substitution instance of an axiom in $\chiudi(\E)$.
\end{proof}

In light of this result, the saturation of a finite axiom system that includes the equations A1--A4, F0--F2 results in an equivalent, finite collection of equations (Lemma~\ref{Lem:properties-closure}(\ref{cl-finite}) and~(\ref{cl-equi})).

We are now ready to state our counterpart of~\cite[Proposition~5.1.5]{Mo89}.

\begin{proposition}
\label{Propn:proofswithout0}
Assume that $\E$ is a saturated axiom system. 
Suppose furthermore that we have a closed proof from $\E$ of the closed equation $p\approx q$. 
Then replacing each term $r$ in that proof with $r/\nil$ yields a closed proof of the equation $p/\nil\approx  q/\nil$. 
In particular, the proof from $\E$ of an equation $p\approx q$, where $p$ and $q$ are terms not containing occurrences of $\nil$ as a summand or factor, need not use terms containing occurrences of $\nil$ as a summand or factor.
\end{proposition}

\begin{proof}
The proof follows the lines of that of~\cite[Proposition~5.1.5]{Mo89}, and is therefore omitted.
\end{proof}

In light of Proposition~\ref{Propn:proofswithout0}, henceforth, when dealing with an operator $f$ that distributes over $+$ in one of its arguments, we shall limit ourselves to considering saturated axiom systems.


\subsubsection{$\nil$ absorption for a non distributive $f$}
\label{app:nil_absorption}

In Section~\ref{sec:no_dist}, we argued that the set of allowed rules for an operator $f$ that does not distribute over summation in either argument has to include at least a rule of type (\ref{asyncruleleft}) {\sl and} at least one of type (\ref{asyncruleright}).
We also notice that for an operator $f$ having both types of rules for all actions we can distinguish two cases, according to which rules of type (\ref{syncrule}) are available:
\begin{inparaenum}
\item If $f$ has both rules of type (\ref{syncrule}), then it would be a mere rewriting of the parallel composition operator (see the proof of Lemma~\ref{Lem:asyncrules}).
\item If $f$ has only one rule of type (\ref{syncrule}), then one can observe that Moller's argument to the effect that bisimilarity is not finitely based over the fragment of CCS with action prefixing, nondeterministic choice and purely interleaving parallel composition, could be applied to $f$, yielding the desired negative result. 
\end{inparaenum}

Hence, we can assume that there is an action $\mu\in\Act$ such that $f$ has only one rule, of type either (\ref{asyncruleleft}) or (\ref{asyncruleright}), with $\mu$ as label.
This asymmetry in the set of rules for $f$ can cause some $\FCCS^-$ term to behave as $\nil$ when occurring in the scope of $f$, despite not being bisimilar to $\nil$ at all.

\begin{example}
\label{ex:blocking_f}
Consider the term $t = f(a + \bar{a}.u, \tau)$, for some term $u$, and assume that $f$ has only rules of type (\ref{asyncruleleft}) with labels $a$ and $\tau$ and only a rule of type (\ref{asyncruleright}) with label $\bar{a}$.
One can easily check that, since the initial execution of the $\tau$-move in the second argument is prevented by the rules for $f$, then the subterm $\bar{a}.u$ can never contribute to the behaviour of $t$.
Thus, $t \bis a.\tau$, even though $\bar{a}.u \nbis \nil$ for each term $u$.
\end{example}

From a technical point of view, this implies that Lemmas~\ref{lem:prove-nil} and~\ref{Lem:-nil}.\ref{nil1} no longer hold.
In fact, one can always construct a term $t$ of the form $t = f(\sum_{i=1}^n \mu.x_i, \sum_{j = 1}^m \nu.y_j)$ for some $n,m\ge 0$, with $\mu,\nu$ chosen according to the available set of rules for $f$, such that $t \bis \nil$.
We conjecture that \emph{since we are considering an operator $f$ that does not distribute over summation in either of its arguments, the valid equations, modulo bisimilarity, of the form $t \approx \nil$ cannot be proved by means of any finite, sound set of axioms.}
Roughly speaking, this is due to the fact that no valid axiom can be established for a term of the form $f(\mu.x + z, \nu.y + w)$ in that the behaviour of the terms substituted for the variables $z$ and $w$ is crucial to determine that of a closed instantiation of the term.

Summarizing, this would imply that, in the case at hand, we cannot assume that we can use saturation to simplify the axiom systems and, moreover, the family of equations 
\[
f(\sum_{i = 1}^n \mu.p_i, \sum_{j=1}^m \nu.q_j) \approx \nil \qquad n,m\ge 0
\]
for some processes $p_i,q_j$, could play the role of witness family of equations for our desired negative result.
Unfortunately, the presence of two summations would force us to introduce a number of additional technical results that would make the proof of the negative results even heavier than it already is.
Moreover, those supplementary results are not necessary to treat the case of the witness families that we are going to introduce in Sections~\ref{sec:LaRa}--\ref{sec:Lt} to obtain the proof of Theorem~\ref{Thm:nonfin-en}.


\section{Preliminary results}
\label{sec:prelim}

As briefly outlined in Section~\ref{sec:proof_method}, to obtain the desired negative results we will proceed by a case analysis on the actual set of rules that are available for operator $f$.
However, there are a few preliminary results that hold for all the allowed behaviours of $f$ and that will be useful in the upcoming proofs.
We devote this section to presenting such results and notions.


\subsection{Unique prime decomposition}
\label{sec:prime}

In the proof of our main results, we shall often make use of some notions from~\cite{MM93,Mo89}. 
These we now proceed to introduce for the sake of completeness and readability.

\begin{definition}
\label{Def:prime}
A closed term $p$ is {\sl irreducible} if $p \bis q \mathbin{\|} r$ implies $q\bis \nil$ or $r\bis \nil$, for all closed terms $q,r$.
We say that $p$ is {\sl prime} if it is irreducible and is not bisimilar to $\nil$.
\end{definition}

Note that each process $p$ of depth $1$ is prime, as every term of the form $q\|r$ (without $\nil$ factors) has depth at least $2$, and cannot be thus bisimilar to $p$.
Similarly, each process of norm $1$ is prime. 

The following lemma states the primality of two families of closed terms that will play a key role in the proof of our main result.

\begin{lemma}
\label{lem:rhs-prime_1}
\begin{enumerate}
\item \label{claim:alpha<=n} 
The term $\mu^{\scriptstyle \leq m}$ is prime, for each $m\geq 1$.
\item \label{claim:pn} 
Let $\nu \in \{a,\bar{a}\}$, $\mu \in \Act$, $\nu \neq \mu$, $m\geq 1$ and $1\leq i_1 < \ldots < i_m$.
Then the term $\nu.\mu^{\scriptstyle \leq i_1} + \cdots + \nu.\mu^{\scriptstyle \leq i_m}$ is prime.  
\end{enumerate}
\end{lemma}

\begin{proof}
The first claim is immediate because the norm of $\mu^{\scriptstyle \leq m}$ is one, for each $m\geq 1$.

For the second claim, assume by contradiction that there are process terms $p,q$ such that $p,q\nbis\nil$ and $\nu.\mu^{\scriptstyle \leq i_1} + \cdots + \nu.\mu^{\scriptstyle \leq i_m} \bis p \| q$.
Clearly, this would imply the existence of process terms $p',q'$ such that $p \trans[\nu] p'$ and $q \trans[\nu] q'$ so that $p \| q \trans[\nu] p'\| q$ and $p \| q \trans[\nu] p \| q'$.
However, these transitions would in turn imply that $p \| q \trans[\nu] p' \| q \trans[\nu] p' \| q'$, namely $p \| q$ could perform two $\nu$-moves in a row, whereas $\nu.\mu^{\scriptstyle \leq i_1} + \cdots + \nu.\mu^{\scriptstyle \leq i_m}$ cannot perform such a sequence of actions, thus contradicting $\nu.\mu^{\scriptstyle \leq i_1} + \cdots + \nu.\mu^{\scriptstyle \leq i_m} \bis p \| q$.
\end{proof}

In~\cite{MM93} the notion of \emph{unique prime decomposition} of a process $p$ was introduced, as the unique multiset $\{\!|\,q_1,\dots,q_n\,|\!\}$ of primes s.t.\ $p \bis q_1 \| \dots \| q_n$.
Inspired by the unique prime decomposition result of~\cite{MM93}, the authors of~\cite{LvO05} proposed the notion of \emph{decomposition order for commutative monoids}, and proved that the existence of a decomposition order on a commutative monoid implies that the monoid has the unique prime decomposition property.
$\FCCS$ modulo $\bistext$ is a \emph{commutative monoid} with respect to $\parallel$, having $\nil$ as unit, and the transition relation defines a \emph{decomposition order} over bisimilarity equivalence classes of closed terms.
Then, by~\cite[Theorem 32]{LvO05}, the following result holds:

\begin{proposition}
\label{prop:unique_decomposition}
Any $\FCCS$ term can be expressed uniquely, up to $\bistext$, as a parallel composition of primes.
\end{proposition}

As we will see, this property will play a crucial role in some of the upcoming proofs.


\subsection{Decomposing the semantics of terms}
\label{sec:open}

In the proofs to follow, we shall sometimes need to establish a correspondence between the behaviour of open terms and the semantics of their closed instances, with a special focus on the role of variables.
In detail, we need to consider the possible origins of a transition of the form $\sigma(t)\mv{\alpha}p$, for some action $\alpha\in\{a,\bar{a}\}$, closed substitution $\sigma$, $\FCCS^-$ term $t$ and closed term $p$. 
In fact, the equational theory is defined over process terms, whereas the semantic properties can be verified only on their closed instances.

\begin{lemma}
\label{lem:substitution}
Let $\mu \in \Act$.
Then, for all $t,t'$ and substitutions $\sigma$, if $t \trans[\mu] t'$ then $\sigma(t) \trans[\mu] \sigma(t')$.
\end{lemma}

However, a transition $\sigma(t) \trans[\mu] p$ may also derive from the initial behaviour of some closed term $\sigma(x)$, provided that the collection of initial moves of $\sigma(t)$ depends, in some formal sense, on that of the closed term substituted for the variable $x$. 
Roughly speaking, our aim is now to provide the conditions under which $\sigma(t) \trans[\mu] p$ can be inferred from $\sigma(x) \trans[\nu] q$, for some $\mu,\nu \in\Act$ and processes $p,q$.
As one might expect, in our setting the provability of transitions needs to be parametric with respect to the rules for $f$.

\begin{example}
\label{ex:open_running_1}
Consider the $\FCCS^-$ term $t = f(x,\tau)$.
Firstly, we notice that if $R^{\ff}_{\tau}$ holds then we can infer that $\sigma(t) \trans[\tau] \sigma(x) \| \nil$ for all closed substitutions $\sigma$.
Assume now that $\sigma(x) = a$.
Clearly, we can derive $\sigma(t) \trans[a] \nil \| \tau$ only if $L^{\ff}_{a}$ holds.
\end{example}

To fully describe this situation, for each $\mu \in \Act$, we introduce the auxiliary transition relation $\trans[]_{\mu}$ over open terms.
To this end, we present the notion of {\sl configuration} over $\FCCS^-$ terms, which stems from~\cite{AFIN06}.
Configurations are terms defined over a set of variables $\DVar= \{x_d \mid x \in \Var\}$, disjoint from $\Var$, and $\FCCS^-$ terms.
Intuitively, the symbol $x_d$ (read ``$x$ derivative'') will be used to denote that the closed term substituted for an occurrence of variable $x$ has begun its execution.

\begin{definition}
The collection of {\sl $\FCCS^-$ configurations} is given by the following grammar:
\[
c :: = t \bfmid x_d \bfmid c \mathbin{\|} t \bfmid t \mathbin{\|} c \enspace ,
\]
where $t$ is a $\FCCS^-$ term, and $x_d \in \Var_d$. 
\end{definition}

For example, the configuration $x_d \mathbin{\|} f(a,x)$ is meant to describe a state of the computation of some term in which the (closed term substituted for the) occurrence of variable $x$ on the left-hand side of the $\|$ operator has begun its execution, but the one on the right-hand side has not.

We introduce also special labels for the auxiliary transitions $\trans[]_{\mu}$, to keep track of which rules for $f$ are available, and thus which one triggered the move by the closed instance of $x$.
In detail, we let $x_{\sx}$ denote that the closed instance of $x$ is responsible for the transition when $L^{\ff}_{\mu}$ holds.
In case $R^{\ff}_{\mu}$ holds, we use $x_{\rr}$.
Finally, $x_{\bb}$ is used when $L^{\ff}_{\mu} \wedge R^{\ff}_{\mu}$ holds. 

The auxiliary transitions of the form $\trans[]_{\mu}$ are then formally defined via the inference rules in Table~\ref{tab:aux_rules}.

\begin{table}[t]
\[
\begin{array}{ll}
\scalebox{0.9}{$(a_1)$}\; \SOSrule{L^{\ff}_{\mu}}
{x \trans[x_{\sx}]_{\mu} x_{\dd}}
\qquad\qquad
\scalebox{0.9}{$(a_2)$}\; \SOSrule{R^{\ff}_{\mu}}
{x \trans[x_{\rr}]_{\mu} x_{\dd}}
&
\qquad\qquad
\scalebox{0.9}{$(a_3)$}\; \SOSrule{L^{\ff}_{\mu} \wedge R^{\ff}_{\mu}}
{x \trans[x_{\bb}]_{\mu} x_{\dd}}
\\
\scalebox{0.9}{$(a_4)$}\; \SOSrule{t_1 \trans[x_{\ww}]_{\mu} c}
{t_1 + t_2 \trans[x_{\ww}]_{\mu} c} \; \ww \in \{\sx,\rr,\bb\}
&
\qquad\qquad
\scalebox{0.9}{$(a_5)$}\; \SOSrule{t_2 \trans[x_{\ww}]_{\mu} c}
{t_1 + t_2 \trans[x_{\ww}]_{\mu} c} \; \ww \in \{\sx,\rr,\bb\}
\\
\scalebox{0.9}{$(a_6)$}\; \SOSrule{t_1 \trans[x_{\sx}]_{\mu} c}
{f(t_1,t_2) \trans[x_{\sx}]_{\mu} c \| t_2}
&
\qquad\qquad
\scalebox{0.9}{$(a_7)$}\; \SOSrule{t_2 \trans[x_{\rr}]_{\mu} c}
{f(t_1,t_2) \trans[x_{\rr}]_{\mu} t_1 \| c}
\\
\scalebox{0.9}{$(a_8)$}\; \SOSrule{t_1 \trans[x_{\bb}]_{\mu} c}
{f(t_1,t_2) \trans[x_{\bb}]_{\mu} c \| t_2}
&
\qquad\qquad
\scalebox{0.9}{$(a_9)$}\; \SOSrule{t_2 \trans[x_{\bb}]_{\mu} c}
{f(t_1,t_2) \trans[x_{\bb}]_{\mu} t_1 \| c}
\end{array}
\]
\caption{\label{tab:aux_rules} The auxiliary rules.}
\end{table}

\begin{example}
\label{ex:open_running_2}
Consider the term $t = f(x,\tau)$ from Example~\ref{ex:open_running_1}.
Assume, for instance, that $L^{\ff}_{a}$ holds, yielding the transition $x \trans[x_{\sx}]_a x_{\dd}$, due to rule ($a_1$).
Then, an application of rule ($a_6$) would give $f(x,\tau) \trans[x_{\sx}]_a x_{\dd} \| \tau$ with the following meaning: since the rules for $f$ allow $a$-moves of the first argument to yield $a$-moves of terms of the form $f(p,q)$, then an $a$-transition by (an instance of) variable $x$ occurring in the first argument of $f$ will induce an $a$-move of $f(x,\tau)$.

Conversely, assume that only $R^{\ff}_a$ holds.
Then, by applying rule ($a_2$) we obtain that $x \trans[x_{\rr}]_a x_{\dd}$ and, from the rules, it is not possible to derive any $\trans[]_a$ transition of $f(x,\tau)$ from that of $x$, modelling the fact that the rules for $f$ prevent the execution of $a$-moves from the first argument. 
\end{example}

The structure of the targets of the auxiliary rules in Table~\ref{tab:aux_rules} can be characterised modulo $\bistext$.

\begin{lemma}
\label{lem:aux_target}
Let $t$ be a $\FCCS^-$ term, $x$ be a variable and $\ww \in \{\sx,\rr,\bb\}$.
If $t \trans[x_{\ww}]_\mu c$ can be inferred from the rules in Table~\ref{tab:aux_rules}, for some action $\mu$, then $c \bis x_{\dd} \mathbin{\|} t'$ for some $\FCCS$ term $t'$.
\end{lemma}

\begin{proof}
The proof follows by induction on the derivation of $t \trans[x_\ww]_\mu c$.
\end{proof}

Lemmas~\ref{lem:o2c} and~\ref{lem:c2o} below formalise the decomposition of the semantics of $\FCCS^-$ terms.
We recall that, due to Lemma~\ref{Lem:asyncrules}, at least one between $L^{\ff}_{\mu}$ and $R^{\ff}_{\mu}$ holds for each $\mu$.

\begin{lemma}
\label{lem:o2c}
Let $\mu \in \Act$, $t$ be a $\FCCS^-$ term, $x$ be a variable, $\ww \in \{\sx,\rr,\bb\}$ and $\sigma$ be a closed substitution.
If $\sigma(x)\trans[\mu] p$ for some process $p$, and $t \trans[x_{\ww}]_{\mu} c$ for some configuration $c$, then $\sigma(t) \trans[\mu] \sigma[x_{\dd} \mapsto p](c)$. 
\end{lemma}

\begin{proof}
The proof follows by induction on the structure of $t$ and the derivation of $t \trans[x_{\ww}]_{\mu} c$.
\end{proof}

\begin{lemma}
\label{lem:c2o}
Let $\alpha \in \{a,\bar{a}\}$, $t$ be a $\FCCS^-$ term, $\sigma$ be a closed substitution and $p$ be a closed term.
Whenever $\sigma(t) \trans[\alpha] p$, then one of the following holds:
\begin{enumerate}
\item \label{lem:c2o_prefix}
There is term $t'$ such that $t \trans[\alpha] t'$ and $\sigma(t') = p$.
\item There are a variable $x$, a process $q$ and a configuration $c$ such that:
\begin{enumerate}
\item \label{lem:c2o_l}
only $L^{\ff}_{\alpha}$ holds, $\sigma(x) \trans[\alpha] q$, $t \trans[x_{\sx}]_{\alpha} c$ and $\sigma[x_{\dd} \mapsto q] (c) = p$;
\item \label{lem:c2o_r}
only $R^{\ff}_{\alpha}$ holds, $\sigma(x) \trans[\alpha] q$, $t \trans[x_{\rr}]_{\alpha} c$ and $\sigma[x_{\dd} \mapsto q] (c) = p$; or
\item \label{lem:c2o_b}
$L^{\ff}_{\alpha} \wedge R^{\ff}_{\alpha}$ holds, $\sigma(x) \trans[\alpha] q$, $t \trans[x_{\bb}]_{\alpha} c$ and $\sigma[x_{\dd} \mapsto q] (c) = p$.
\end{enumerate}
\end{enumerate}
\end{lemma}

\begin{proof}
The proof is by induction on the structure of $t$.
The interesting case is the inductive step corresponding to $t = f(t_1,t_2)$, which we expand below.
According to which rules are available for $f$ with respect to $\alpha$, we can distinguish three cases:
\begin{enumerate}
\item {\sc Case only $L^{\ff}_{\alpha}$ holds.} 
Then, $f(\sigma(t_1),\sigma(t_2)) \trans[\alpha] p$ can be inferred only from a transition of the form $\sigma(t_1) \trans[\alpha] p'$ for some closed term $p'$ with $p = p' \| \sigma(t_2)$.
By induction over the derivation of $\sigma(t_1) \trans[\alpha] p'$, and considering that only $L^{\ff}_{\alpha}$ holds, we can then distinguish two cases:
\begin{itemize}
\item There is a term $t_1'$ such that $t_1 \trans[\alpha] t_1'$ and $\sigma(t_1') = p'$.
As $f$ has the rule of the form~\eqref{asyncruleleft} for $\alpha$ we can immediately infer that $t \trans[\alpha] t_1' \| t_2$.
Hence, by letting $t' = t_1' \| t_2$, we obtain $t \trans[\alpha] t'$ and $\sigma(t') = p$.
\item There are a variable $x$, a closed term $q$ and a configuration $c_1$ such that $\sigma(x) \trans[\alpha] q$, $t_1 \trans[x_{\sx}]_{\alpha} c_1$ with $\sigma[x_{\dd} \mapsto q](c_1) = p'$.
Hence, by applying the auxiliary rule $(a_6)$ we can infer that $f(t_1,t_2) \trans[x_{\sx}]_{\alpha} c_1 \| t_2$ and moreover, since $x_{\dd}$ may occur only in $c_1$, we have $p = p' \| \sigma(t_2) = \sigma[x_{\dd} \mapsto q](c_1 \| t_2)$.
\end{itemize} 

\item {\sc Case only $R^{\ff}_{\alpha}$ holds.}
This case is analogous to the previous one (it is enough to switch the roles of $t_1$ and $t_2$ and consider $x_{\rr}$ in place of $x_{\sx}$) and therefore omitted.

\item {\sc Case $L^{\ff}_{\alpha} \wedge R^{\ff}_{\alpha}$ holds.}
This case follows by noticing that $t \trans[x_{\bb}]_{\alpha}$ can be inferred from both $t_1 \trans[x_{\bb}]_{\alpha}$ and $t_2 \trans[x_{\bb}]_{\alpha}$, and therefore the follows from the structure of the previous two cases, using rules ($a_8$) and ($a_9$).
\end{enumerate}
\end{proof}

Next, we proceed to a more detailed analysis of the contribution of variables to the behaviour of closed instantiations of terms in which they occur.

\begin{lemma}
\label{lem:variables}
Let $t$ be a term in $\FCCS^-$, $\sigma$ be a closed substitution and $\alpha \in \{a,\bar{a}\}$.
Assume that 
$
\sigma(t) \bis \sum_{i = 1}^n \alpha.p_i + q
$
for some $n$ greater than the size of $t$ and closed terms $p_i,q$ with $p_i \nbis  p_j$ whenever $i \neq j$.
Then $t$ has a summand $x$, for some variable $x$, such that
$
\sigma(x) \bis \sum_{j \in J} \alpha.q_j + q'
$
for some $J \subseteq \{1,\dots,n\}$, with $|J| \ge 2$, some process $q'$, and processes $q_j$ such that:
\begin{itemize}
\item $q_j \nbis q_k$ whenever $j \neq k$.
\item Either $q_j \bis p_j$, for each $j \in J$, or there is a process $r$ such that $p_j \bis q_j \mathbin{\|} r$, for each $j \in J$.
\end{itemize}
\end{lemma}

\begin{proof}
For simplicity of notation let $I = \{1,\dots,n\}$. 
Since there is a transition $\sum_{i \in I} \alpha.p_i + q \trans[\alpha] p_i$ for each $i \in I$, from $\sigma(t) \bis \sum_{i \in I} \alpha.p_i + q$ we get that $\sigma(t) \trans[\alpha] r_i$ with $r_i \bis p_i$, for all $i \in I$.
Since $n$ is greater than the size of $t$, we infer that Lemma~\ref{lem:c2o}.\ref{lem:c2o_prefix} can be applied only to $m$ such transitions, for some $m <n$, so that there are an index set $H \subset I$ (possibly empty) and $\FCCS$ terms $t_h$, for $h \in H$ such that $|H| = m$, $t \trans[\alpha] t_h$ and $\sigma(t_h) \bis p_h$.
Notice that since $p_i \nbis p_j$ for $i \neq j$ we get that the $t_h$ are pairwise distinct.
Let $J = I \setminus H$. 
For the remaining $\alpha$-transitions $\sigma(t) \trans[\alpha] r_j$ for $j \in J$ we have that one among cases~\ref{lem:c2o_l}--\ref{lem:c2o_b} of Lemma~\ref{lem:c2o} applies, according to which rules are available for $f$ with respect to action $\alpha$.
Hence, we have that, for each $j \in J$ there are a variable $x_j$, a closed term $q_j$ and a configuration $c_j$ such that $\sigma(x_j) \trans[\alpha] q_j$, $t \trans[x_{j,\ww}]_{\alpha} c_j$ and $\sigma[x_{j,\dd} \mapsto q_j](c_j) =r_j$, where $\ww \in \{\sx,\rr,\bb\}$ depends on the rules for $f$.
Since $n$ is greater than the size of $t$ there cannot be more than $|J|-1$ distinct variables $x_j$ occurring in $t$ and causing such $\alpha$-moves (actually the constraint is also on the number of \emph{distinct occurrences} of variables in $t$).
Hence, there is at least one variable $x \in \var(t)$ such that $\sigma(x) \bis \alpha.q_{j_1} + \alpha.q_{j_2} + q'$ for some $j_1 \neq j_2 \in J$, $q_{j_1} \nbis q_{j_2}$, and closed term $q'$.  

Let $c$ be the configuration such that $t \trans[x_{\ww}]_{\alpha} c$.
By Lemma~\ref{lem:aux_target} we have that, $c \bis x_{\dd} \mathbin{\|} t'$, for some term $t'$.
Moreover from the analysis carried out above, we have that $\sigma[x_{\dd} \mapsto q_{j_1}](c) \bis q_{j_1} \mathbin{\|} \sigma(t') \bis p_{j_1}$, and $\sigma[x_{\dd} \mapsto q_{j_2}](c) \bis q_{j_2} \mathbin{\|} \sigma(t') \bis p_{j_2}$.
In particular, if $\sigma(t') \bis \nil$ we get that $q_{j_1} \bis p_{j_1}$ and $q_{j_2} \bis p_{j_2}$.
\end{proof}

The next result shows a particular case of Lemma~\ref{lem:variables}, in which we can infer that, provided the term $t$ has only one summand and has neither $\nil$ summands nor factors, not only is a variable $x$ responsible for the additional behaviour of $t$, but that $t$ coincides with $x$. 

\begin{lemma}
\label{Lem:vh-claim_gen}
Let $t$ be a term in $\FCCS^-$ that does not have $+$ as head operator, and let $\sigma$ be a closed substitution. 
Let $\alpha \in \{a,\bar{a}\}$ and $\mu \in \Act$ with $\alpha \neq \mu$.
Assume that $\sigma(t)$ has neither $\nil$ summands nor factors, and that
$
\sigma(t) \bis \alpha.\mu^{\scriptstyle \le i_1} + \cdots + \alpha.\mu^{\scriptstyle \le i_m},
$
for some $m>1$ and $1\le i_1 < \ldots < i_m$. 
Then $t=x$, for some variable $x$.
\end{lemma}

\begin{proof}
Assume, towards a contradiction, that $t$ is not a variable. 
We proceed by a case analysis on the possible form this term may have.
\begin{enumerate}
\item 
\label{Case:t-pre-vh_gen} 
{\sc Case $t = \nu.t'$ for some term $t'$}. 
Then $\nu=\alpha$ and $\mu^{\scriptstyle \le i_1} \bis \sigma(t') \bis \mu^{\scriptstyle \le i_m}$. 
However, this is a contradiction because, since ${i_1} < {i_m}$, the terms $\mu^{\scriptstyle \le i_1}$ and $\mu^{\scriptstyle \le i_m}$ have different depths, and are therefore not bisimilar.

\item \label{Case:t-hmerge-vh_gen} 
{\sc Case $t = f(t',t'')$ for some terms $t',t''$}. Since $\sigma(t)$ has no $\nil$ factors, we have that $\sigma(t')\nbis \nil$ and $\sigma(t'')\nbis \nil$.
Moreover, observe that $\alpha.\mu^{\scriptstyle \le i_1} + \alpha.\mu^{\scriptstyle \le i_m} \mv{\alpha} \mu^{\scriptstyle \le i_m}$.  

Thus, as
$
\sigma(t)= f(\sigma(t'),\sigma(t'')) \bis \alpha.\mu^{\scriptstyle \le i_1} + \cdots + \alpha.\mu^{\scriptstyle \le i_m}
$,
according to which rules are available for $f$ with respect to $\nu$, we can distinguish the following two cases:
\begin{itemize}
\item $L^{\ff}_{\alpha}$ holds and there is a term $p'$ such that 
$
\sigma(t') \mv{\alpha} p' \text{ and } p' \| \sigma(t'') \bis \mu^{\scriptstyle \le i_m}
$.
As $\sigma(t'')\nbis \nil$ and $\mu^{\scriptstyle \le i_m}$ is prime (Lemma~\ref{lem:rhs-prime_1}(\ref{claim:alpha<=n})), this implies that $p'\bis \nil$ and
$
\sigma(t'') \bis \mu^{\scriptstyle \le i_m}
$.
Since $\alpha.\mu^{\scriptstyle \le i_1} + \cdots + \alpha.\mu^{\scriptstyle \le i_m} \mv{\alpha} \mu^{\scriptstyle \le i_1}$, a similar reasoning allows us to conclude that 
$
\sigma(t'') \bis \mu^{\scriptstyle \le i_1} 
$
also holds. 
However, this is a contradiction because by the proviso of the lemma $m>1$ and $1\le i_1 < \ldots < i_m$, and therefore $\mu^{\scriptstyle \le i_1}$ and
$\mu^{\scriptstyle \le i_m}$ are not bisimilar.
\item $R^{\ff}_{\alpha}$ holds and there is a term $p''$ such that
$
\sigma(t'') \mv{\alpha} p'' \text{ and } \sigma(t') \| p'' \bis \mu^{\scriptstyle \le i_m}
$.
This case is analogous to the previous one and leads as well to a contradiction.
\end{itemize}
\end{enumerate}
We may therefore conclude that $t$ must be a variable, which was to be shown.
\end{proof}

We can now establish whether some of the initial behaviour of two bisimilar terms is determined by the same variable (Proposition~\ref{prop:x_in_tu}).

We start by arguing that we can also give a syntactic characterisation of the occurrences in a term of the variables that can contribute to the behaviour of closed instances of that term.
Formally, to infer the behaviour of a term $t$ from that of (a closed instance of) a variable $x$, the latter must occur \emph{unguarded} in $t$, namely $x$ cannot occur in the scope of a prefixing operator in $t$.
Inspired by~\cite{ACILvdP11}, for $\mu \in \Act$ and $\ww \in \{\sx,\rr,\bb\}$, we introduce a relation $\trt{\ww}{\mu}$ between a variable $x$ and a term $t$.
Intuitively, the role of the label $\ww$ is the same as in the auxiliary transitions, namely, to identify which predicates hold for $f$ (and thus which rules for $f$ are available) with respect to action $\mu$.
Then $x \trt{\ww}{\mu} t$ holds if the predicate associated with $\ww$ holds for $f$, and whenever $t$ has a subterm of the form $f(t_1,t_2)$ and $x$ occurs in $t_i$ (with $i=1$ if $\ww \in\{\sx,\bb\}$ and $i=2$ if $\ww \in \{\rr,\bb\}$) then the occurrence of $x$ is unguarded and can contribute to an initial $\mu$-transition of $\sigma(t)$ when $\sigma(x) \trans[\mu]$.

\begin{definition}
[Relation $\trt{}{}$]
\label{def:trt}
Let $\mu \in \Act$ and $\ww \in \{\sx,\rr,\bb\}$.
The relation $\trt{\ww}{\mu}$ between variables and terms is defined inductively as follows:
\[
\begin{array}{ll}
\scalebox{0.9}{1.}\; x \trt{\sx}{\mu} x \text{ if } L^{\ff}_{\mu}
\qquad\qquad
\scalebox{0.9}{2.}\; x \trt{\rr}{\mu} x \text{ if } R^{\ff}_{\mu}
& \quad
\scalebox{0.9}{3.}\; x \trt{\bb}{\mu} x \text{ if } L^{\ff}_{\mu}\wedge R^{\ff}_{\mu}
\\[.1cm]
\scalebox{0.9}{4.}\; x \trt{\ww}{\mu} t \;\Rightarrow\; x \trt{\ww}{\mu} t + u \; \wedge \; x \trt{\ww}{\mu} u + t  
& \quad
\scalebox{0.9}{5.}\; x \trt{\sx}{\mu} t \;\Rightarrow\; x \trt{\sx}{\mu} f(t,u) 
\\[.1cm]
\scalebox{0.9}{6.}\; x \trt{\rr}{\mu} t \;\Rightarrow\; x \trt{\rr}{\mu} f(u,t)
& \quad
\scalebox{0.9}{7.}\; x \trt{\bb}{\mu} t \;\Rightarrow\; x \trt{\bb}{\mu} f(t,u) \;\wedge\; x \trt{\bb}{\mu} f(u,t).
\end{array}
\]
\end{definition}

\begin{example}
\label{ex:open_running_3}
Assume, for instance, that $L^{\ff}_a$, $R^{\ff}_{\bar{a}}$ and $L^{\ff}_{\tau} \wedge R^{\ff}_{\tau}$ are the only predicates holding.
Then, for $t=f(x,\tau)$ we have that $x \trt{\sx}{a} t$, $x \trt{\sx}{\tau} t$ and $x \trt{\bb}{\tau} t$.
\end{example}

There is a close relation between unguarded occurrences of variables in terms and the auxiliary transitions, as stated in the following:

\begin{lemma}
\label{lem:trt_open}
Let $\mu \in \Act$ and $\ww\in\{\sx,\rr,\bb\}$.
Then $x \trt{\ww}{\mu} t$ if and only if $t \trans[x_{\ww}]_{\mu} c$ for a configuration $c \bis x_{\dd} \| t'$ for some $\FCCS^-$ term $t'$. 
\end{lemma}

\begin{proof}
We prove the two implications separately.

($\Rightarrow$) We proceed by induction over the structure of $t$.
The interesting case is the inductive step corresponding to $t = f(t_1,t_2)$ which we expand below, by distinguishing three cases, according to which rules for $f$ are available:
\begin{itemize}
\item $x \trt{\sx}{\mu} f(t_1,t_2)$.
This can only be due to $x \trt{\sx}{\mu} t_1$.
By the induction hypothesis for $t_1$, this implies that $t_1 \trans[x_{\sx}]_{\mu} c_1$ with $c_1 \bis x_{\dd} \| t_1'$ for some $t_1'$.
By applying the auxiliary rule ($a_6$), we infer $f(t_1,t_2) \trans[x_{\sx}]_{\mu} c$ with $c = c_1 \| t_2$ and, since $\bis$ is a congruence with respect to $\|$ and $\|$ is associative with respect to $\bis$, we get $c \bis (x_{\dd} \| t_1') \| t_2 \bis x_{\dd} \| t'$ with $t' \bis t_1' \| t_2$.
\item $x \trt{\rr}{\mu} f(t_1,t_2)$.
This can only be due to $x \trt{\rr}{\mu} t_2$.
Thus, we can proceed as above, by applying the auxiliary rule ($a_7$) in place of rule ($a_6$) and using the commutativity of $\|$ modulo $\bis$.
\item $x \trt{\bb}{\mu} f(t_1,t_2)$.
This can be due to either $x \trt{\bb}{\mu} t_1$ or $x \trt{\bb}{\mu} t_2$.
For both, we can proceed as above, by applying the auxiliary rules ($a_8$) or, respectively, ($a_9$) in place of rules ($a_6$) and ($a_7$).
\end{itemize}

($\Leftarrow$) We proceed by induction over the derivation of the open transition $t \trans[x_{\ww}]_{\mu} c$.
Again, the interesting case is the inductive step corresponding to $t = f(t_1,t_2)$, which we expand below by considering three cases, according to which rules are available for $f$:
\begin{itemize}
\item $f(t_1,t_2) \trans[x_{\sx}]_{\mu} c$ with $c \bis x_{\dd} \| t'$ for some $t'$.
According to the auxiliary operational semantics, it must be the case that $t_1 \trans[x_{\sx}]_{\mu} c_1$ for some $c_1$ such that $c = c_1 \| t_2$. 
Notice that since $x_{\dd}$ can occur only in $c_1$, from $c = c_1 \| t_2$ and $c \bis x_{\dd} \| t'$, we infer $c_1 \bis x_{\dd} \| t''$ for some $t''$ such that $t'' \| t_2 \bis t'$.
Hence, we can apply the induction hypothesis to the transition from $t_1$ and obtain $x \trt{\sx}{\mu} t_1$.
Since $t = f(t_1,t_2)$ we can immediately conclude that $x \trt{\sx}{\mu} t$.
\item The cases of $f(t_1,t_2) \trans[x_{\rr}]_{\mu} c$ and $f(t_1,t_2) \trans[x_{\bb}]_{\mu} c$ follow by a similar reasoning.
\end{itemize}
\end{proof}

We now discuss the necessary conditions to relate the depth of closed instances of a term to the depth of the closed instances of the variables occurring in it.

\begin{lemma}
\label{lem:that_one}
Let $t$ be a $\FCCS^-$ term and $\sigma$ be a closed substitution.
If $t$ has no $\nil$ summands or factors and $x \trt{\ww}{\mu} t$, for some $\ww \in \{\sx,\rr,\bb\}$ and $\mu \in \Act$ with $\init{\sigma(x)} \subseteq \{\mu \mid x \trt{\ww}{\mu}t \}$, then $\depth(\sigma(t)) \ge \depth(\sigma(x))$.
\end{lemma}

\begin{proof}
The proof proceeds by structural induction over $t$ and a case analysis over $\ww\in\{\sx,\rr,\bb\}$.
The interesting case is the inductive step corresponding to $t = f(t_1,t_2)$ which we expand below for the case of $\ww = \sx$.
The other cases can be obtained by applying a similar reasoning.

Moreover, always for sake of simplicity, assume that there is only one action $\mu$ such that $x \trt{\sx}{\mu} t$, so that $\init{\sigma(x)} = \{\mu\}$.
Once again, the general case can be easily derived from this one.
Notice that $\init{\sigma(x)} = \{\mu\}$ implies the existence of a closed term $q$ such that $\sigma(x) \trans[\mu] q$ and $\depth(\sigma(x)) = \depth(q) + 1$. 
We have that $x \trt{\sx}{\mu} f(t_1,t_2)$ can be derived only by $x \trt{\sx}{\mu} t_1$.
Hence, structural induction over $t_1$ gives $\depth(\sigma(t_1)) \ge \depth(\sigma(x))$.
Moreover, by Lemma~\ref{lem:trt_open} we obtain that $t_1 \trans[x_{\sx}]_{\mu} c_1$ for some $c_1 \bis x_{\dd} \| t'$ for some term $t'$.
Furthermore, $\sigma(x) \trans[\mu] q$ together with Lemma~\ref{lem:o2c} gives $\sigma(t_1) \trans[\mu] \sigma[x_{\dd} \mapsto q](c_1)$.
Then we can infer that $\sigma(t) \trans[\mu] \sigma[x_{\dd} \mapsto q](c_1) \| \sigma(t_2)$ $\bis q \| (\sigma(t') \| \sigma(t_2))$.
Therefore
\begin{align*}
\depth(\sigma(t)) \ge{} & 1 + \depth(q \| (\sigma(t') \| \sigma(t_2))) \\
={} & 1 + \depth(q) + \depth(\sigma(t') \| \sigma(t_2)) \\
\ge{} & 1 + \depth(q) = \depth(\sigma(x)).
\end{align*}
\end{proof}

\begin{example}
\label{ex:open_running_4}
We remark that, due to the potential asymmetry of the rules for $f$, the requirement on the set of initials of $\sigma(x)$ cannot be relaxed in any trivial way.
Consider, for instance, the term $t = f(x,\tau)$ from our running example and assume that the only predicates holding are $L^{\ff}_{\alpha}$, $L^{\ff}_{\tau}$ and $R^{\ff}_{\bar{\alpha}}$.
Notice that $x \trt{\sx}{\alpha} t$ and $x \trt{\sx}{\tau} t$.
Consider the closed substitution $\sigma$ with $\sigma(x) = \alpha + \tau + \bar{\alpha}.\alpha^n$, for some $n \ge 2$, so that $\{\alpha,\tau\} \subset \init{\sigma(x)} = \{\alpha,\tau,\bar{\alpha}\}$.
As $L^{\ff}_{\bar{\alpha}}$ and $R^{\ff}_{\tau}$ do not hold, the only inferable initial transitions for $\sigma(t)$ are those resulting from the $\alpha$-move and the $\tau$-move by $\sigma(x)$.
Thus, we get that $\depth(\sigma(t)) = 2$, whereas $\depth(\sigma(x)) \ge 3$.
This is due to the fact that the computation of $\sigma(x)$ starting with a $\bar{\alpha}$-move is \emph{blocked} by the rules for $f$ and, thus, it cannot contribute to the behaviour of $t$.
\end{example}

We can now proceed to prove the following:

\begin{proposition}
\label{prop:x_in_tu}
Let $\alpha \in \{a,\bar{a}\}$, $x$ be a variable and $t,u$ be $\FCCS^-$ with $t \bis u$ and such that neither $t$ nor $u$ has $\nil$ summands or factors.
If $x \trt{\ww}{\alpha} t$ for some $\ww\in \{\sx,\rr,\bb\}$, then $x \trt{\ww}{\alpha} u$.
In particular, if $x \trt{\ww}{\alpha} t$ because $t$ has a summand $x$, then so does $u$.
\end{proposition}

\begin{proof}
Observe, first of all, that since $t$ and $u$ have no $\nil$ summands or factors, by Remark~\ref{rmk:summands} we can assume that $t  =  \sum_{i\in I} t_i$ and $u  =  \sum_{j\in J} u_j$ for some finite non-empty index sets $I, J$, where none of the $t_i$ ($i\in I$) and $u_j$ ($j\in J$) has $+$ as its head operator, and none of the $t_i$ ($i\in I$) and $u_j$ ($j\in J$) have $\nil$ summands or factors.
Therefore, $x \trt{\ww}{\alpha} t$ implies that there is some index $i \in I$ such that $x \trt{\ww}{\alpha} t_i$.
We then proceed by a case analysis on the rules available for $f$.
Actually we expand only the case in which only $L^{\ff}_{\alpha}$ holds, as the other two cases, in which respectively only $R^{\ff}_{\alpha}$ holds, or $L^{\ff}_{\alpha} \wedge R^{\ff}_{\alpha}$ holds, can be obtained analogously.

Since only $L^{\ff}_{\alpha}$ holds, then it must be the case that $x \trt{\sx}{\alpha} t_i$.
By Lemma~\ref{lem:trt_open} we get that $t_i \trans[x_{\sx}]_{\alpha} c$ for some configuration $c$ with $c \bis x_{\dd} \| t'$ for some $t'$.
Let $n$ be greater than the sizes of $t$ and $u$, and consider the substitution $\sigma$ such that 
\[
\sigma(y) = 
\begin{cases}
\alpha \sum_{i =1}^n \bar{\alpha}\alpha^{\scriptsize \le i} & \text{ if } y = x \\
\nil & \text{ otherwise.}
\end{cases}
\]
For simplicity of notation, let $p_n = \sum_{i=1}^n \bar{\alpha}\alpha^{\scriptsize \le i}$.
Clearly $\sigma(x) \trans[\alpha] p_n$.
By Lemma~\ref{lem:o2c} we obtain that $\sigma(t_i) \trans[\alpha] p$ with $p = \sigma[x_{\dd} \mapsto p_n] (c)$ and, thus, $p \bis p_n \| \sigma(t')$.
As $t \bis u$ implies $\sigma(t) \bis \sigma(u)$,  we get that there is an index $j \in J$ such that $\sigma(u_j) \trans[\alpha] q$ for some $q \bis p_n \| \sigma(t')$.
As only $L^{\ff}_{\alpha}$ holds, by Lemma~\ref{lem:c2o} we can distinguish two cases:
\begin{itemize}
\item There are a variable $y$, a closed term $q'$ and a configuration $c'$ such that $\sigma(y) \trans[\alpha] q'$, $u_j \trans[y_{\sx}]_{\alpha} c'$ and $q = \sigma[y_{\dd} \mapsto q'](c')$.
Since $\sigma$ maps all variables but $x$ to $\nil$, we can directly infer that $y = x$, $q' = p_n$.
Moreover, as $p_n$ is prime and there is a unique prime decomposition of processes, we also infer that $c' \bis x_{\dd} \| u'$ for some $u'$ with $\sigma(u') \bis \sigma(t')$.
Consequently, by Lemma~\ref{lem:trt_open} we can conclude that $x \trt{\sx}{\alpha} u_j$ and thus $x \trt{\sx}{\alpha} u$ as required.

\item There is a term $u'$ such that $u_j \trans[\alpha] u'$ and $\sigma(u') \bis p_n \| \sigma(t')$.
We proceed to show that this case leads to a contradiction.
We distinguish two cases:
\begin{itemize}
\item $\sigma(t') \bis \nil$. Thus $\sigma(u') \bis p_n$ and we can rewrite $u' = \sum_{h \in H} v_h$ for some terms $v_h$ that do not have $+$ as head operator.
Moreover, since $u$ not having $\nil$ summands nor factors implies that neither $u_j$ no $u'$ have some, the same holds for all the $v_h$.
Since $n$ is larger than the size of $u$, and thus than that of $u'$, by Lemma~\ref{Lem:vh-claim_gen} $\sigma(u') \bis p_n$ implies that there is one index $h \in H$ such that $v_h = y$ for some variable $y$ and $\sigma(y) \bis \bar{\alpha}\alpha^{\scriptsize \le i_1} + \dots + \bar{\alpha} \alpha^{\scriptsize \le i_m}$ for some $m > 1$ and $1 \le i_1 < \dots < i_m \le n$.
However, by the choice of $\sigma$, all variables but $x$ are mapped to $\nil$, and moreover $\sigma(x) \nbis \bar{\alpha}\alpha^{\scriptsize \le i_1} + \dots + \bar{\alpha} \alpha^{\scriptsize \le i_m}$ thus contradicting $\sigma(u') \bis p_n$.
\item $\sigma(t') \nbis \nil$.
Consequently, $\sigma(t') \bis \sum_{h \in H} \mu_h q_h$ for some actions $\mu_h \in \Act$ and closed terms $q_h$.
We can therefore apply the expansion law for parallel composition obtaining
\[
\sigma(u') 
\bis 
p_n \| \sigma(t') 
\bis 
\sum_{i = 1}^n \bar{\alpha} (\alpha^{\scriptsize \le i} \| \sigma(t')) +
\sum_{h \in H} \mu_h (p_n \| q_h) + 
\sum_{i = 1,\dots,n \atop h \in H \text{ s.t. } \mu_h = \alpha} \tau (\alpha^{\scriptsize \le i} \| q_h).
\]
We notice that the first term in the expansion has size at least $n+1$ and therefore greater than the size of $u$ and in particular of $u'$.
Moreover $\alpha^{\scriptsize \le i} \| \sigma(t') \nbis \alpha^{\scriptsize \le j} \| \sigma(t')$ whenever $i \neq j$.
Therefore, by Lemma~\ref{lem:variables} there is a variable $y \in \var(u')$ such that $\sigma(y) \bis \bar{\alpha} q_{i_1} + \dots + \bar{\alpha} q_{i_m} + r$ for some $m > 1$ and $1 \le i_1 < \dots < i_m$, closed term $r$ and proper closed terms $q_{i_1}, \dots, q_{i_m}$ according to Lemma~\ref{lem:variables} (their exact form is irrelevant to our purposes).
However, $\sigma(y) = \nil$ whenever $y \neq x$ and $\sigma(x) \nbis \bar{\alpha} q_{i_1} + \dots + \bar{\alpha} q_{i_m}+ r$, for any closed term $r$ (since $\init{\sigma(x)} = \{\alpha\}$) thus contradicting $\sigma(u') \bis p_n \| \sigma(t')$.
\end{itemize}
\end{itemize}
We have therefore obtained that whenever $x \trt{\sx}{\alpha}t$ then also $x \trt{\sx}{\alpha}u$. 

Assume now that $t$ has a summand $x$.
We aim to show that $u$ has a summand $x$ as well.
Since $x \trt{\sx}{\alpha} x$ gives $x \trt{\sx}{\alpha} t$, by the first part of the Proposition we get $x \trt{\sx}{\alpha} u$ and thus there is an index $j \in J$ such that $x \trt{\sx}{\alpha} u_j$.
We now treat the cases of an operator $f$ that distributes over $+$ in its first argument and of an operator $f$ that does not distribute in either argument separately.

\textcolor{blue}{\sc Case of an operator $f$ that distributes over $+$ in its first argument.}
Consider the substitution $\sigma_{\scriptstyle \nil}$ mapping each variable to $\nil$.  
Pick an integer $m$ larger than the depth of $\sigma_{\scriptstyle \nil}(t)$ and of $\sigma_{\scriptstyle\nil}(u)$. 
Let $\sigma$ be the substitution mapping $x$ to the term $a^{\scriptstyle m+1}$ and agreeing with $\sigma_{\scriptstyle \nil}$ on all the other variables.

As $t \approx u$ is sound modulo bisimilarity, we have that
$
\sigma(t)\bis \sigma(u)
$.
Moreover, the term $\sigma(t)$ affords the transition $\sigma(t)\mv{a} a^{\scriptstyle m}$, for $t_i=x$ and $\sigma(x)=a^{\scriptstyle m+1}\mv{a} a^{\scriptstyle m}$. 
Hence, for some closed term $p$,
\[
\sigma(u)= \sum_{j\in J} \sigma(u_j) \mv{a} p \bis a^{\scriptstyle m} \enspace .
\]
This means that there is a $j\in J$ such that $\sigma(u_j) \mv{a}p$. 
We claim that this $u_j$ can only be the variable $x$. 
To see that this claim holds, observe, first of all, that $x\in \var(u_j)$.
In fact, if $x$ did not occur in $u_j$, then we would reach a contradiction thus:
\[
m = 
\depth(p) < \depth(\sigma(u_j)) = 
\depth(\sigma_{\scriptstyle \nil}(u_j)) \leq 
\depth(\sigma_{\scriptstyle \nil}(u)) < 
m \enspace .
\]
Using this observation and Lemma~\ref{lem:that_one}, it is not hard to show that, for each of the other possible forms $u_j$ may
have, $\sigma(u_j)$ does not afford an $a$-labelled transition leading to a term of depth $m$.  
We may therefore conclude that $u_j=x$, which was to be shown.

\textcolor{blue}{\sc Case of an operator $f$ that does not distribute over $+$ in either argument.}
Notice that in the case at hand, there must be at least one action $\mu \in \Act$ such that $R^{\ff}_{\mu}$ holds.
Assume such an action $\mu$.
Again, let $n$ be greater than the sizes of $t$ and $u$, and consider the substitution
\[
\sigma_1(y) = 
\begin{cases}
\alpha\alpha^{\scriptsize \le n} & \text{ if } y = x \\
\alpha + \mu & \text{ otherwise.}
\end{cases}
\]
Thus $\sigma_1(x) \trans[\alpha] \alpha^{\scriptsize \le n}$ and consequently $\sigma_1(t) \trans[\alpha] \alpha^{\scriptsize \le n}$.
Since $\sigma_1(t) \bis \sigma_1(u)$ it must hold that $\sigma_1(u) \trans[\alpha] q$ for some $q \bis \alpha^{\scriptsize \le n}$.
As $n$ is greater than the size of $u$, one can infer that $u$ can have a summand given by at most $\lfloor \frac{n-2}{2}\rfloor$ nested occurrences of $f$ (which is a binary operator of size at least $3$).
Since, moreover, all variables but $x$ are mapped into a term of depth $1$, we can infer that the only term that can be responsible for the $\alpha$-move to $q$ is a summand $u_j$ such that $x \trt{\sx}{\alpha} u_j$.
To show $u_j = x$ we show that the only other possible case, namely $u_j = f(u',u'')$ with $x \trt{\sx}{\alpha} u'$ leads to a contradiction.
Recall that by the proviso of the Proposition $u$ has no $\nil$ factors, which implies that $u',u'' \nbis \nil$.
Since moreover, $x \trt{\sx}{\alpha} u'$, by Lemma~\ref{lem:trt_open} and Lemma~\ref{lem:c2o} we get $u' \trans[x_{\sx}]_{\alpha} c$ and thus $u_{j} \trans[x_{\sx}]_{\alpha} c \| u''$ for some configuration $c \bis x_{\dd} \| u'''$ for some term $u'''$, so that $\sigma_1(u_j) \trans[\alpha] \sigma_1[x_{\dd} \mapsto \alpha^{\scriptsize \le n}](c) \| \sigma_1(u'') = q$.
However, $u'' \nbis \nil$ implies that either there is a term $v$ such that $u''\trans[\nu] v$, for some action $\nu$, or in $u''$ at least one variable occurs unguarded.
Hence, by the choice of $\sigma_1$, as both $L^{\ff}_{\alpha}$ and $R^{\ff}_{\mu}$ hold, we can infer that $\depth(\sigma_1(u'')) \ge 1$ which gives
\begin{align*}
n ={} & \depth(\alpha^{\scriptsize \le n}) 
= \depth(q) \\
={} & \depth(\sigma_1[x_{\dd} \mapsto \alpha^{\scriptsize \le n}](c) \| \sigma_1(u'')) \\
={} & \depth(\sigma_1[x_{\dd} \mapsto \alpha^{\scriptsize \le n}](c)) + \depth(\sigma_1(u'')) \\
\ge{} & \depth(\alpha^{\scriptsize \le n}) + \depth(\sigma_1(u'')) 
\ge n+1 
\end{align*}
thus contradicting $q \bis \alpha^{\scriptsize \le n}$.   
\end{proof}


\section{Negative result: the case $L^{\ff}_{a},L^{\ff}_{\bar{a}},L^{\ff}_{\tau}$}
\label{sec:Labat}

In this section we discuss the nonexistence of a finite axiomatisation of $\FCCS$ in the case of an operator $f$ that, modulo bisimilarity, distributes over summation in one of its arguments.
We expand only the case of $f$ distributing in the first argument.
(The case of distributivity in the second argument follows by a straightforward adaptation of the arguments we use in this section.)
Hence, in the current setting, we can assume the following set of SOS rules for $f$: 
\[
\SOSrule{x_1 \trans[\mu] y_1}{f(x_1,x_2) \trans[\mu] y_1 \| x_2}\; \forall\,\mu\in\Act
\qquad\quad
\SOSrule{x_1 \trans[\alpha] y_1 \quad x_2 \trans[\bar{\alpha}] y_2}{f(x_1,x_2) \trans[\tau] y_1 \| y_2}
\]
namely, only $L^{\ff}_{\mu}$ holds for each action $\mu$, and only $S_{\alpha,\bar{\alpha}}$ holds for some $\alpha \in \{a,\bar{a}\}$.

According to the proof strategy sketched in Section~\ref{sec:proof_method}, we now introduce a particular family of equations on which we will build our negative result.
We define
\begin{align*}
& p_n  =  \sum_{i=0}^{n} \bar{\alpha} \alpha^{\scriptstyle \leq i} & (n \ge 0) \enspace \phantom{.} \\
& e_n \colon \quad f(\alpha,p_n)  \approx  \alpha p_n + \sum_{i=0}^{n} \tau \alpha^{\scriptstyle \leq i} & (n\ge 0) \enspace .
\end{align*}
It is not dif{f}icult to check that the infinite family of equations $e_n$ is sound modulo bisimilarity.

Our order of business is now to prove the instance of Theorem~\ref{Thm:nonfin-en} considering the family of equations $e_n$ above, showing that no finite collection of equations over $\FCCS$ that are sound modulo bisimilarity can prove all of the equations $e_n$ ($n\ge 0$).

Formally, we prove the following theorem:

\begin{theorem}
\label{thm:Labat}
Assume an operator $f$ such that only $L^{\ff}_{\mu}$ holds for each action $\mu$ and only $S_{\alpha,\bar{\alpha}}^f$ holds.
Let $\E$ be a finite axiom system over $\FCCS$ that is sound modulo $\bistext$, 
$n$ be larger than the size of each term in the equations in $\E$, and $p,q$ be closed terms such that $p,q \bis f(\alpha,p_n)$. 
If $\E \vdash p \approx q$ and $p$ has a summand bisimilar to $f(\alpha,p_n)$, then so does $q$.
\end{theorem}

Then, since the left-hand side of equation $e_n$, viz.~the term $f(\alpha,p_n)$, has a summand bisimilar to $f(\alpha,p_n)$, whilst the right-hand side, viz.~the term $\alpha p_n + \sum_{i=0}^{n} \tau \alpha^{\scriptstyle \le i}$, does not, we can conclude that the infinite collection of equations $\{e_n \mid n \ge 0\}$ is the desired witness family.
Theorem~\ref{Thm:nonfin-en} is then proved for the considered class of auxiliary binary operators.

The remainder of this section is entirely devoted to a proof of the above statement.


\subsection{Case specific properties of $f(\alpha,p_n)$}

Firstly, we present a technical lemma stating that, under the considered set of rules for $f$, if a closed term $\sigma(t)$ is not bisimilar to $\nil$, then by instantiating the variables in $t$ with a process which is not bisimilar to $\nil$ we cannot obtain a closed instance of $t$ which is bisimilar to $\nil$. 

\begin{lemma}
\label{Lem:non-nil-subst}
Let $t$ be a $\FCCS^-$ term and let $\sigma$ be a substitution with $\sigma(t)\nbis\nil$. 
Assume that $u$ is a $\FCCS^-$ term that is not bisimilar to $\nil$. 
Then $\sigma[x \mapsto u](t)\nbis \nil$ for each variable $x$.
\end{lemma}

\begin{proof}
By induction on the structure $t$.
\end{proof}

\begin{remark}
\label{rmk:synch}
We have defined the processes $p_n$ in a such a way that an initial synchronization, in the scope of operator $f$, with the process $\alpha$ is always possible.
This choice will allow us to slightly simplify the reasoning in the proof of the upcoming Proposition~\ref{Propn:crux} and thus of the negative result (cf., for instance, with the proof of Proposition~\ref{prop:LaRa_substitution} in Section~\ref{sec:LaRa}).
Clearly, the possibility of synchronization is directly related to which rules of type (\ref{syncrule}) are available for $f$.
However, since $f$ has a rule of type (\ref{asyncruleleft}) for all actions, it is then always possible to identify a pair $\mu,p_n$ such that $f(\mu,p_n) \trans[\tau]$ due to an application of the rule of type (\ref{syncrule}) allowed for $f$.
\end{remark}

We now study some properties of the processes $f(\alpha,p_n)$, which also depend on the particular configuration of rules for $f$ that we are considering.

\begin{lemma}
\label{lem:rhs-prime}
The term $f(\alpha,p_n)$ is prime, for each $n\geq 0$.
\end{lemma}

\begin{proof}
Since $f(\alpha,p_n)$ is not bisimilar to $\nil$, to prove the statement it suffices only to show that $f(\alpha,p_n)$ is irreducible for $n\geq 0$.
 
If $n=0$ then $f(\alpha,p_n) = f(\alpha,\nil)$ is a term of depth $1$, and is therefore irreducible as claimed.
  
Consider now $n\geq 1$. 
Assume, towards a contradiction, that $f(\alpha,p_n) \bis p\|q$ for two closed terms $p$ and $q$ with $p\nbis \nil$ and $q\nbis \nil$, that is,  $f(\alpha,p_n)$ is
{\sl not} irreducible. 
We have that
\[
f(\alpha,p_n) \mv{\alpha} \nil\| p_n \bis p_n \enspace .
\]
As $f(\alpha,p_n) \bis p\|q$, there is a transition $p\|q \mv{\alpha} r$ for some $r\bis p_n$. 
Without loss of generality, we may assume that $p \mv{\alpha} p'$ and $r = p'\|q$. 
Since we have assumed that $n\geq 1$, by Lemma~\ref{lem:rhs-prime_1}.\eqref{claim:pn} and our assumption that $q\nbis \nil$, we have that $p'\bis \nil$ and $q\bis  p_n$. 
Again using that $n\geq 1$, it follows that $q \mv{\bar{\alpha}} q'$ for some $q'$.  
This means that $p\|q \mv{\bar{\alpha}}$, contradicting the assumption that $f(\alpha,p_n)  \bis p\|q$. 
Thus $f(\alpha,p_n)$ is irreducible, which was to be shown.
\end{proof}

\begin{lemma}
\label{Lem:decomposing}
Let $n\geq 1$. 
Assume that $f(p,q) \bis f(\alpha, p_n)$, where $q \nbis \nil$. 
Then $p\bis \alpha$ and $q \bis p_n$.
\end{lemma}

\begin{proof}
Since $f(p,q) \bis f(\alpha, p_n)$ and $f(\alpha, p_n) \mv{\alpha} \nil \| p_n \bis p_n$, there is a $p'$ such that $p \mv{\alpha} p'$ and $p'\| q \bis p_n$. 
It follows that $q \bis p_n$ and $p' \bis \nil$, because $p_n$ is prime (Lemma~\ref{lem:rhs-prime_1}(\ref{claim:pn})) and $q\nbis \nil$. 
We are therefore left to prove that $p$ is bisimilar to $\alpha$. 
To this end, note, first of all, that, as $\bis$ is a congruence over the language $\FCCS$, we have that
\[
f(p,p_n) \bis f(\alpha, p_n) \enspace .
\]
Assume now that $p\mv{\mu} p''$ for some action $\mu$ and closed term $p''$. 
In light of the above equivalence, one of the following two cases may arise:
\begin{enumerate}
\item $\mu=\alpha$ and $p''\| p_n \bis p_n $ or
\item $\mu=\tau$ and $p''\| p_n \bis \alpha^{\scriptstyle \leq i}$, for some $i\in\{1,\ldots,n\}$. 
\end{enumerate}
In the former case, $p''$ must have depth $0$ and is thus bisimilar to $\nil$. 
The latter case is impossible, because the depth of $p''\| p_n$ is at least $n+1$. 
  
We may therefore conclude that every transition of $p$ is of the form $p\mv{\alpha} p''$, for some $p''\bis \nil$. 
Since we have already seen that $p$ affords an $\alpha$-labelled transition leading to $\nil$, modulo bisimilarity, it follows that $p\bis \alpha$, which was to be shown.
\end{proof}


\subsection{Proving Theorem~\ref{thm:Labat}}

We now proceed to present a detailed proof of Theorem~\ref{thm:Labat}.
The following result, stating that the invariant property of having a summand bisimilar to $f(\alpha,p_n)$ holds for all closed instantiations of axioms in $\E$, will be the crux in such a proof.

\begin{proposition}
\label{Propn:crux}
Assume an operator $f$ that, modulo $\bistext$, distributes over $+$ in its first argument and such that only $L^{\ff}_{\mu}$ holds for each action $\mu$, and only $S_{\alpha,\bar{\alpha}}^f$ holds.

Let $t \approx u$ be an equation over $\FCCS^-$ that is sound modulo $\bistext$. 
Let $\sigma$ be a closed substitution with $p=\sigma(t)$ and $q=\sigma(u)$. 
Suppose that $p$ and $q$ have neither $\nil$ summands or factors and $p,q \bis f(\alpha,p_n)$ for some $n$ larger than the size of $t$. 
If $p$ has a summand bisimilar to $f(\alpha,p_n)$, then so does $q$.
\end{proposition}

\begin{proof}
Observe, first of all, that since $\sigma(t)=p$ and $\sigma(u)=q$ have no $\nil$ summands or factors, then neither do $t$ and $u$.  
Hence, by Remark~\ref{rmk:summands}, we have that, for some finite non-empty index sets $I, J$, $t = \sum_{i\in I} t_i$ and $u = \sum_{j\in J} u_j$, where none of the $t_i$ ($i\in I$) and $u_j$ ($j\in J$) is $\nil$, has $+$ as its head operator, has $\nil$ summands and factors.

As $p=\sigma(t)$ has a summand bisimilar to $f(\alpha,p_n)$, there is an index $i\in I$ such that $\sigma(t_i) \bis f(\alpha,p_n)$. 

Our aim is now to show that there is an index $j\in J$ such that $\sigma(u_j) \bis f(\alpha,p_n)$, proving that $q=\sigma(u)$ also has a summand bisimilar to $f(\alpha,p_n)$. 

We proceed by a case analysis on the form $t_i$ may have.

\begin{enumerate}
\item \label{Case:t-var} 
\textcolor{blue}{\sc Case $t_i = x$ for some variable $x$}.
In this case, we have $\sigma(x)\bis f(\alpha,p_n)$, and $t$ has $x$ as a summand. 
As $t \approx u$ is sound modulo bisimilarity and neither $t$ nor $u$ have $\nil$ summands or factors, it follows that $u$ also has $x$ as a summand (Proposition~\ref{prop:x_in_tu}). 
Thus there is an index $j\in J$ such that $u_j = x$, and, modulo bisimulation, $\sigma(u)$ has $f(\alpha,p_n)$ as a summand, which was to be shown.


\item \label{Case:t-pre} 
\textcolor{blue}{\sc Case $t_i = \mu t'$ for some term $t'$}. 
This case is vacuous because, since $\mu \sigma(t') \mv{\mu} \sigma(t')$ is the only transition afforded by $\sigma(t_i)$, this term cannot be bisimilar to $f(\alpha,p_n)$. 
Indeed $f(\alpha,p_n)$ can perform both, an $\alpha$-labelled transition triggered by the first argument, and the $\tau$-move due to the synchronisation between $\alpha$ and $p_n$.


\item \label{Case:t-hmerge} 
\textcolor{blue}{\sc Case $t_i = f(t',t'')$ for some terms $t',t''$}. 
In this case, $f(\sigma(t'),\sigma(t'')) \bis f(\alpha,p_n)$.
As $\sigma(t_i)$ has no $\nil$ factors, it follows that $\sigma(t')\nbis \nil$ and $\sigma(t'')\nbis \nil$.  
Thus $\sigma(t') \bis \alpha$ and $\sigma(t'')\bis p_n$ (Lemma~\ref{Lem:decomposing}). 
Now, $t''$ can be written as $t'' = v_1 + \cdots + v_\ell,  (\ell > 0)$, where none of the summands $v_i$ is $\nil$ or a sum. 
Observe that, since $n$ is larger than the size of $t$, we have that $\ell < n$.
Hence, since $\sigma(t'')\bis p_n = \sum_{i=1}^{n} \bar{\alpha}\alpha^{\scriptstyle \leq i}$, there must be some $h\in\{1,\ldots,\ell\}$ such that $\sigma(v_h) \bis \bar{\alpha}.\alpha^{\scriptstyle \leq i_1} + \cdots + \bar{\alpha}.\alpha^{\scriptstyle \leq i_m}$ for some $m>1$ and $1\leq i_1 < \ldots < i_m\leq n$. 
The term $\sigma(v_h)$ has no $\nil$ summands or factors---or else, so would $\sigma(t'')$, and thus $p=\sigma(t)$. 
By Lemma~\ref{Lem:vh-claim_gen}, it follows that $v_h$ can only be a variable $x$ and thus that
\begin{equation}
\label{Eqn:sigma(x)}
\sigma(x) \bis \bar{\alpha}.\alpha^{\scriptstyle \leq i_1} + \cdots + \bar{\alpha}.\alpha^{\scriptstyle \leq i_m} \enspace .
\end{equation}
Observe, for later use, that, since $t'$ has no $\nil$ factors, the above equation yields that $x\not\in\var(t')$---or else $\sigma(t')\nbis \alpha$ (Lemma~\ref{lem:that_one}). 
So, modulo bisimilarity, $t_i$ has the form $f(t',(x+t'''))$, for some term $t'''$, with $x\not\in\var(t')$ and $\sigma(t')\bis \alpha$.
    
Our order of business will now be to use the information collected so far to argue that $\sigma(u)$ has a summand bisimilar to $f(\alpha,p_n)$. 
To this end, consider the substitution 
\[
\sigma' = \sigma[x \mapsto \bar{\alpha} f(\alpha,p_n)] \enspace . 
\]
We have that 
\begin{align*}
\sigma'(t_i) &  ={} f(\sigma'(t'),\sigma'(t'')) \\
& ={} f(\sigma(t'),\sigma'(t'')) & \text{(As $x\not\in\var(t')$)} \\ 
& \bis f(\alpha, (\bar{\alpha} f(\alpha,p_n) +\sigma'(t''')) & \text{(As $t'' = x + t'''$).}
\end{align*}
Thus, $\sigma'(t_i) \mv{\tau} p' \bis f(\alpha,p_n)$ for some $p'$, so that 
\[
\sigma'(t) \mv{\tau} p' \bis f(\alpha,p_n)
\]
also holds. 
Since $t\approx u$ is sound modulo $\bistext$, it follows that 
\[
\sigma'(t) \bis \sigma'(u) \enspace . 
\]
Hence, we can infer that there are a $j\in J$ and a $q'$ such that
\begin{equation}
\label{Eqn:A.14}
\sigma'(u_j) \mv{\tau} q' \bis f(\alpha,p_n) \enspace .
\end{equation}
Recall that, by one of the assumptions of the proposition, $\sigma(u)\bis f(\alpha,p_n)$, and thus $\sigma(u)$ has depth $n+2$. 
On the other hand, by (\ref{Eqn:A.14}),
\[
\depth(\sigma'(u_j))\geq n+3 \enspace . 
\]
Since $\sigma$ and $\sigma'$ differ only in the closed term they map variable $x$ to, it follows that
\begin{equation}
\label{Eqn:A.15}
x\in \var(u_j) \enspace .
\end{equation}
We now proceed to show that $\sigma(u_j)\bis f(\alpha,p_n)$ by a further case analysis on the form a term $u_j$ satisfying (\ref{Eqn:A.14}) and (\ref{Eqn:A.15}) may have.
\begin{enumerate}
\item \label{Case:uj-var} 
\textcolor{red}{\sc Case $u_j = x$}. 
This case is vacuous because $\sigma' (x)= \bar{\alpha} f(\alpha,p_n)\nv{\tau}$, and thus this possible form for $u_j$ does not meet (\ref{Eqn:A.14}).

\item \label{Case:uj-pre} 
\textcolor{red}{\sc Case $u_j = \mu u'$ for some term $u'$}. 
In light of (\ref{Eqn:A.14}), we have that $\mu=\tau$ and $q'=\sigma'(u')\bis f(\alpha,p_n)$. 
Using (\ref{Eqn:A.15}) and the fact that $u'$ has no $\nil$ factors, we get $\depth(\sigma'(u'))\geq n+3$ (Lemma~\ref{lem:that_one}).  
Since $f(\alpha,p_n)$ has depth $n+2$, this contradicts $q'\bis f(\alpha,p_n)$.

\item \label{Case:uj-hmerge} 
\textcolor{red}{\sc Case $u_j = f(u',u'')$ for some terms $u',u''$}. 
Our assumption that $\sigma(u)$ has no $\nil$ factors yields that none of the terms $u',u'',\sigma(u')$ and $\sigma(u'')$ is bisimilar to $\nil$. 
Moreover, by (\ref{Eqn:A.15}), either $x\in\var(u')$ or $x\in\var(u'')$.
      
Since $\sigma'(u_j)=f(\sigma'(u'),\sigma'(u''))$ affords transition (\ref{Eqn:A.14}), we have that $q' = q_1\|q_2$ for some $q_1,q_2$. As $f(\alpha,p_n)$ is prime (Lemma~\ref{lem:rhs-prime}), it follows that either $q_1\bis \nil$ or $q_2\bis \nil$. 
Hence, we can distinguish two cases, according to the possible origins for transition (\ref{Eqn:A.14}):

\begin{enumerate} 
\item \textcolor{ForestGreen}{$\sigma'(u') \mv{\tau} q_1$ and $q_2 =\sigma'(u'')$.}
We now proceed to argue that this case produces a contradiction.

To this end, note first of all that $\sigma'(u'')\nbis\nil$, because $\sigma(u'')\nbis \nil$ (Lemma~\ref{Lem:non-nil-subst}). 
Thus it must be the case that $q_1\bis \nil$ and $q_2 =\sigma'(u'')\bis f(\alpha,p_n)$. 
In light of the definition of $\sigma'$, it follows that $x$ occurs in $u'$, but not in $u''$ (Lemma~\ref{lem:that_one}). 
Therefore, since $\sigma$ and $\sigma'$ only differ at the variable $x$,
\[
\sigma(u'') = \sigma'(u'')\bis f(\alpha,p_n) \enspace .
\]
Since $\bis$ is a congruence, we derive that 
\begin{equation}
\label{Eqn:uprime-tau}
\sigma(u_j) = f(\sigma(u'),\sigma(u'')) \bis f(\sigma(u'), f(\alpha,p_n)).
\end{equation}
Since $\sigma(u')\nbis \nil$ because $q=\sigma(u)$ has no $\nil$-factors, we may infer that
\begin{align*}
n+2 ={} & \depth(f(\alpha,p_n)) \\
={} & \depth(\sigma(u)) & \text{(As $\sigma(u) \bis f(\alpha,p_n)$)} \\ 
\ge{} & \depth(\sigma(u_j)) \\
={} & \depth(\sigma(u')) + n+2 & \text{(By (\ref{Eqn:uprime-tau}))} \\ 
>{} & n+2  & \text{(As $\depth(\sigma(u'))>0$),}
\end{align*}
which is the desired contradiction. 

\item \textcolor{ForestGreen}{$\sigma'(u') \mv{\alpha} q_1$ and $\sigma'(u'') \mv{\bar{\alpha}} q_2$.}
Recall that exactly one of $q_1,q_2$ is bisimilar to $\nil$. 
We proceed with the proof by considering these two possible cases in turn.
\begin{itemize}
\item \textcolor{purple}{\sc Case $q_1\bis \nil$}. 
Our order of business will be to argue that, in this case, $\sigma(u_j) \bis f(\alpha,p_n)$, and thus that $q=\sigma(u)$ has a summand bisimilar to  $f(\alpha,p_n)$.
        
To this end, observe, first of all, that $q_2\bis f(\alpha,p_n)$ by (\ref{Eqn:A.14}). 
It follows that $x\in\var(u'')$, for otherwise we could derive a contradiction thus:
\begin{align*}
\depth(f(\alpha,p_n)) ={} & \depth(\sigma(u)) & \text{(As $\sigma(u) \bis f(\alpha,p_n)$)} \\ 
\ge{} & \depth(\sigma(u_j)) \\
>{} & \depth(\sigma(u'')) & \text{(As $\depth(\sigma(u'))>0$)}\\
={} & \depth(\sigma'(u'')) & \text{(As $x\not\in\var(u'')$)} \\
>{} & \depth(f(\alpha,p_n)) & \text{(As $\sigma'(u'')\mv{\bar{\alpha}}q_2\bis f(\alpha,p_n)$).}  
\end{align*}        
Moreover, we claim that $x\not\in\var(u')$. 
Indeed, if $x$ also occurred in $u'$, then, since $u'$ has no $\nil$ factors, the term $\sigma(x)$ would contribute to the behaviour of $\sigma(u_j)$. 
Therefore, by (\ref{Eqn:sigma(x)}), the term $\sigma(u_j)$ would afford a sequence of actions containing two occurrences of $\bar{\alpha}$, contradicting our assumption that $\sigma(u)\bis f(\alpha,p_n)$. 
        
Observe now that, as $\sigma'(u'') \mv{\bar{\alpha}} q_2 \bis f(\alpha,p_n)$, it must be the case that $u''$ has a summand $x$.
To see that this does hold, we examine the other possible forms  a summand $w$ of $u''$ responsible for the transition 
\[
\sigma'(u'') \mv{\bar{\alpha}} q_2 \bis f(\alpha,p_n)
\]
may have, and argue that each of them leads to a contradiction. 
\begin{enumerate}
\item \textcolor{magenta}{\sc Case $w = \bar{\alpha} w'$, for some term $w'$}. 
In this case, $q_2 = \sigma'(w')$. However, the depth of such a $q_2$ is either smaller than $n+2$ (if $x\not\in\var(w')$), or larger than $n+2$ (if $x\in\var(w')$). 
More precisely, in the former case $x \not \in \var(w')$ implies $\sigma(w) = \sigma'(w)$ and thus $\sigma(u) \bis f(\alpha,p_n)$ gives $n+2 = \depth(\sigma(u)) \ge \depth(\sigma(w)) = 1 + \depth(\sigma(w'))$, giving $\depth(\sigma'(w')) \le n+1$.
In the latter case, as $x \in \var(w')$ and $w'$ does not have $\nil$ factors (or otherwise $u''$ would have $\nil$ factors), by Lemma~\ref{lem:that_one}, we would have $\depth(\sigma'(w')) \ge \depth(\sigma'(x)) = n+3$.
Both cases then contradict the fact that $q_2$ is bisimilar to $f(\alpha,p_n)$, because the latter term has depth $n+2$.
\item \textcolor{magenta}{\sc Case $w = f(w_1,w_2)$, for some terms $w_1$ and $w_2$}. 
Observe, first of all, that $\sigma(w_1)$ and $\sigma(w_2)$ are not bisimilar to $\nil$, because $\sigma(u)$ has no $\nil$ factors. 
It follows that $\sigma'(w_1)$ and $\sigma'(w_2)$ are not bisimilar to $\nil$ either (Lemma~\ref{Lem:non-nil-subst}). 
          
Now, since
\[
\sigma'(w)=f(\sigma'(w_1),\sigma'(w_2)) \mv{\bar{\alpha}} q_2 \enspace ,
\]
there is a closed term $q_3$ such that $\sigma'(w_1) \mv{\bar{\alpha}} q_3$ and 
\[
q_2 = q_3 \| \sigma'(w_2) \bis f(\alpha,p_n) \enspace .
\]
As the term $f(\alpha,p_n)$ is prime, and $\sigma'(w_2)$ is not bisimilar to $\nil$, we may infer that $q_3\bis \nil$ and 
\[
\sigma'(w_2) \bis f(\alpha,p_n) \enspace .
\]
It follows that $x\not\in\var(w_2)$, or else the depth of $\sigma'(w_2)$ would be at least $n+3$, and therefore that 
\[
\sigma'(w_2)=\sigma(w_2) \bis f(\alpha,p_n) \enspace .
\]
However, this contradicts our assumption that 
\[
q = \sigma(u) \bis f(\alpha,p_n) \enspace .
\]
\end{enumerate}
Summing up, we have argued that $u''$ has a summand $x$.
Therefore, by (\ref{Eqn:sigma(x)}),
\[
\sigma(u'') \bis \bar{\alpha}.\alpha^{\scriptstyle \leq i_1} + \cdots + \bar{\alpha}.\alpha^{\scriptstyle \leq i_m} + r'' \enspace ,
\]
for some closed term $r''$. We have already noted that 
\[
\sigma(u') = \sigma'(u')\mv{\alpha} q_1 \bis \nil \enspace .
\]
Therefore, we have that 
\[
\sigma(u') \bis \alpha + r' \enspace ,
\]
for some closed term $r'$. 
Using the congruence properties of bisimilarity, we may infer
\[
\sigma(u_j) = 
f(\sigma(u'),\sigma(u'')) \bis 
f((\alpha + r'), (\sum_{j = 1}^m \bar{\alpha}.\alpha^{\scriptstyle \leq i_j} + r'')).
\]
In light of this equivalence, we have that 
\[
\sigma(u_j) \mv{\alpha} r \bis \sum_{j = 1}^m \bar{\alpha}.\alpha^{\scriptstyle \leq i_j} + r'' \bis \sigma(u''),
\]
for some closed term $r$, and thus 
\[
q = \sigma(u) \mv{\alpha} r \enspace . 
\]
Since $q = \sigma(u) \bis f(\alpha,p_n)$ by our assumption, it must be the case that $r \bis \sigma(u'')\bis p_n$.
So, again using the congruence properties of $\bis$,  we have that
\[
\sigma(u_j) = f(\sigma(u'),\sigma(u'')) \bis f((\alpha + r'), p_n).
\]
As $\sigma(u) \bis f(\alpha, p_n)$, using Lemma~\ref{Lem:decomposing} it is now a simple matter to infer that
\[
\sigma(u') \bis \alpha \enspace . 
\]
Hence $\sigma(u_j)\bis f(\alpha, p_n)$. 
Note that $\sigma(u_j)$ is a summand of $q=\sigma(u)$. 
Therefore $q$ has a summand bisimilar to $f(\alpha, p_n)$, which was to be shown.
\item \textcolor{purple}{\sc Case $q_2\bis \nil$}. 
We now proceed to argue that this case produces a contradiction. 
To this end, observe, first of all, that $q_1\bis f(\alpha, p_n)$. 
Reasoning as in the analysis of the previous case, we may infer that $x$ occurs in $u'$, but $x$ does not occur in $u''$. 
Moreover, as $\sigma'(u') \mv{\alpha} q_1\bis f(\alpha, p_n)$, it must be the case that $u' \mv{\alpha} u'''$ for some $u'''$ such that 
\[
\sigma'(u''') = q_1 \bis f(\alpha,p_n) \enspace .
\]
\textcolor{DarkSlateGray}{(For, otherwise, using Lemma~\ref{lem:c2o}.\ref{lem:c2o_l}, we would have that $\sigma'(u') \mv{\alpha} q_1$ because $u' \mv{y} c$, $\sigma(y) \mv{\alpha} q_1'$ and $q_1 = \sigma'[y_d \mapsto q_1'](c)$, for some variable $y$, configuration $c$ and closed term $q_1'$. 
Then we would necessarily have that $y\neq x$.
In fact, if $y=x$, then we would have that $\alpha=\bar{\alpha}$ by the definition of $\sigma'$, contradicting the distinctness of these two complementary actions.      
Observe now that, again in light of the definition of $\sigma'$, the variable $x$ cannot occur in $c$, or else the depth of 
$
q_1 = \sigma'[y_d \mapsto q_1'](c)
$
would be at least $n+3$, contradicting our assumption that 
$
q_1 \bis f(\alpha, p_n) \enspace . 
$
Hence, since the variable $y$ is different from $x$, it is not hard to see that $\sigma(u') \mv{\alpha} q_1$ also holds, and thus that 
$
\depth(q_1) < \depth(\sigma(u)) = n+2 \enspace ,
$
contradicting our assumption that $q_1 \bis f(\alpha, p_n)$.)}

Since $u$ contains no $\nil$ factors, in light of the definition of $\sigma'$, this $u'''$ cannot contain occurrences of the variable $x$.
\textcolor{DarkSlateGray}{(For, otherwise, Lemma~\ref{lem:that_one} would yield that
$
\depth(\sigma'(u''')) = \depth(q_1) \geq n+3 \enspace , 
$
contradicting our assumption that $q_1 \bis f(\alpha, p_n)$.)}

So
\[
\sigma(u''') = q_1 \bis f(\alpha , p_n) 
\]
also holds. 
Thus 
\begin{align*}
n+2 ={} & \depth(f(\alpha, p_n)) \\
={} & \depth(\sigma(u)) & \text{(As $\sigma(u) \bis f(\alpha, p_n)$)} \\
\ge{} & \depth(\sigma(u_j)) \\
={} & \depth(f(\sigma(u'),\sigma(u''))) \\
>{} & \depth(\sigma(u''')) + \depth(\sigma(u'')) &  \text{(As $\sigma(u')\mv{\alpha} \sigma(u''')$)} \\
>{} &  n+2
\end{align*}
where the last inequality follows by the fact that $\depth(\sigma(u''))>0$ and $\depth(\sigma(u'''))=n+2$, and gives the desired contradiction. 
\end{itemize}
\end{enumerate}
This completes the proof for the case $u_j = f(u',u'')$ for some terms $u',u''$.
\end{enumerate}
\end{enumerate}
The proof of Proposition~\ref{Propn:crux} is now complete. 
\end{proof}

We are now ready to prove Theorem~\ref{thm:Labat}. \\

\begin{apx-proof}{Theorem~\ref{thm:Labat}}
Assume that $\E$ is a finite axiom system over the language $\FCCS^-$ that is sound modulo bisimilarity, and that the following hold, for some closed terms $p$ and $q$ and positive integer $n$ larger than the size of each term in the equations in $\E$:
\begin{enumerate}
\item $\E \vdash p \approx q$,
\item $p \bis q \bis f(\alpha,p_n)$, 
\item $p$ and $q$ contain no occurrences of $\nil$ as a summand or factor, and
\item $p$ has a summand bisimilar to $f(\alpha,p_n)$.
\end{enumerate}
We prove that $q$ also has a summand bisimilar to $f(\alpha,p_n)$ by induction on the depth of the closed proof of the equation $p\approx q$ from $\E$. 
Recall that, without loss of generality, we may assume that the closed terms involved in the proof of the equation $p \approx q$ have no $\nil$ summands or factors (by Proposition~\ref{Propn:proofswithout0}, as $\E$ may be assumed to be saturated), and that applications of symmetry happen first in equational proofs
(that is, $\E$ is closed with respect to symmetry). 
  
We proceed by a case analysis on the last rule used in the proof of $p \approx q$ from $\E$. 
The case of reflexivity is trivial, and that of transitivity follows immediately by using the inductive hypothesis twice. 
Below we only consider the other possibilities.
\begin{itemize}
\item \textcolor{blue}{\sc Case $\E \vdash p \approx q$, because $\sigma(t)=p$ and $\sigma(u)=q$ for some equation $(t\approx u)\in E$ and closed substitution $\sigma$}. 
Since $\sigma(t)=p$ and $\sigma(u)=q$ have no $\nil$ summands or factors, and $n$ is larger than the size of each term mentioned in equations in $\E$, the claim follows by Proposition~\ref{Propn:crux}.
\item \textcolor{blue}{\sc Case $\E \vdash p \approx q$, because $p=\mu p'$ and $q=\mu q'$ for some $p',q'$ such that $E \vdash p' \approx q'$}.
This case is vacuous because $p=\mu p'\nbis f(\alpha , p_n)$, and thus $p$ does not have a summand bisimilar to $f(\alpha, p_n)$.
\item \textcolor{blue}{\sc Case $\E \vdash p \approx q$, because $p=p'+p''$ and $q=q'+q''$ for some $p',q',p'',q''$ such that $E \vdash p' \approx q'$ and $E \vdash p'' \approx q''$}.
Since $p$ has a summand bisimilar to $f(\alpha , p_n)$, we have that so does either $p'$ or $p''$. 
Assume, without loss of generality, that $p'$ has a summand bisimilar to $f(\alpha , p_n)$. 
Since $p$ is bisimilar to $f(\alpha , p_n)$, so is $p'$. Using the soundness of $\E$ modulo bisimulation, it follows that $q'\bis f(\alpha , p_n)$. 
The inductive hypothesis now yields that $q'$ has a summand bisimilar to $f(\alpha , p_n)$. 
Hence, $q$ has a summand bisimilar to $f(\alpha , p_n)$, which was to be shown.
\item \textcolor{blue}{\sc Case $\E \vdash p \approx q$, because $p=f(p',p'')$ and $q=f(q',q'')$ for some $p',q',p'',q''$ such that $E \vdash p' \approx q'$ and $E \vdash p'' \approx q''$}. 
Since the proof involves no uses of $\nil$ as a summand or a factor, we have that $p',p''\nbis \nil$ and $q',q''\nbis \nil$. 
It follows that $q$ is a summand of itself. 
By our assumptions, $f(\alpha , p_n) \bis q$.
Therefore we have that $q$ has a summand bisimilar to $f(\alpha, p_n)$, and we are done.
\end{itemize}
This completes the proof of Theorem~\ref{thm:Labat} and thus of Theorem~\ref{Thm:nonfin-en} in the case of an operator $f$ that, modulo bisimilarity distributes over summation in its first argument. 
\end{apx-proof}


\section{Negative result: the case $L^{\ff}_{\alpha} \wedge R^{\ff}_{\alpha}$}
\label{sec:LaRa}

In this section we investigate the first case, out of three, related to an operator $f$ that, modulo bisimilarity, does not distribute over summation in either of its arguments.

We choose $\alpha \in \{a,\bar{a}\}$ and we assume that the set of rules for $f$ includes 
\[
\SOSrule{x_1 \trans[\alpha] y_1}{f(x_1,x_2) \trans[\alpha] y_1 \| x_2}
\qquad\quad
\SOSrule{x_2 \trans[\alpha] y_2}{f(x_1,x_2) \trans[\alpha] x_1 \| y_2} \enspace ,
\]
namely, predicate $L^{\ff}_{\alpha} \wedge R^{\ff}_{\alpha}$ holds for $f$.

We stress that the validity of the negative result we prove in this section does not depend on which types of rules with labels $\bar{\alpha}$ and $\tau$ are available for $f$.
Moreover, the case of an operator for which $L^{\ff}_{\bar{\alpha}} \wedge R^{\ff}_{\bar{\alpha}}$ holds can be easily obtained from the one we are considering, and it is therefore omitted.

We now introduce the infinite family of valid equations, modulo bisimilarity, that will allow us to obtain the negative result in the case at hand.
We define
\begin{align*}
& q_n = \sum_{i = 0}^n \alpha\bar{\alpha}^{\scriptsize \le i} & (n \ge 0) \enspace\phantom{.} \\
& e_n \colon \quad f(\alpha,q_n) \approx \alpha q_n + \sum_{i = 0}^n \alpha ( \alpha \| \bar{\alpha}^{\scriptsize \le i} ) & (n \ge 0) \enspace.
\end{align*}

Following the proof strategy from Section~\ref{sec:proof_method}, we aim to show that, when $n$ is \emph{large enough}, the witness property of having a summand bisimilar to $f(\alpha,q_n)$ is preserved by derivations from a finite, sound axiom system $\E$, as stated in the following theorem:

\begin{theorem}
\label{thm:LaRa}
Assume an operator $f$ such that $L^{\ff}_{\alpha} \wedge R^{\ff}_{\alpha}$ holds.
Let $\E$ be a finite axiom system over $\FCCS$ that is sound modulo $\bistext$, $n$ be larger than the size of each term in the equations in $\E$, and $p,q$ be closed terms such that $p,q \bis f(\alpha,q_n)$.
If $\E \vdash p \approx q$ and $p$ has a summand bisimilar to $f(\alpha,q_n)$, then so does $q$.
\end{theorem}

Then, we can conclude that the infinite collection of equations $\{e_n \mid n \ge 0\}$ is the desired witness family. 
In fact, the left-hand side of equation $e_n$, viz.~the term $f(\alpha,q_n)$, has a summand bisimilar to $f(\alpha,q_n)$, whilst the right-hand side, viz.~the term $\alpha q_n + \sum_{i=0}^{n} \alpha(\alpha \| \bar{\alpha}^{\scriptstyle \le i})$, does not.


\subsection{Case specific properties of $f(\alpha,q_n)$}

Before proceeding to the proof of Theorem~\ref{thm:LaRa}, we discuss a few useful properties of the processes $f(\alpha,q_n)$.
Such properties are stated in Lemmas~\ref{lem:f_vs_par} and~\ref{lem:f_vs_f_aa} and they are the updated versions of, respectively, Lemmas~\ref{lem:rhs-prime} and~\ref{Lem:decomposing} with respect to the current set of SOS rules that are allowed for $f$.

\begin{lemma}
\label{lem:f_vs_par}
For each $n \ge 0$ it holds that $f(\alpha,q_n) \bis \alpha \| q_n$.
\end{lemma}

\begin{lemma}
\label{lem:f_vs_f_aa}
Let $n \ge 1$.
Assume that $f(p,q) \bis f(\alpha,q_n)$ for $p,q \nbis \nil$.
Then 
\begin{inparaenum}
\item either $p \bis \alpha$ and $q \bis q_n$, 
\item or $q \bis \alpha$ and $p \bis q_n$.
\end{inparaenum}
\end{lemma}

\begin{proof}
Since $f(p,q) \bis f(\alpha, q_n)$ and $f(\alpha, q_n)\mv{\alpha} \nil \| q_n \bis q_n$, we can distinguish the following two cases depending on whether a matching transition from $f(p,q)$ stems from $p$ or $q$:
\begin{itemize}
\item There is a $p'$ such that $p \mv{\alpha}p'$ and $p'\| q \bis q_n$. 
It follows that $q \bis q_n$ and $p'  \bis \nil$, because $q_n$ is prime (Lemma~\ref{lem:rhs-prime_1}(\ref{claim:pn})) and $q\nbis \nil$. 
We are   therefore left to prove that $p$ is bisimilar to $\alpha$. 
To this end, note, first of all, that, as $\bis$ is a congruence over the language $\FCCS$, we have that
\[
f(p,q_n) \bis f(\alpha, q_n) \enspace .
\]
First of all, notice that the equivalence above implies that $\depth(p) = 1$.
We proceed to prove that $p \bis \alpha$.
Assume towards a contradiction that $p \nbis \alpha$ and thus that $p \trans[\mu] \nil$ for some $\mu \neq \alpha$.
We can distinguish two cases, according to whether the predicate $L^{\ff}_{\mu}$ holds or not.
\begin{itemize}
\item Assume first that $L^{\ff}_{\mu}$ holds.
Then we would have $\init{f(p,q_n)} = \{\alpha,\mu\}$ and $\init{f(\alpha,q_n)} = \{\alpha\}$, thus contradicting $f(p,q_n) \bis f(\alpha,q_n)$.
\item Assume now that $L^{\ff}_{\mu}$ does not hold.
Then, in light of the above equivalence, from $q_n \nbis \bar{\alpha}^{\scriptsize \le n}$ and the fact that $f(\alpha,q_n) \trans[\alpha] \alpha \| \bar{\alpha}^{\scriptsize \le n}$, we infer $f(p,q_n) \trans[\alpha] p \| \bar{\alpha}^{\scriptsize \le n}$ and $p \| \bar{\alpha}^{\scriptsize \le n} \bis \alpha \| \bar{\alpha}^{\scriptsize \le n}$.

Now, if $\mu = \tau$, then $p \| \bar{\alpha}^{\scriptsize \le n} \trans[\tau] \nil \| \bar{\alpha}^{\scriptsize \le n} \bis \bar{\alpha}^{\scriptsize \le n}$.
However, $\alpha \| \bar{\alpha}^{\scriptsize \le n}$ can perform a $\tau$-move only due to a synchronization between $\alpha$ and one of the $\bar{\alpha}$, thus implying that $\alpha \| \bar{\alpha}^{\scriptsize \le n} \trans[\tau] \nil \| \bar{\alpha}^{\scriptsize i} \bis \bar{\alpha}^{\scriptsize i}$ for some $i \in \{0,\dots,n-1\}$.
Since there is no such index $i$ such that $\bar{\alpha}^{\scriptsize \le n} \bis \bar{\alpha}^{\scriptsize i}$, this contradicts $f(p,q_n) \bis f(\alpha, q_n)$.

Similarly, if $\mu = \bar{\alpha}$, then $p \| \bar{\alpha}^{\scriptsize \le n}$ could perform a sequence of $n+1$ transitions all with label $\bar{\alpha}$, whereas $\alpha \| \bar{\alpha}^{\scriptsize \le n}$ can perform at most $n$ $\bar{\alpha}$-moves in a row.
Therefore, also this case is in contradiction with $f(p,q_n) \bis f(\alpha, q_n)$.
\end{itemize}
We may therefore conclude that every transition of $p$ is of the form $p\mv{\alpha} p''$, for some $p''\bis \nil$. 
Since we have already seen that $p$ affords an $\alpha$-labelled transition leading to $\nil$, modulo bisimilarity, it follows that $p\bis \alpha$, which was to be shown.

\item There is a $q'$ such that $q \trans[\alpha] q'$ and $p \| q' \bis q_n$.
This case can be treated similarly to the previous case and allows us to conclude that $q \bis \alpha$ and $p \bis q_n$.
\end{itemize}
\end{proof}


\subsection{Proving Theorem~\ref{thm:LaRa}}

The negative result stated in Theorem~\ref{thm:LaRa} is strongly based on the following proposition, which ensures that the property of having a summand bisimilar to $f(\alpha,q_n)$ is preserved by the closure under substitution of equations in a finite sound axiom system.

\begin{proposition}
\label{prop:LaRa_substitution}
Assume an operator $f$ such that $L^{\ff}_{\alpha} \wedge R^{\ff}_{\alpha}$ holds.

Let $t \approx u$ be an equation over $\FCCS^-$ that is sound modulo $\bistext$.
Let $\sigma$ be a closed substitution with $p = \sigma(t)$ and $q = \sigma(u)$.
Suppose that $p$ and $q$ have neither $\nil$ summands nor factors, and $p,q \bis f(\alpha,q_n)$ for some $n$ larger than the size of $t$.
If $p$ has a summand bisimilar to $f(\alpha,q_n)$, then so does $q$.
\end{proposition}

\begin{proof}
First of all we notice that since $\sigma(t)$ and $\sigma(u)$ have no $\nil$ summands or factors, then neither do $t$ and $u$.
Therefore by Remark~\ref{rmk:summands} we get that $t = \sum_{i \in I} t_i$ and $u = \sum_{j \in J} u_j$, for some finite non-empty index sets $I,J$ with all the $t_i$ and $u_j$ not having $+$ as head operator, $\nil$ summands nor factors.
By the hypothesis, there is some $i \in I$ with $\sigma(t_i) \bis f(\alpha,q_n)$.
We proceed by a case analysis over the structure of $t_i$ to show that there is a $u_j$ such that $\sigma(u_j) \bis f(\alpha,q_n)$.
\begin{enumerate}
\item \textcolor{blue}{\sc Case $t_i = x$ for some variable $x$ such that $\sigma(x) \bis f(\alpha,q_n)$.}
By Proposition~\ref{prop:x_in_tu}, $t$ having a summand $x$ implies that $u$ has a summand $x$ as well.
Thus, we can immediately conclude that $\sigma(u)$ has a summand bisimilar to $f(\alpha,q_n)$ as required.

\item \textcolor{blue}{\sc Case $t_i = \mu.t'$ for some term $t'$.}
This case is vacuous, as it contradicts our assumption $\sigma(t_i) \bis f(\alpha,q_n)$. 
Indeed, if $\mu = \alpha$ then $\sigma(t')$ cannot be bisimilar to both $q_n$ and $\alpha \| \bar{\alpha}^{\scriptsize \le i}$, for any $i \in \{1,\dots,n\}$.

\item \textcolor{blue}{\sc Case $t_i = f(t',t'')$ for some terms $t',t''$.}
As $\sigma(t)$ has no $\nil$ factors, we have $\sigma(t'),\sigma(t'') \nbis \nil$.
Hence, from $f(\sigma(t'),\sigma(t'')) \bis f(\alpha,q_n)$ and Lemma~\ref{lem:f_vs_f_aa} we can distinguish two cases: 
\begin{inparaenum}
\item either $\sigma(t') \bis \alpha$ and $\sigma(t'') \bis q_n$, 
\item or $\sigma(t') \bis q_n$ and $\sigma(t'') \bis \alpha$.
\end{inparaenum}
We expand only the former case, as the latter follows from an identical (symmetrical) reasoning.
By Remark~\ref{rmk:summands}, from $\sigma(t'') \bis q_n$ we infer that $t'' = \sum_{h \in H} v_h$ for some terms $v_h$ that do not have $+$ as head operator and have no $\nil$-summands or factors.
Since $n$ is larger that the size of $t$, we have that $|H| < n$ and thus there is some $h \in H$ such that $\sigma(v_h) \bis \sum_{k = 1}^m \alpha\bar{\alpha}^{\scriptsize \le i_k}$ for some $m > 1$ and $1 \le i_1 < \dots < i_m \le n$.
Since $\sigma(v_h)$ has no $\nil$ summands or factors, from Lemma~\ref{Lem:vh-claim_gen} we infer that $v_h$ can only be a variable $x$ with
\begin{equation}
\label{eq:sigma_x}
\sigma(x) \bis \sum_{k = 1}^m \alpha\bar{\alpha}^{\scriptsize \le i_k}.
\end{equation}
Therefore, $t_i = f(t', x+t''')$ for some $t'''$ such that $\sigma(x + t''') \bis q_n$.
We also notice that since $\sigma(t') \bis \alpha$ and $\sigma(t')$ has no $\nil$ summands or factors, then it cannot be the case that $x \in \var(t')$.

To prove that $u$ has a summand bisimilar to $f(\alpha,q_n)$, consider the closed substitution
\[
\sigma' = \sigma[x \mapsto \alpha q_n].
\]
Since $R^{\ff}_{\alpha}$ and Lemma~\ref{lem:f_vs_par} hold, we have
\[
\sigma'(t_i) \trans[\alpha] p' \bis \alpha \| q_n \bis f(\alpha,q_n).
\]
As $t \approx u$ implies $\sigma'(t) \bis \sigma'(u)$, we infer that there must be a summand $u_j$ such that $\sigma'(u_j) \trans[\alpha] r$ for some $r \bis f(\alpha,q_n)$.
Notice that, since $\sigma(u) \bis f(\alpha,q_n)$ and $\sigma(u_j) = \sigma'(u_j)$ if $x \not \in \var(u_j)$, then it must be the case that $x \in \var(u_j)$, or otherwise we get a contradiction with $\sigma(u) \bis f(\alpha,q_n)$, as $\sigma(u_j) = \sigma'(u_j) \trans[\alpha] r$ would give $\sigma(u) \trans[\alpha] r \bis f(\alpha,q_n)$. 
However, there is no $r'$ such that $f(\alpha,q_n) \trans[\alpha] r'$ and $r' \bis f(\alpha,q_n)$.
By Lemma~\ref{lem:c2o}, as $L^{\ff}_{\alpha} \wedge R^{\ff}_{\alpha}$ holds, we can distinguish two cases:

\begin{enumerate}
\item \textcolor{red}{There is a term $u'$ s.t.\ $u_j \trans[\alpha] u'$ and $\sigma'(u') \bis f(\alpha,q_n)$.}
Then, since $f(\alpha,q_n) \bis \alpha \parallel q_n$ (Lemma~\ref{lem:f_vs_par}) we can apply the expansion law, obtaining 
\[
\sigma'(u') \bis \sum_{i = 1}^n \alpha(\alpha \| \bar{\alpha}^{\scriptsize \le i}) + \alpha q_n.
\]
As $n$ is greater than the size of $u$, and thus of those of $u_j$ and $u'$, by Lemma~\ref{lem:variables} we get that $u'$ has a summand $y$, for some variable $y$, such that 
\[
\sigma'(y) \bis \sum_{k = 1}^{m'} \alpha q_{i'_k} + r',
\]
for some $m' > 1$, $1 \le i'_1 < \dots < i'_{m'} \le n$, closed term $r'$ and closed terms $q_{i'_k}$ such that either $q_{i'_k} \bis \alpha \| \bar{\alpha}^{\scriptsize \le i'_k}$ or $q_{i'_k} \bis \bar{\alpha}^{\scriptsize \le i'_k}$, for each $k = 1,\dots, m'$.
(We can infer the exact form of the $q_{i'_k}$ since $\alpha$ and $\bar{\alpha}^{\scriptsize \le i'_k}$ are prime, the parallel component $\alpha$ is common to all summands and $\bar{\alpha}^{\scriptsize \le i'_k} \nbis \bar{\alpha}^{\scriptsize \le i'_j}$ if $k \neq j$).
In both cases, we can infer that $y \neq x$, as $\sigma'(x) \nbis \sigma'(y)$ for any closed term $r'$.
We have thus $\sigma'(y) = \sigma(y)$ and we get a contradiction with $\sigma(u) \bis f(\alpha,q_n)$ in that $\sigma(u_j)$ would be able to perform three $\alpha$-moves in a row.
In fact
\begin{align*}
\sigma(u_j) \trans[\alpha]{} & \sigma(u')  & \text{($u'$ has a summand $y$)}\\
\trans[\alpha]{} & \alpha \| \bar{\alpha}^{\scriptsize \le i'_k} & \text{for some $k \in \{1,\dots,m'\}$} \\
\trans[\alpha]{} & \bar{\alpha}^{\scriptsize \le i'_k},
\end{align*}
whereas $\sigma(u) \bis f(\alpha,q_n)$ can perform only two such transitions.

\item \textcolor{red}{There are a variable $y$, a closed term $r'$ and a configuration $c$ s.t.\ $\sigma'(y) \trans[\alpha] r'$, $u_j \trans[y_{\bb}]_{\alpha} c$ and $\sigma'[y_{\dd} \mapsto r'](c) \bis f(\alpha,q_n)$.}
We claim that it must be the case that $y =x$.
To see this, assume towards a contradiction that $y \neq x$.
We proceed by a case analysis on the possible occurrences of $x$ in $c$.
\begin{itemize}
\item \textcolor{Purple4}{$x \not \in \var(c)$ or $x \in \var(c)$ but its occurrence is in a guarded context that prevents the execution of its closed instances.}
In this case we get $r = \sigma[y_{\dd} \mapsto r'](c) \bis \sigma'[y_{\dd}\mapsto r'](c) \bis f(\alpha,q_n)$.
This contradicts $\sigma(u) \bis f(\alpha,q_n)$ since we would have $\sigma(u) \trans[\alpha] r \bis f(\alpha,q_n)$, and such a transition cannot be mimicked by $f(\alpha,q_n)$.
\item \textcolor{Purple4}{$x \in\var(c)$ and its execution is not prevented.}
We can distinguish two sub-cases, according to whether the occurrence of $x$ is guarded or not.
\begin{itemize}
\item \textcolor{DarkGoldenrod4}{Assume that $x$ occurs guarded in $c$.}
In this case we get a contradiction with $r \bis f(\alpha,q_n)$, as $x$ guarded implies:
\[
n+2 = \depth(f(\alpha,q_n)) 
= \depth(r) 
\ge 1 + \depth(\sigma'(x)) 
= n+3. 
\]
\item \textcolor{DarkGoldenrod4}{Assume now that $x \trt{\bb}{\alpha} c$.}
We proceed by a case analysis on the structure of $c$.
\begin{itemize}
\item \textcolor{magenta}{$c \bis y_{\dd}  \| (x + u_1) \| u_2$.}
Notice that in this case we have $r = r' \| (\sigma'(x) + \sigma'(u_1)) \| \sigma'(u_2)$.
Then, the only transition available for $\sigma'(x)$ is $\sigma'(x) \trans[\alpha] q_n$, which gives $r \trans[\alpha] r' \| q_n \| \sigma'(u_2)$.
Since $r \bis f(\alpha,q_n)$, then it must be the case that $f(\alpha,q_n) \trans[\alpha] r''$ for some $r'' \bis r' \| q_n \| \sigma'(u_2)$.
Since $q_n$ is prime, we can infer that $r'' \bis q_n$ and thus that $r' \bis \nil \bis \sigma'(u_2)$.
Hence, we have that $r \bis \sigma'(x) + \sigma'(u_1)$.
As the one we wrote is the only transition available for $\sigma'(x)$, we can infer that, for all $i \in \{1,\dots,n\}$, the transitions $r \trans[\alpha] \alpha \| \bar{\alpha}^{\scriptsize \le i}$ cannot be derived from $\sigma'(x)$, but only from $\sigma'(u_1)$.
Moreover, notice that $y \neq x$ gives $\sigma'(y) = \sigma(y)$, and from $\init{\sigma'(x)} = \init{\sigma(x)} = \{\alpha\}$ and the fact that $L^{\ff}_{\alpha} \wedge R^{\ff}_{\alpha}$ holds, we can infer that $\sigma(u_2) \bis \sigma'(u_2) \bis \nil$.
Therefore, this contradicts $\sigma(u) \bis f(\alpha,q_n)$, since $\sigma(u) \trans[\alpha] r' \| \sigma(x) + \sigma(u_1) \| \sigma(u_2) \bis \sigma(x) + \sigma(u_1) \trans[\alpha] \alpha \| \bar{\alpha}^{\scriptsize \le i}$, for any $i \in \{1,\dots,n\}$.
Process $f(\alpha,q_n)$, in turn, by performing two $\alpha$-moves can only reach processes bisimilar to $\bar{\alpha}^{\scriptsize \le i}$, for $i \in \{1,\dots,n\}$.
\item \textcolor{magenta}{$c$ has a subterm $u_3$ of the form $u_3 \bis f(x + u_2, u_1)$ or $u_3 \bis f(u_1, x +u_2)$.}
In both cases, we get that $\sigma'(x) \trans[\alpha] q_n$ implies $\sigma'(u_3) \trans[\alpha] q_n \| \sigma'(u_1)$.
However, $f(\alpha,q_n) \trans[\alpha] \nil \| q_n \bis q_n$ and $q_n$ prime give $\sigma'(u_1) \bis \nil$.
One can then argue that, as $\init{\sigma'(x)} = \{\alpha\}$, either $x$ does not occur in $u_1$, or it does it in a guarded context that prevents its execution.
Hence, we infer $\sigma(u_1) \bis \sigma'(u_1) \bis \nil$, thus contradicting $\sigma(u)$ not having $\nil$ factors.
\end{itemize}
\end{itemize}
\end{itemize}
Therefore, we can conclude that it must be the case that $y=x$ and $r' = q_n$.
In particular, notice that $x \trt{\bb}{\alpha} u_j$.
We now proceed by a case analysis on the structure of $u_j$ to show that $\sigma(u_j) \bis f(\alpha,q_n)$.
\begin{enumerate}
\item \textcolor{ForestGreen}{$u_j = x$.} 
This case is vacuous, as $\sigma'(x) \trans[\alpha] q_n$ and
$q_n \nbis f(\alpha,q_n)$.
\item \textcolor{ForestGreen}{$u_j = f(u',u'')$ for some $u',u''$.}
Notice that $x \trt{\bb}{\alpha} u_j$ can be due either to $x \trt{\bb}{\alpha} u'$ or $x \trt{\bb}{\alpha} u''$.
As both $\sigma'(u')$ and $\sigma'(u'')$ can be responsible for the $\alpha$-move by $\sigma'(u_j)$, we distinguish two cases:

\begin{enumerate}
\item \textcolor{purple}{$\sigma'(u') \trans[\alpha] r_1$ and $r_1 \| \sigma'(u'') \bis f(\alpha,q_n)$.}
As $f(\alpha,q_n) \bis \alpha \| q_n$ and both $\alpha$ and $q_n$ are prime, by the existence of a unique prime decomposition, we distinguish two cases:
\begin{itemize}
\item \textcolor{DeepPink}{$r_1 \bis \alpha$ and $\sigma'(u'') \bis q_n$.}
Since $x \trt{\bb}{\alpha} u''$ is in contradiction with $\sigma'(u'') \bis q_n$, we infer that $x \trt{\bb}{\alpha} u'$.
Moreover $\init{\sigma(x)} = \init{\sigma'(x)} = \{\alpha\}$, $L^{\ff}_{\alpha} \wedge R^{\ff}_{\alpha}$, $\sigma'(u'') \bis q_n$ and the fact that $\sigma(u)$ has no $\nil$ factors we get that either $x \not\in \var(u'')$ or $x$ occurs in $u''$ but its execution is prevented by the rules for $f$. 
Therefore 
\[
\sigma'(u'') \bis \sigma(u'') \bis q_n.
\]
However, $\depth(\sigma(x)) \ge 3$, and $x \trt{\bb}{\alpha} u'$ with $\init{\sigma(x)} = \{\alpha\}$ give us, by Lemma~\ref{lem:that_one}, that $\depth(\sigma(u')) \ge \depth(\sigma(x))$.
Therefore we get a contradiction, since
\begin{align*}
n+2 ={} & \depth(f(\alpha,q_n)) 
= \depth(\sigma(u))
\ge \depth(\sigma(u_j)) \\
={} & \depth(f(\sigma(u'),\sigma(u''))) 
\ge \depth(\sigma(x)) + \depth(\sigma(u'')) 
\ge 3 + n + 1.
\end{align*}
\item \textcolor{DeepPink}{$r_1 \bis q_n$ and $\sigma'(u'') \bis \alpha$.}
By reasoning as above, we can infer that either $x \not \in \var(u'')$ or its execution is blocked by the rules for $f$, so that $\sigma'(u'') \bis \sigma(u'')$.
Moreover, we get that $x \trt{\bb}{\alpha} u'$.
We aim at showing that $u'$ has a summand $x$.
We proceed by showing that the only other possibility, namely $u' = f(w_1,w_2)$ for some $w_1,w_2$, leads to a contradiction.
As $u' = f(w_1,w_2)$ we have that either $x \trt{\bb}{\alpha} w_1$ or $x \trt{\bb}{\alpha} w_2$.
However, $\sigma'(u') \trans[\alpha] r_1 \bis q_n$ gives two possibilities:
\begin{itemize}
\item \textcolor{RoyalBlue}{$\sigma'(w_1) \trans[\alpha] r_1'$ and $r_1' \| \sigma'(w_2) \bis q_n$.}
Since $q_n$ is prime, then either $r_1' \bis \nil$ and $\sigma'(w_2) \bis q_n$, or $r_1' \bis q_n$ and $\sigma'(w_2) \bis \nil$.
In both cases we infer that either $x \not \in \var(w_2)$ or its execution in it is always prevented, so that $\sigma(w_2) \bis \sigma'(w_2)$.
Therefore, the former case, combined with $\sigma(u'') \bis \alpha$, gives a contradiction with $\sigma(u) \bis f(\alpha,q_n)$.
The latter case contradicts $\sigma(u)$ not having $\nil$ factors.
\item \textcolor{RoyalBlue}{$\sigma'(w_2) \trans[\alpha] r_2'$ and $\sigma'(w_1) \| r_2' \bis q_n$.}
The same reasoning as in the previous case allows us to conclude that this case gives a contradiction.
\end{itemize}
Summing up, we have argued that $u'$ has a summand $x$.
Therefore, by Equation~\eqref{eq:sigma_x},
\[
\sigma(u') \bis \sum_{k = 1}^m\alpha.\bar{\alpha}^{\scriptstyle \leq i_k} + r'' \enspace ,
\]
for some closed term $r''$. 
We have already noted that
\[
\sigma(u'') \bis \sigma'(u'')  \bis \alpha \enspace .
\]
Therefore, using the congruence properties of bisimilarity, we may infer that
\[
\sigma(u_j) = 
f(\sigma(u'),\sigma(u'')) \bis 
f(\sum_{k = 1}^m \alpha\bar{\alpha}^{\scriptsize \le i_k} + r'', \alpha)
\enspace .
\]
In light of this equivalence, we have $\sigma(u_j) \trans[\alpha] r' \bis \sigma(u')$ and thus $\sigma(u) \trans[\alpha] r'$.
Since by hypothesis $\sigma(u) \bis f(\alpha,q_n)$ we have that either $r' \bis q_n$, or $r' \bis \alpha \| \alpha^{\scriptsize \le i}$ for some $i \in \{1,\dots,n\}$.
However, the latter case is in contradiction with $r' \bis \sigma(u')$, and thus it must be the case that $r' \bis q_n$.
Therefore, we can conclude that $\sigma(u_j) \bis f(q_n, \alpha)$.
It is easy to check that $f(\alpha,q_n) \bis f(q_n,\alpha)$.
Hence, $\sigma(u)$ has the desired summand.
\end{itemize}

\item \textcolor{purple}{$\sigma'(u'') \trans[\alpha] r_2$ and $\sigma'(u') \| r_2 \bis f(\alpha,q_n)$.}
This case follows as the previous one and allows us to conclude as well that $\sigma(u)$ has the desired summand.
\end{enumerate}
\end{enumerate}
\end{enumerate}
\end{enumerate}
The proof of Proposition~\ref{prop:LaRa_substitution} is now complete. 
\end{proof}

We have now all the necessary ingredients for the proof of Theorem~\ref{thm:LaRa}, which we present below.\\

\begin{apx-proof}{Theorem~\ref{thm:LaRa}}
Assume that $\E$ is a finite axiom system over the language $\FCCS^-$ that is sound modulo bisimilarity, and that the following hold, for some closed terms $p$ and $q$ and positive integer $n$ larger than the size of each term in the equations in $\E$:
\begin{enumerate}
\item $\E \vdash p \approx q$,
\item $p \bis q \bis f(\alpha,q_n)$, 
\item $p$ and $q$ contain no occurrences of $\nil$ as a summand or factor, and
\item $p$ has a summand bisimilar to $f(\alpha,q_n)$.
\end{enumerate}
We proceed by induction on the depth of the closed proof of the equation $p \approx q$ from $\E$, to prove that also $q$ has a summand bisimilar to $f(\alpha,q_n)$.
Recall that, without loss of generality, we may assume that $\E$ is closed with respect to symmetry, and thus applications of symmetry happen first in equational proofs.
We proceed by a case analysis on the last rule used in the proof of $p \approx q$ from $\E$.
The case of reflexivity is trivial, and that of transitivity follows by applying twice the inductive hypothesis.
We proceed now to a detailed analysis of the remaining cases:
\begin{enumerate}
\item \textcolor{blue}{\sc Case $\E \vdash p \approx q$ because $\sigma(t) = p$ and $\sigma(u) = q$ for some terms $t,u$ with $E \vdash t \approx u$ and closed substitution $\sigma$.}
The proof of this case follows by Proposition~\ref{prop:LaRa_substitution}.
\item \textcolor{blue}{\sc Case $\E \vdash p \approx q$ because $p = \mu.p'$ and $q = \mu.q'$ for some $p',q'$ with $E \vdash p' \approx q'$.}
This case is vacuous in that $p = \mu.p' \nbis f(\alpha,q_n)$ and thus $p$ does not have a summand bisimilar to $f(\alpha,q_n)$.
\item \textcolor{blue}{\sc Case $\E \vdash p \approx q$ because $p = r_1 + r_2$ and $q = s_1 + s_2$ for some $r_i,s_i$ with $E \vdash r_i \approx s_i$, for $i \in \{1,2\}$.} 
Since $p$ has a summand bisimilar to $f(\alpha,q_n)$ then so does either $r_1$ or $r_2$.
Assume without loss of generality that $r_1$ has such a summand.
As $p \bis f(\alpha,q_n)$ then $r_1 \bis f(\alpha,q_n)$ holds as well.
Then, from $E \vdash r_1 \approx s_1$ we infer $s_1 \bis f(\alpha,q_n)$.
Thus, by the inductive hypothesis we obtain that $s_1$ has a summand bisimilar to $f(\alpha,q_n)$ and, consequently, so does $q$.
\item \textcolor{blue}{\sc Case $\E \vdash p \approx q$ because $p = f(r_1,r_2)$ and $q = f(s_1,s_2)$ for some $r_i,s_i$ with $E \vdash r_i \approx s_i$, for $i \in \{1,2\}$.}
By the proviso of the theorem $p,q$ have neither $\nil$ summands nor factors, thus implying $r_i,s_i \nbis \nil$.
Hence, from $p \bis f(\alpha,q_n)$ and $p = f(r_1,r_2)$ and Lemma~\ref{lem:f_vs_f_aa} we obtain $r_i \bis \alpha$ and $r_{3-i} \bis q_n$, thus implying, by the soundness of the equations in $\E$, that $s_i \bis \alpha$ and $s_{3-i} \bis q_n$, so that either $q = f(\alpha,q_n)$ or $q = f(q_n,\alpha)$.
In both cases, we can infer that $q$ has itself as the desired summand.
\end{enumerate}
This completes the proof of Theorem~\ref{thm:LaRa} and thus of Theorem~\ref{Thm:nonfin-en} in the case of an operator $f$ that does not distribute over summation in either argument, case $L^{\ff}_{\alpha} \,\wedge\, R^{\ff}_{\alpha}$. 
\end{apx-proof}


\section{Negative result: the case $L^{\ff}_{\alpha}$, $R^{\ff}_{\bar{\alpha}}$}
\label{sec:LaRba}

In this section we deal with the second case related to an operator $f$ that does not distribute over summation in either argument.
This time, given $\alpha \in \{a,\bar{a}\}$, we assume that operator $f$ has only one rule with label $\alpha$ and only one rule with label $\bar{\alpha}$, and moreover we assume such rules to be of different types.
In detail, we expand the case in which for action $\alpha$ only the predicate $L^{\ff}_{\alpha}$ holds, and for action $\bar{\alpha}$ only $R^{\ff}_{\bar{\alpha}}$ holds, namely $f$ has rules:
\[
\SOSrule{x_1 \trans[\alpha] y_1}{f(x_1,x_2) \trans[\alpha] y_1 \| x_2}
\qquad\quad
\SOSrule{x_2 \trans[\bar{\alpha}] y_2}{f(x_1,x_2) \trans[\bar{\alpha}] x_1 \| y_2} \enspace .
\]

Once again, the proof for the symmetric case with $L^{\ff}_{\bar{\alpha}}$ and $R^{\ff}_{\alpha}$ holding is omitted.

To obtain the proof of the negative result, we consider the same family of witness processes $f(\alpha,p_n)$ from Section~\ref{sec:Labat}.
However, dif{f}erently from the previous case, the definition of the witness family of equations depends on which rules of type (\ref{syncrule}) are available for $f$.
More precisely, we need to split the proof of the negative result into two cases, according to whether the rules for $f$ allow $\alpha$ and $p_n$ to synchronise or not.


\subsection{Case 1: Possibility of synchronisation}

Assume first that $S_{\alpha,\bar{\alpha}}^f$ holds, so that the rule
\[
\SOSrule{x_1 \trans[\alpha] y_1 \quad x_2 \trans[\bar{\alpha}] y_2}{f(x_1,x_2) \trans[\tau] y_1 \| y_2}
\]
allows for synchronisation between $\alpha$ and $p_n$.
In this setting, the infinite family of equations
\[
e_n \colon \quad f(\alpha,p_n) \approx \alpha p_n + \sum_{i = 0}^n \bar{\alpha}(\alpha \| \alpha^{\scriptsize \le i}) + \sum_{i = 0}^n \tau\alpha^{\scriptsize \le i} \qquad (n \ge 0) 
\]
is sound modulo bisimilarity and it constitutes a family of witness equations. 

\begin{theorem}
\label{thm:LaRbaC}
Assume an operator $f$ such that only $L^{\ff}_{\alpha}$ holds for $\alpha$, only $R^{\ff}_{\bar{\alpha}}$ holds for $\bar{\alpha}$, and $S_{\alpha,\bar{\alpha}}^f$ holds.
Let $\E$ be a finite axiom system over $\FCCS$ that is sound modulo $\bistext$, $n$ be larger than the size of each term in the equations in $\E$, and $p,q$ be closed terms such that $p,q \bis f(\alpha,p_n)$.
If $\E \vdash p \approx q$ and $p$ has a summand bisimilar to $f(\alpha,p_n)$, then so does $q$.
\end{theorem}

This proves Theorem~\ref{Thm:nonfin-en} in the considered setting, as the left-hand side of equation $e_n$, viz.~the term $f(\alpha,p_n)$, has a summand bisimilar to $f(\alpha,p_n)$, whilst the right-hand side, viz.~the term $\alpha p_n + \sum_{i=0}^{n} \bar{\alpha}(\alpha \| \bar{\alpha}^{\scriptstyle \le i}) + \sum_{i = 0}^n \tau \alpha^{\scriptsize \le i}$, does not.

Before proceeding to the proof, we remark that the processes $f(\alpha,p_n)$ enjoy the following properties, according to the current set of allowed rules for operator $f$:

\begin{lemma}
\label{lem:f_vs_par_aba}
For each $n \ge 0$ it holds that $f(\alpha,p_n) \bis \alpha \| p_n$.
\end{lemma}

\begin{lemma}
\label{lem:f_vs_f_aba}
Let $n \ge 1$.
Assume that $f(p,q) \bis f(\alpha,p_n)$ for $p,q \nbis \nil$.
Then $p \bis \alpha$ and $q \bis p_n$.
\end{lemma}

\begin{proof}
The proof is analogous to that of Lemma~\ref{Lem:decomposing} and therefore omitted.
\end{proof}


\subsubsection{Proving Theorem~\ref{thm:LaRbaC}}

The crucial point in the proof of the negative result is (also in this case) the preservation of the witness property when instantiating an equation from a finite, sound axiom system.
We expand this case in the following proposition:

\begin{proposition}
\label{prop:LaRbaC_substitution}
Assume an operator $f$ such that only $L^{\ff}_{\alpha}$ holds for $\alpha$, only $R^{\ff}_{\bar{\alpha}}$ holds for $\bar{\alpha}$, and $S_{\alpha,\bar{\alpha}}$ holds.

Let $t \approx u$ be an equation over $\FCCS^-$ that is sound modulo $\bistext$.
Let $\sigma$ be a closed substitution with $p = \sigma(t)$ and $q = \sigma(u)$.
Suppose that $p$ and $q$ have neither $\nil$ summands nor factors, and $p,q \bis f(\alpha,p_n)$ for some $n$ larger than the size of $t$.
If $p$ has a summand bisimilar to $f(\alpha,p_n)$, then so does $q$.
\end{proposition}

\begin{proof}
First of all we notice that since $\sigma(t)$ and $\sigma(u)$ have no $\nil$ summands or factors, then neither do $t$ and $u$.
Therefore by Remark~\ref{rmk:summands} we get that $t = \sum_{i \in I} t_i$ and $u = \sum_{j \in J} u_j$, for some finite non-empty index sets $I,J$ with all the $t_i$ and $u_j$ not having $+$ as head operator, $\nil$ summands nor factors.
By the hypothesis, there is some $i \in I$ with $\sigma(t_i) \bis f(\alpha,p_n)$.
We proceed by a case analysis on the structure of $t_i$ to show that there is a $u_j$ such that $\sigma(u_j) \bis f(\alpha,p_n)$, establishing our claim.
\begin{enumerate}
\item \textcolor{blue}{\sc Case $t_i = x$ for some variable $x$ such that $\sigma(x) \bis f(\alpha,p_n)$.}
By Proposition~\ref{prop:x_in_tu}, $t$ having a summand $x$ implies that $u$ has a summand $x$ as well.
Thus, we can immediately conclude that $\sigma(u)$ has a summand bisimilar to $f(\alpha,p_n)$ as required.

\item \textcolor{blue}{\sc Case $t_i = \mu.t'$ for some term $t'$.}
This case is vacuous, as it contradicts $\sigma(t_i) \bis f(\alpha,p_n)$. 

\item \textcolor{blue}{\sc Case $t_i = f(t',t'')$ for some terms $t',t''$.}
Since $\sigma(t)$ has no $\nil$ factors, we have $\sigma(t'),\sigma(t'') \nbis \nil$.
Hence, from $f(\sigma(t'),\sigma(t'')) \bis f(\alpha,p_n)$ and Lemma~\ref{lem:f_vs_f_aba} we obtain $\sigma(t') \bis \alpha$ and $\sigma(t'') \bis p_n$.
By Remark~\ref{rmk:summands} we infer that $t'' = \sum_{h \in H} v_h$ for some terms $v_h$ that do not have $+$ as head operator and have no $\nil$-summands or factors.
Since $n$ is larger that the size of $t$, we have that $|H| < n$ and thus there is some $h \in H$ such that $\sigma(v_h) \bis \sum_{k = 1}^m \bar{\alpha}\alpha^{\scriptsize \le i_k}$ for some $m > 1$ and $1 \le i_1 < \dots < i_m \le n$.
Since $\sigma(v_h)$ has no $\nil$ summands or factors, from Lemma~\ref{Lem:vh-claim_gen} we infer that $v_h$ can only be a variable $x$ with
\begin{equation}
\label{eq:sigma_x_aba}
\sigma(x) \bis \sum_{k = 1}^m \bar{\alpha}\alpha^{\scriptsize \le i_k}.
\end{equation}
Therefore, $t_i = f(t', x+t''')$ for some $t'''$ such that $\sigma(x + t''') \bis p_n$.
We also notice that since $\sigma(t') \bis \alpha$ and $\init{\sigma(x)} =\{\bar{\alpha}\}$, we can infer that $x \trt{\rr}{\bar{\alpha}} t'$ does not hold (otherwise, $\sigma'(t)$ would afford an initial $\bar{\alpha}$-transition and would not be bisimilar to $\alpha$). 

To prove that $u$ has a summand bisimilar to $f(\alpha,p_n)$, consider the closed substitution
\[
\sigma' = \sigma[x \mapsto \bar{\alpha} p_n].
\]
Notice that, since $\sigma(t') \bis \alpha$, $\sigma(t')$ has no $\nil$ summands or factors, $\init{\sigma(x)} = \init{\sigma'(x)} = \{\bar{\alpha}\}$ and $x$ is the only variable which is affected when changing $\sigma$ into $\sigma'$, then we can infer that either $x \not \in \var(t')$ or its execution is always prevented.
In both cases we get $\sigma(t') \bis \sigma'(t') \bis \alpha$.
Then, using Lemma~\ref{lem:f_vs_par_aba} and $t_i = f(t', x+t'')$, we have
\[
\sigma'(t_i) \trans[\bar{\alpha}] p' \bis \alpha \| p_n \bis f(\alpha,p_n).
\]
As $t \approx u$ implies $\sigma'(t) \bis \sigma'(u)$, we infer that there must be a summand $u_j$ such that $\sigma'(u_j) \trans[\bar{\alpha}] r$ for some $r \bis f(\alpha,p_n)$.
Notice that, since $\sigma(u) \bis f(\alpha,p_n)$ and $\sigma(u_j) = \sigma'(u_j)$ if $x \not \in \var(u_j)$, then it must be the case that $x \in \var(u_j)$, or otherwise we get a contradiction with $\sigma(u) \bis f(\alpha,p_n)$.
By Lemma~\ref{lem:c2o}, as only $R^{\ff}_{\bar{\alpha}}$ holds, we can distinguish two cases:

\begin{enumerate}
\item \textcolor{red}{There is a term $u'$ s.t.\ $u_j \trans[\bar{\alpha}] u'$ and $\sigma'(u') \bis f(\alpha,p_n)$.}
Then, since $f(\alpha,p_n) \bis \alpha \parallel p_n$ (Lemma~\ref{lem:f_vs_par_aba}) we can apply the expansion law, obtaining $\sigma'(u') \bis \alpha p_n +  \sum_{i = 1}^n \bar{\alpha}(\alpha \| \alpha^{\scriptsize \le i}) + \sum_{i = 1}^n \tau \alpha^{\scriptsize \le i}$.
As $n$ is greater than the size of $u$, and thus of those of $u_j$ and $u'$, by Lemma~\ref{lem:variables} we get that $u'$ has a summand $y$, for some variable $y$, such that $\sigma'(y) \bis \sum_{k = 1}^{m'} \bar{\alpha} q_{i'_k} + r'$, for some $m' > 1$, $1 \le i'_1 < \dots < i'_{m'} \le n$, closed term $r'$, and closed terms $q_{i'_k}$ such that either $q_{i'_k} \bis \alpha \| \alpha^{\scriptsize \le i'_k}$ or $q_{i'_k} \bis \alpha^{\scriptsize \le i'_k}$, for each $k=1,\dots,m'$.
(We can infer the exact form of the $q_{i'_k}$ since $\alpha$ and $\alpha^{\scriptsize \le i'_k}$ are prime, the parallel component $\alpha$ is common to all summands and $\alpha^{\scriptsize \le i'_k} \nbis \alpha^{\scriptsize \le i'_j}$ if $k \neq j$).
In both cases, we can infer that $y \neq x$, as $\sigma'(x) \nbis \sigma'(y)$ for any closed term $r'$.
Thus we have $\sigma'(y) = \sigma(y)$ and we get a contradiction with $\sigma(u) \bis f(\alpha,p_n)$ in that $\sigma(u_j)$ would be able to perform two $\bar{\alpha}$-moves in a row unlike $f(\alpha,p_n)$.

\item \textcolor{red}{There are a variable $y$, a closed term $r'$ and a configuration $c$ s.t.\ $\sigma'(y) \trans[\bar{\alpha}] r'$, $u_j \trans[y_{\rr}]_{\bar{\alpha}} c$ and $\sigma'[y_{\dd} \mapsto r'](c) \bis f(\alpha,p_n)$.}
We claim that it must be the case that $y =x$.
To see this claim, assume towards a contradiction that $y \neq x$.
We proceed by a case analysis on the possible occurrences of $x$ in $c$.
\begin{itemize}
\item \textcolor{Purple4}{$x \not \in \var(c)$ or $x \in \var(c)$ but its occurrence is in a guarded context that prevents the execution of its closed instances.}
In this case we get $\sigma[y_{\dd} \mapsto r'](c) \bis \sigma'[y_{\dd}\mapsto r'](c) \bis f(\alpha,p_n)$.
This contradicts $\sigma(u) \bis f(\alpha,p_n)$ since we would have $\sigma(u) \trans[\bar{\alpha}] r \bis f(\alpha,p_n)$, and such a transitions cannot be mimicked by $f(\alpha,p_n)$.
\item \textcolor{Purple4}{$x \in\var(c)$ and its execution is not prevented.}
We can distinguish two sub-cases, according to whether the occurrence of $x$ is guarded or not.
\begin{itemize}
\item \textcolor{DarkGoldenrod4}{Assume that $x$ occurs guarded in $c$.}
In this case we get a contradiction with $r \bis f(\alpha,p_n)$ since $x$ being guarded implies:
\[
n+2 = \depth(f(\alpha,p_n)) 
= \depth(r) 
\ge  1 + \depth(\sigma'(x))
= n+3. 
\]
\item \textcolor{DarkGoldenrod4}{Assume now that $x \trt{\bb}{\alpha} c$.}
This case contradicts our assumption that $\sigma(u) \bis f(\alpha,p_n)$ since we would have $\sigma(u) \trans[\bar{\alpha}] \sigma[y_{\dd} \mapsto r'](c) \trans[\bar{\alpha}]$, due to Lemmas~\ref{lem:trt_open} and~\ref{lem:o2c}, whereas $f(\alpha,p_n)$ cannot perform two $\bar{\alpha}$-moves in a row.
\end{itemize}
\end{itemize}
Therefore, we can conclude that it must be the case that $y=x$ and $r' = p_n$.
In particular, notice that $x \trt{\rr}{\bar{\alpha}} u_j$.
We now proceed by a case analysis on the structure of $u_j$ to show that $\sigma(u_j) \bis f(\alpha,p_n)$.
\begin{enumerate}
\item \textcolor{ForestGreen}{$u_j = x$.} 
This case is vacuous, as $\sigma'(x) \trans[\bar{\alpha}] p_n$ and $p_n \nbis f(\alpha,p_n)$.
\item \textcolor{ForestGreen}{$u_j = f(u',u'')$ for some $u',u''$.}
Notice that $x \trt{\rr}{\bar{\alpha}} u_j$ can be due only to $x \trt{\rr}{\bar{\alpha}} u''$.
We have $\sigma'(u'') \trans[\bar{\alpha}] r_1$ and $\sigma'(u_j) \trans[\bar{\alpha}] \sigma'(u') \| r_1 \bis f(\alpha,p_n)$.
Since $f(\alpha,p_n) \bis \alpha \| p_n$ and both $\alpha$ and $p_n$ are prime, by Proposition~\ref{prop:unique_decomposition}, we distinguish two cases:
\begin{itemize}
\item \textcolor{purple}{Case $\sigma'(u') \bis \alpha$ and $r_1 \bis p_n$.}
As $\init{\sigma(x)} = \init{\sigma'(x)} = \{\bar{\alpha}\}$, $R^{\ff}_{\bar{\alpha}}$, $\sigma'(u') \bis \alpha$ and $\sigma(u)$ has no $\nil$ factors, we get that either $x \not\in \var(u')$ or $x$ occurs in $u'$ but its execution is prevented by the rules for $f$.
Therefore $\sigma'(u') \bis \sigma(u') \bis \alpha$.
We aim at showing that $u''$ has a summand $x$.
We proceed by proving that the only other possibility, namely $u'' = f(w_1,w_2)$ for some $w_1,w_2$ with $x \trt{\rr}{\bar{\alpha}} w_2$, leads to a contradiction.

As $\sigma'(u'') \trans[\bar{\alpha}] r_1 \bis p_n$, we have $\sigma'(w_2) \trans[\bar{\alpha}] r_2$ and $\sigma'(w_1) \| r_2 \bis p_n$.
Since, $p_n$ is prime, we have that either $\sigma'(w_1) \bis \nil$ and $r_2 \bis p_n$, or $\sigma'(w_1) \bis p_n$ and $r_2 \bis \nil$.
In both cases, as $\sigma'(x) \nbis \sigma'(w_1)$ and the previous considerations, we infer $\sigma(w_1) \bis \sigma'(w_1)$.
Hence, the former case contradicts $\sigma(u)$ not having $\nil$ factors.
The latter case contradicts $\sigma(u) \bis f(\alpha,p_n)$ as, considering that $x\trt{\rr}{\bar{\alpha}} w_2$, the transition $\sigma'(w_2) \trans[\bar{\alpha}] r_2 \bis \nil$ cannot be due to $\sigma'(x)$ and therefore it would be available also to $\sigma(w_2)$ thus implying $\sigma(u_j) \trans[\bar{\alpha}] r''$ with $r'' \bis f(\alpha,p_n)$.

Summing up, we have argued that $u''$ has a summand $x$.
Therefore, by Equation~\eqref{eq:sigma_x_aba},
\[
\sigma(u'') \bis \sum_{k = 1}^m\bar{\alpha}.\alpha^{\scriptstyle \leq i_k} + r'' \enspace ,
\]
for some closed term $r''$. 
We have already noted that
\[
\sigma(u') \bis \sigma'(u')  \bis \alpha \enspace .
\]
Thus, using the congruence properties of bisimilarity, we may infer that
\[
\sigma(u_j) = 
f(\sigma(u'),\sigma(u'')) \bis 
f(\alpha,\sum_{k = 1}^m \bar{\alpha}\alpha^{\scriptsize \le i_k} + r'') 
\enspace .
\]
In light of this equivalence, we have $\sigma(u_j) \trans[\alpha] r' \bis \sigma(u'')$ and thus $\sigma(u) \trans[\alpha] r'$.
Since, by hypothesis, $\sigma(u) \bis f(\alpha,p_n)$ then it must be the case that $r' \bis p_n$.
Therefore, we can conclude that $\sigma(u_j) \bis f(\alpha,p_n)$.
Hence, $\sigma(u)$ has the desired summand.

\item \textcolor{purple}{Case $\sigma'(u') \bis p_n$ and $r_1 \bis \alpha$.}
By reasoning as above, we can infer that either $x \not \in \var(u')$ or it is blocked by the rules for $f$, so that 
\[
\sigma'(u') \bis \sigma(u') \bis p_n.
\]
However, $\depth(\sigma(x)) \ge 3$, and $x \trt{\rr}{\bar{\alpha}} u''$ with $\init{\sigma(x)} = \{\bar{\alpha}\}$ give us, by Lemma~\ref{lem:that_one}, that $\depth(\sigma(u'')) \ge \depth(\sigma(x))$.
Therefore we get a contradiction, in that
\begin{align*}
n+2 ={} & \depth(f(\alpha,p_n))
= \depth(\sigma(u)) 
\ge \depth(\sigma(u_j)) 
= \depth(f(\sigma(u'),\sigma(u''))) \\
\ge{} & \depth(\sigma(u')) + \depth(\sigma(u'')) 
\ge \depth(\sigma(u')) + \depth(\sigma(x)) 
\ge n+1 + 3.
\end{align*}
\end{itemize}
\end{enumerate}
\end{enumerate}
\end{enumerate}
The proof of Proposition~\ref{prop:LaRbaC_substitution} is now complete.
\end{proof}

We can now formalise the proof of Theorem~\ref{thm:LaRbaC}.\\

\begin{apx-proof}{Theorem~\ref{thm:LaRbaC}}
The proof follows the same lines of that of Theorem~\ref{thm:LaRa}.
The only difference (besides the use of Proposition~\ref{prop:LaRbaC_substitution} in place of Proposition~\ref{prop:LaRa_substitution} in the case of substitutions) is the following inductive step:
\begin{enumerate}
\setcounter{enumi}{3}
\item \textcolor{blue}{\sc Case $\E \vdash p \approx q$ because $p = f(p_1,p_2)$ and $q = f(q_1,q_2)$ for some $p_i,q_i$ with $E \vdash p_i \approx q_i$, for $i \in \{1,2\}$.}
By the proviso of the theorem $p,q$ have neither $\nil$ summands nor factors, thus implying $p_i,q_i \nbis \nil$.
Hence, from $p \bis f(\alpha,p_n)$ and $p = f(p_1,p_2)$ and Lemma~\ref{lem:f_vs_f_aba} we obtain $p_1 \bis \alpha$ and $p_2 \bis p_n$, thus implying, by the soundness of the equations in $\E$, that $q_1 \bis \alpha$ and $q_2 \bis p_n$, so that $q = f(\alpha,p_n)$.
In both cases, we can infer that $q$ has itself as the desired summand.
\end{enumerate}
This completes the proof of Theorem~\ref{thm:LaRbaC} and thus of Theorem~\ref{Thm:nonfin-en} in the case of an operator $f$ that does not distribute over summation in either argument, case $L^{\ff}_{\alpha}, R^{\ff}_{\bar{\alpha}}, S_{\alpha,\bar{\alpha}}^f$.
\end{apx-proof}


\subsection{Case 2: No synchronisation}

Assume now that the synchronisation between $\alpha$ and $p_n$ is prevented, namely only $S_{\bar{\alpha},\alpha}^f$ holds.
Then, the witness family of equations changes as follows:
\[
e_n \colon \quad f(\alpha,p_n) \approx \alpha p_n + \sum_{i = 0}^n \bar{\alpha}(\alpha \| \alpha^{\scriptsize \le i}) \qquad (n \ge 0) \enspace .
\]
Our order of business is then to prove the following:

\begin{theorem}
\label{thm:LaRba}
Assume an operator $f$ such that only $L^{\ff}_{\alpha}$ holds for $\alpha$, only $R^{\ff}_{\bar{\alpha}}$ holds for $\bar{\alpha}$, and only $S_{\bar{\alpha},\alpha}^f$ holds.
Let $\E$ be a finite axiom system over $\FCCS$ that is sound modulo $\bistext$, $n$ be larger than the size of each term in the equations in $\E$, and $p,q$ be closed terms such that $p,q \bis f(\alpha,p_n)$.
If $\E \vdash p \approx q$ and $p$ has a summand bisimilar to $f(\alpha,p_n)$, then so does $q$.
\end{theorem}

Once again, the validity of Theorem~\ref{Thm:nonfin-en} follows by noticing that the left-hand side of equation $e_n$, viz.~the term $f(\alpha,p_n)$, has a summand bisimilar to $f(\alpha,p_n)$, whilst the right-hand side, viz.~the term $\alpha p_n + \sum_{i=0}^{n} \bar{\alpha}(\alpha \| \bar{\alpha}^{\scriptstyle \le i})$, does not.

\subsubsection{Proving Theorem~\ref{thm:LaRba}}

The proof of Theorem~\ref{thm:LaRba} follows that of Theorem~\ref{thm:LaRbaC} in a step by step manner, by exploiting Proposition~\ref{prop:LaRba_substitution} below in place of Proposition~\ref{prop:LaRbaC_substitution}.
The only dif{f}erence with the proof of Proposition~\ref{prop:LaRbaC_substitution} is that, in the case at hand, Lemma~\ref{lem:f_vs_par_aba} does not hold anymore.
(In fact one could prove, as done for Lemma~\ref{lem:rhs-prime}, that $f(\alpha,p_n)$ is prime for all $n \ge 0$.)

\begin{proposition}
\label{prop:LaRba_substitution}
Assume an operator $f$ such that only $L^{\ff}_{\alpha}$ holds for $\alpha$, only $R^{\ff}_{\bar{\alpha}}$ holds for $\bar{\alpha}$, and only $S_{\bar{\alpha},\alpha}$ holds.

Let $t \approx u$ be an equation over $\FCCS^-$ that is sound modulo $\bistext$.
Let $\sigma$ be a closed substitution with $p = \sigma(t)$ and $q = \sigma(u)$.
Suppose that $p$ and $q$ have neither $\nil$ summands nor factors, and $p,q \bis f(\alpha,p_n)$ for some $n$ larger than the size of $t$.
If $p$ has a summand bisimilar to $f(\alpha,p_n)$, then so does $q$.
\end{proposition}

\begin{proof}
The proof follows exactly as the proof of Proposition~\ref{prop:LaRbaC_substitution}, with the only difference that when we consider the derived transition
\[
\sigma'(t_1) \trans[\bar{\alpha}] p'
\]
we have that $p' \bis \alpha \| p_n \nbis f(\alpha,p_n)$.
However, by substituting $f(\alpha,p_n)$ with $\alpha \| p_n$ in the remaining of the proof, the same arguments hold.
\end{proof}


\section{Negative result: the case $L^{\ff}_{\tau}$}
\label{sec:Lt}

This section considers the last case in our analysis, namely that of an operator $f$ that does not distribute, modulo bisimilarity, over summation in either argument and that has the same rule type for actions $\alpha,\bar{\alpha}$.
Here, we present solely the case in which $L^{\ff}_{\tau}$ holds, and only $R^{\ff}_{\alpha},R^{\ff}_{\bar{\alpha}}$ hold for $\alpha,\bar{\alpha}$, namely $f$ has rules:
\[
\SOSrule{x_1 \trans[\tau] y_1}{f(x_1,x_2) \trans[\tau] y_1 \| x_2}
\qquad
\SOSrule{x_2 \trans[\alpha] y_2}{f(x_1,x_2) \trans[\alpha] x_1 \| y_2}
\qquad
\SOSrule{x_2 \trans[\bar{\alpha}] y_2}{f(x_1,x_2) \trans[\bar{\alpha}] x_1 \| y_2}.
\]
The symmetric case can be obtained from this one in a straightforward manner.

Interestingly, the validity of the negative result we consider in this section is independent of which rules of type (\ref{syncrule}) are available for $f$, and of the validity of the predicate $R^{\ff}_{\tau}$.

Consider the family of equations defined by:
\[
e_n \colon \quad f(\tau,q_n) \approx \tau q_n + \sum_{i = 0}^n \alpha(\tau \| \bar{\alpha}^{\scriptsize \le i}) \qquad (n \ge 0)
\]
where the processes $q_n$ are the same used in Section~\ref{sec:LaRa}.
Theorem~\ref{thm:Lt_1} below proves that the collection of equations $e_n$, $n \ge 0$, is a witness family of equations for our negative result.

\begin{theorem}
\label{thm:Lt_1}
Assume an operator $f$ such that $L^{\ff}_{\tau}$ holds and only $R^{\ff}_{\alpha}$ and $R^{\ff}_{\bar{\alpha}}$ hold for actions $\alpha$ and $\bar{\alpha}$.
Let $\E$ be a finite axiom system over $\FCCS$ that is sound modulo $\bistext$, $n$ be larger than the size of each term in the equations in $\E$, and $p,q$ be closed terms such that $p,q \bis f(\tau,q_n)$.
If $\E \vdash p \approx q$ and $p$ has a summand bisimilar to $f(\tau,q_n)$, then so does $q$.
\end{theorem}
 
As the left-hand side of equation $e_n$, viz.~the term $f(\tau,q_n)$, has a summand bisimilar to $f(\tau,q_n)$, whilst the right-hand side, viz.~the term $\tau q_n + \sum_{i=0}^{n} \alpha(\tau \| \bar{\alpha}^{\scriptstyle \le i})$, does not, we can conclude that the collection of infinitely many equations $e_n$ ($n \ge 0$) is the desired witness family.
This concludes the proof of Theorem~\ref{Thm:nonfin-en} for this case and our proof of Theorem~\ref{Thm:f(x,y)+f(y,x)}.


\subsection{Case specific properties of $f(\alpha,q_n)$}

First of all, we remark that the witness processes $f(\tau,q_n)$ enjoy the properties formalised in Lemmas~\ref{lem:f_vs_par_t} and~\ref{lem:f_vs_f_t} below.

\begin{lemma}
\label{lem:f_vs_par_t}
For each $n \ge 0$ it holds that $f(\tau,q_n) \bis \tau \| q_n$.
\end{lemma}

\begin{lemma}
\label{lem:f_vs_f_t}
Let $n \ge 1$.
Assume that $f(p,q) \bis f(\tau,q_n)$ for $p,q \nbis \nil$.
Then $p \bis \tau$ and $q \bis q_n$.
\end{lemma}

\begin{proof}
The proof is analogous to that of Lemma~\ref{Lem:decomposing}.
We remark that the $\tau$-transition by $f(\tau,q_n)$ can be mimicked only by a $\tau$-move by $p$.
To see this, we show that any other case would lead to a contradiction with the proviso of the lemma $f(p,q) \bis f(\tau,q_n)$.
In particular, we distinguish three cases, according to which rule of type (\ref{syncrule}) is available for $f$ and whether the predicates $R^{\ff}_{\tau}$ holds or not.
\begin{itemize}
\item Assume $p \trans[\alpha] p'$ and $q \trans[\bar{\alpha}] q'$ with $p' \| q' \bis q_n$.
This would contradict $f(\tau,q_n) \bis f(p,q)$ since $f(p,q) \trans[\bar{\alpha}] p \| q'$, whereas $f(\tau,q_n) \ntrans[\bar{\alpha}]$.
\item Assume $p \trans[\bar{\alpha}] p'$ and $q \trans[\alpha] q'$ with $p' \| q' \bis q_n$.
Notice that since $q_n$ is prime, then we have that either $p' \bis \nil$ and $q'\bis q_n$, or $p' \bis q_n$ and $q' \bis \nil$.
The latter case contradicts $f(p,q) \bis f(\tau,q_n)$ since the transition $f(p,q) \trans[\alpha] p \| q' \bis p \| q_n$ cannot be mimicked by $f(\tau,q_n)$.
The former case also contradicts the proviso of the lemma, since we would have $f(p,q) \trans[\alpha] p \| q' \bis p \trans[\bar{\alpha}] p' \bis q_n$, whereas $f(\tau,q_n) \trans[\alpha] \tau \| \bar{\alpha}^{\scriptsize \le i}$, for some $i \in \{1,\dots,n\}$, and there is no $r$ such that $\tau \| \bar{\alpha}^{\scriptsize \le i} \trans[\bar{\alpha}] r$ and $r \bis q_n$, for any $i \in \{1,\dots,n\}$.
\item Finally, assume that the predicate $R^{\ff}_{\mu}$ holds, and thus that $f$ has a rule of type (\ref{asyncruleright}) with label $\tau$.
Hence, assume $q \trans[\tau] q'$, for some $q'$, so that $f(p,q) \trans[\tau] p \| q' \bis q_n$.
Since $q_n$ is prime and $p \nbis \nil$, we have that $p \bis q_n$ and $q' \bis \nil$.
So, by congruence closure, we get 
\[
f(p,q) \bis f(q_n,q) \bis f(\tau,q_n).
\]
Since $f(\tau,q_n) \trans[\alpha] \tau \| \bar{\alpha}^{\scriptsize \le n}$ and only $R^{\ff}_{\alpha}$ holds, we have that $q \trans[\alpha] q_1$ for some $q_1$ such that $q_n \| q_1 \bis \tau \| \bar{\alpha}^{\scriptsize \le n}$, which is a contradiction as $q_n \trans[\alpha]$ implies $q_n \| q_1 \trans[\alpha]$, whereas $\tau \|\bar{\alpha}^{\scriptsize \le n} \ntrans[\alpha]$.
\end{itemize}
\end{proof}


\subsection{Proving Theorem~\ref{thm:Lt_1}}

The same reasoning used in the proof of Theorem~\ref{thm:LaRbaC} allows us to prove Theorem~\ref{thm:Lt_1}, by exploiting Proposition~\ref{prop:Lt_substitution_1} in place of Proposition~\ref{prop:LaRbaC_substitution}.

\begin{proposition}
\label{prop:Lt_substitution_1}
Assume an operator $f$ such that only $R^{\ff}_{\alpha}$ and $R^{\ff}_{\bar{\alpha}}$ hold for $\alpha,\bar{\alpha}$, and $L^{\ff}_{\tau}$ holds.

Let $t \approx u$ be an equation over $\FCCS^-$ that is sound modulo $\bistext$.
Let $\sigma$ be a closed substitution with $p = \sigma(t)$ and $q = \sigma(u)$.
Suppose that $p$ and $q$ have neither $\nil$ summands nor factors, and $p,q \bis f(\tau,q_n)$ for some $n$ larger than the size of $t$.
If $p$ has a summand bisimilar to $f(\tau,q_n)$, then so does $q$.
\end{proposition}

\begin{proof}
The claim follows by the same arguments used in the proof of Proposition~\ref{prop:LaRbaC_substitution} and by considering the substitution 
\[
\sigma'=\sigma[x \mapsto \alpha q_n].
\]
\end{proof}


\section{Conclusions}
\label{sec:conclusion}

In this paper, we have shown that, under a number of simplifying assumptions, we cannot use a single binary auxiliary operator $f$, whose semantics is defined via inference rules in the de Simone format, to obtain a finite axiomatisation of bisimilarity over the recursion-, restriction-, and relabelling-free fragment of CCS. 
Our result constitutes a first step towards a definitive justification of the canonical standing of the left and communication merge operators by Bergstra and Klop. 

We envisage the following ways in which we might generalise the contribution presented in this study.  
Firstly, we will try to relax Assumption~\ref{Ass:deSimone} by considering the GSOS format~\cite{BloomIM1995} in place of the de Simone format. 
However, as shown by the heavy amount of technical results necessary to prove our main result even in our simplified setting, we believe that this generalisation cannot be obtained in a straightforward manner and that it will require the introduction of new techniques.
It would also be very interesting to explore whether some version of problem (\ref{eq:problem}) can be solved using existing results from equational logic and universal algebra. 

Recentily, in~\cite{AACIL21} the negative result by Moller on the non-finite axiomatisability of bisimilarity~\cite{Mo90a} has been extended to a family of weak congruences.
Formally, it has been proved that all congruences that coincide with strong bisimilarity on processes without silent moves, impose the root condition on initial silent moves, and satisfy a particular family of equations introduced by Moller in~\cite{Mo90a}, \emph{have no finite, complete axiomatisation over (recursion-, restriction-, and relabelling-free) CCS}.
These include \emph{rooted weak bisimilarity} (also known as \emph{obsevational congruence}~\cite{HM85}), \emph{rooted branching bisimilarity}, \emph{rooted delay bisimilarity}, and \emph{rooted $\eta$-bisimilarity}~\cite{vG93}.
It is still an open question whether the use of auxiliary operators can be of help to obtain a finite axiomatisability result, as in the case of strong bisimilarity.
Hence, a generalisation of our results to weak semantics would help to solve this problem.

Yet, this generalisation and other possible ones of our work, e.g. to other process algebras and/or to other semantics, are related to some general open questions in equational logic:
\begin{equation*}
\parbox{\dimexpr\linewidth-4em}{
\strut
\emph{Are there general techniques for lifting negative results across process algebras? And from strong to weak congruences? And from qualitative to quantitative semantics?}
\strut
}
\end{equation*}
Understanding whether it is possible to lift non-finite axiomatisability results among different algebras and semantics, and under which constraints this can be done, is an interesting research avenue and we aim to investigate it in future work. 
A methodology for transferring non-finite-axiomatisability results across languages was presented in~\cite{AFIM10}, where a reduction-based approach was proposed. 
However, that method has some limitations and thus further studies are needed.
Similarly, the ever increasing interest in probabilistic systems has inspired a number of studies on the axiomatisation of probabilistic congruences (see~\cite{De05} for a survey).
We can find studies on strong probabilistic semantics~\cite{AEI02,BBS95,vGSST90,GV19,LS92,Mi20,SS00,TG20}, weak probabilistic semantics~\cite{ABW06,AG09,vGGV19}, as well as on metric semantics~\cite{dAGL14}.
Further studies in this direction are encouraged by recent achievements on probabilistic branching semantics~\cite{CT20a,CT20b} and behavioural metrics~\cite{CLT19,CLT20}.

\subsection*{Acknowledgements}
This work has been supported by the project `\emph{Open Problems in the Equational Logic of Processes}' (OPEL) of the Icelandic Research Fund (grant No.~196050-051).

\bibliographystyle{splncs04}
\bibliography{nobinary_tocl}

\end{document}